\documentclass[manyauthors]{fundam}

\usepackage{latexsym,amssymb,amsmath}
\usepackage{url}
\usepackage[defblank]{paralist}
\usepackage{tikz}
\usetikzlibrary{arrows,shapes,shapes.multipart,snakes,automata,backgrounds,petri,positioning,shadows,matrix,decorations.pathmorphing,fit,positioning,calc,backgrounds}
\usepackage{todonotes}
\usepackage{xspace}

\usepackage{xcolor}

\usepackage{subcaption}
\usepackage[multiple]{footmisc}

\usepackage{hyperref}

\usepackage{todonotes}

\newcommand{\A}{\ensuremath{\mathcal{A}}}
\newcommand{\B}{\ensuremath{\mathcal{B}}}

\newcommand{\I}{\ensuremath{\mathcal{I}}}

\renewcommand{\L}{\ensuremath{\mathcal{L}}}

\newcommand{\V}{\ensuremath{\mathcal{V}}}
\newcommand{\W}{\ensuremath{\mathcal{W}}}

\newcommand{\set}[1]{\{#1\}}                      

\newcommand{\tup}[1]{\langle #1\rangle}            

\newcommand{\rng}[1]{\textsc{rng}(#1)}

\newcommand{\mult}[1]{#1^\oplus}
\newcommand{\supp}[1]{\mathit{supp}(#1)}

\newcommand{\funsym}[1]{\mathtt{#1}}
\newcommand{\setsym}[1]{\mathit{#1}}
\newcommand{\restr}[2]{\left.{#1}\right|_{#2}}
\newcommand{\powerset}[1]{\raisebox{.15\baselineskip}{\Large\ensuremath{\wp}}(#1)}

\newcommand{\cname}[1]{\ensuremath{\mathtt{#1}}\xspace} 

\newcommand{\type}{\ensuremath {type}}
\newcommand{\id}[1]{\ensuremath \mathit{Id}(#1)}
\newcommand{\emptysequence}{\ensuremath \epsilon}
\newcommand{\sequence}[1]{\ensuremath \langle #1 \rangle}
\newcommand{\concat}{\cdot }
\newcommand{\proj}[2]{\ensuremath {#1}_{{\mid {#2}}}}

\newcommand{\tsys}[1]{\Gamma_{#1}}
\newcommand{\states}{S}

\newcommand{\hide}[1]{\funsym{\hat H}_{#1}}

\newcommand{\colset}{\funsym{C}}

\newcommand{\naturals}{\mathbb{N}}

\newcommand{\tpnid}{t-PNID\xspace}
\newcommand{\tpnids}{t-PNIDs\xspace}

\newcommand{\p}{\mathrm{p}\xspace}

\newcommand{\places}{P}
\newcommand{\transitions}{T}
\newcommand{\flow}{F}
\newcommand{\enabled}[2]{#1[#2\rangle}
\newcommand{\fire}[3]{\enabled{#1}{#2}#3}

\newcommand{\reachable}[2]{\mathcal{R}(#1,#2)}

\newcommand{\marking}{\mathit{M}}

\newcommand{\pre}[1]{{{^\bullet{#1}}}}
\newcommand{\post}[1]{{{#1^\bullet}}}

\newcommand{\invar}[1]{{\setsym{In}({#1})}}
\newcommand{\outvar}[1]{{\setsym{Out}({#1})}}
\newcommand{\var}[1]{{\setsym{Var}({#1})}}
\newcommand{\newvar}[1]{{\setsym{Emit}({#1})}}
\newcommand{\delvar}[1]{{\setsym{Collect}({#1})}}

\newcommand{\inp}{\mathit{in}}
\newcommand{\outp}{\mathit{out}}

\definecolor{golden}{rgb}{1.0, 0.84, 0.0}
\definecolor{indigo}{rgb}{0.0, 0.25, 0.42}
\definecolor{mgreen}{rgb}{0.128,0.428,0}
\definecolor{burntorange}{rgb}{0.8, 0.33, 0.0}
\definecolor{camouflagegreen}{rgb}{0.47, 0.53, 0.42}
\definecolor{copperrose}{rgb}{0.6, 0.4, 0.4}

\tikzstyle{placelem}=[draw,
    cloud,
    cloud puffs = 15,
    minimum width=.2cm,
    minimum height=.2cm,
    fill=white,
    thick]

\tikzstyle{place}=[circle,thick,draw=black,fill=white,minimum size=7mm,font=\fontsize{9}{144}\selectfont]
\tikzstyle{transition}=[rectangle,thick,draw=black,fill=gray!20,minimum size=7mm]
\tikzstyle{enabledtransition}=[rectangle,very thick,draw=green!75,fill=green!20,minimum size=7mm]
 \tikzstyle{container}=[rectangle,rounded corners,very thick,draw=black!75,fill=black!20,minimum height=7mm,minimum width=14mm]
\tikzstyle{rplace}=[circle,ultra thick,draw=violet!75,fill=violet!20,minimum size=7mm]

\tikzstyle{erbox}=[draw, fill=gray!20, minimum width=7em, text width=6.0em, text centered,
  minimum height=3em,rounded corners,drop shadow]

\tikzstyle{placelem}=[draw,
    cloud,
    cloud puffs = 15,
    minimum width=.2cm,
    minimum height=.2cm,
    fill=white,
    thick]

\tikzstyle{netelem}=[draw,
    cloud,
    cloud puffs = 20,
    minimum width=1.8cm,
    minimum height=1.8cm,
    fill=blue!10,
    ultra thick]

\tikzstyle{relationselem}=[placelem,fill=pink!20]
\tikzstyle{noopelem}=[placelem,fill=orange!20]
\tikzstyle{enteredplace}=[place,fill=yellow!20]
\tikzstyle{boundplace}=[place,fill=yellow!30]
\tikzstyle{guardokplace}=[place,fill=yellow!40]
\tikzstyle{updatedplace}=[place,fill=yellow!50]
\tikzstyle{violplace}=[place,fill=red!10]
\tikzstyle{constrokplace}=[place,fill=green!10]
\tikzstyle{docommitplace}=[place,fill=green!30]
\tikzstyle{dorollbackplace}=[place,fill=red!30]
\tikzstyle{arc}=[-stealth',thick]
\tikzstyle{readarc}=[-,thick]

\tikzset{
state/.style={
       rectangle,
       fill=yellow!5,
       rounded corners,
       draw=black,  thick,
       minimum height=2em,
       minimum width=1.5cm,
       inner sep=1pt,
       text centered,
       }
}

\begin{document}


\setcounter{page}{159}
\publyear{24}
\papernumber{2169}
\volume{190}
\issue{2-4}

\finalVersionForARXIV

\title{Correctness Notions for Petri Nets with Identifiers}

\author{Jan Martijn E.M. van der Werf \thanks{Address for correspondence: Utrecht University,  Princetonplein 5, 3584 CC Utrecht,
                                  The Netherlands.}
  \\
Utrecht University\\
Princetonplein 5, 3584 CC Utrecht\\
The Netherlands\\
j.m.e.m.vanderwerf@uu.nl
\and Andrey Rivkin\\
Technical University of Denmark\\
Richard Petersens Plads 321\\
 2800 Kgs.,  Lyngby, Denmark\\
ariv@dtu.dk
\and Marco Montali\\
Free University of Bozen-Bolzano\\
piazza Domenicani 3\\
 39100 Bolzano,  Italy\\
montali@inf.unibz.it
\and Artem Polyvyanyy\\
The University of Melbourne\\
Grattan Street, Parkville \\
Victoria, 3010, Australia\\
artem.polyvyanyy@unimelb.edu.au
}

\maketitle

\runninghead{Van~der~Werf et al.}{Correctness Notions for Petri Nets with Identifiers}

\vspace*{-4mm}
\begin{abstract}
A model of an information system describes its processes and how resources are involved in these processes to manipulate data objects. This paper presents an extension to the Petri nets formalism suitable for describing information systems in which states refer to object instances of predefined types and resources are identified as instances of special object types. Several correctness criteria for resource- and object-aware information systems models are proposed, supplemented with discussions on their decidability for interesting classes of systems. These new correctness criteria can be seen as generalizations of the classical soundness property of workflow models concerned with process control flow correctness.\vspace*{-1mm}
\end{abstract}


\section{Introduction}

Petri nets are widely used to describe distributed systems capable of expanding their resources indefinitely~\cite{Reisig2013}.
A \emph{Petri net} describes passive and active components of a system, modeled as places and transitions, respectively.
The active components of a Petri net communicate asynchronously with each other via local interfaces.
Thus, state changes in a Petri net system have local causes and effects and are modeled as tokens consumed, produced, or transferred by the transitions of the system.
A \emph{token} is often used to denote an \emph{object} in the physical world the system manipulates or a \emph{condition} that can cause a state change in the system.

\medskip
Petri nets with identifiers~\cite{HSVW09} extend classical Petri nets to provide formal means to relate tokens to objects.
Every token in such a Petri net is associated with a vector of identifiers, where each identifier uniquely identifies a data object.
Consequently, active components of a Petri net with identifiers model how groups of objects, either envisioned or those existing in the physical world, can be consumed, produced, or transferred by the system.

It is often desirable that modeled systems are correct.
Many criteria have been devised for assessing the correctness of systems captured as Petri nets.
Those criteria target models of systems that use tokens to represent conditions that control their state changes.
In other words, they can be used to verify the correctness of processes the systems can support and not of the object manipulations carried out within those processes.
Such widely-used criteria include boundedness~\cite{Karp1969}, liveness~\cite{Hack1974}, and soundness~\cite{Aalst1997}.
The latter one, for instance, ensures that a system modeled as a \emph{workflow net}, a special type of Petri nets used to encode workflows at organizations, has a terminal state that can be distinguished from other states of the system, the system can always reach the terminal state, and every transition of the system can in principle be enabled and, thus, used by the system.

Real-world systems, such as information systems~\cite{PolyvyanyyWOB19}, are characterized by processes that manipulate objects.
For instance, an online retailer system manipulates products, invoices, and customer records.
However, although tools allow designing such models~\cite{WerfP20}, initial use showed that correctness criteria addressing both aspects, that is, the processes and data, are understood less well~\cite{WerfP18}.
The paper at hand closes this gap.

\medskip
In this paper, we propose a correctness criterion for Petri nets with identifiers that combines the checks of the soundness of the system's processes with the soundness of object manipulations within those processes.
Intuitively, objects of a specific type are correctly manipulated by the system if every object instance of that type, characterized by a unique identifier, can ``leave'' the system, that is, a dedicated transition of the system can consume it, and once that happens, no references to that object instance remain in the system.
When a system achieves this harmony for its processes and all data object types, we say that the system is \emph{identifier sound}, or, alternatively, that the data and processes of the system are in \emph{resonance}.
Specifically, this paper makes these contributions:\medskip

\begin{compactitem}
\itemsep=1.5pt
\item
It motivates and defines the notion of \emph{identifier soundness} for checking correctness of data object manipulations in processes of a system;
\item
It proposes a resource-aware extension for systems and defines a suitable correctness criterion building on top of the one of identifier soundness and requiring that system resources are managed \emph{conservatively};
\item
It discusses aspects related to \emph{decidability of identifier soundness} in the general case and for certain restricted, but still useful, classes of systems;
\item
It establishes connections with existing results on verification of data-aware processes and shows which verification tasks are decidable for object-aware systems.
\end{compactitem}

\medskip
The paper proceeds as follows.
The next section introduces concepts and notions required to support subsequent discussions.
Section~\ref{sec:pnids} introduces typed Petri nets with identifiers, a model for modeling distributed systems whose state is defined by objects the system manipulates.
Section~\ref{sec:soundness} presents various correctness notions for typed Petri nets with identifiers, including identifier soundness, and demonstrates a proof that the notion is in general undecidable. Moreover, the section discusses the connection to existing verification results and shows which verification tasks are decidable for typed Petri nets with identifiers.
Section~\ref{sec:soundness:by:construction} discusses several classes of systems for which identifier soundness is guaranteed by construction.
Section~\ref{sec:soundness:with:resources} presents the formalism extension with resource management capabilities and discusses a series of results, including resource-aware soundness, that is deemed to be undecidable.
Finally, the paper concludes with a discussion of related work and future work.

\makeatletter
\newcommand{\xdasharrow}[2][->,>=angle 90]{
\tikz[baseline=-\the\dimexpr\fontdimen22\textfont2\relax]{
\node[anchor=south,font=\scriptsize, inner ysep=1.5pt,outer xsep=2.5pt](x){\ensuremath{#2}};
\draw[shorten <=3.4pt,shorten >=3.4pt,dashed,#1](x.south west)--(x.south east);
}
}

\section{Preliminaries}
Let $S$ and $T$ be sets.
The powerset of $S$ is denoted by $\powerset{S} = \{ S' \mid S' \subseteq S\}$ and $|S|$ denotes the cardinality of $S$. 
Given a relation $R \subseteq S \times T$,
its range is defined by $\rng{R} = \{y \in T \mid \exists x \in S : (x,y)\in R\}$. \footnote{Notice that $R$ can be also seen as a function $R:X\to T$.}
A \emph{multiset} $m$ over $S$ is a mapping of the form $m:S\rightarrow \naturals$, where $\naturals = \{0, 1, 2, \ldots\}$ denotes the set of natural numbers.
For $s \in S$, $m(s) \in \mathbb{N}$ denotes the number of times $s$ appears in the multiset.
For $x \not\in S$, $m(x) = 0$.
We write $s^n$ if $m(s)=n$.
We use $\mult S$ to denote the set of all finite multisets over $S$ and $\emptyset$ to denote the \emph{empty multiset}.
The support of $m\in\mult S$ is the set of elements that appear in $m$ at least once: $\supp{m} = \set{s\in S\mid m(s) > 0}$.
Given two multisets $m_1$ and $m_2$ over $S$, we consider the following standard multiset operations:\smallskip

\begin{compactitem}
\itemsep=1.2pt
\item $m_1 \leq m_2$ iff $m_1(s) \leq m_2(s)$ for each $s \in S$;
\item $m_1 + m_2 = \set{s^n\mid s\in S, n =m_1(s) + m_2(s)}$;
\item if $m_1 \leq m_2$, $m_2 - m_1 = \set{s^n \mid s\in S, n=m_2(s) - m_1(s)}$. \smallskip
\end{compactitem}

\noindent We also write $|m|=\sum_{s\in S}m(s)$ to denote the cardinality of $m$.
%
A \emph{sequence} over $S$ of length $n \in \naturals$ is a function $\sigma : \{1,\ldots,n\} \to S$. If $n > 0$ and $\sigma(i) = a_i$, for $1\leq i \leq n$, we write $\sigma = \sequence{a_1, \ldots, a_n}$.
The length of $\sigma$ is denoted by $|\sigma|$ and is equal to $n$. The sequence of length $0$ is called the \emph{empty sequence}, and is denoted by $\emptysequence$. The set of all finite sequences over $S$ is denoted by $S^*$.
We write $a \in \sigma$ if there is $1 \leq i \leq |\sigma|$ such that $\sigma(i) = a$.
\emph{Concatenation} of two sequences $\nu,\gamma \in S^*$, denoted by $\sigma = \nu \concat \gamma$, is a sequence defined by $\sigma : \{ 1, \ldots, |\nu|+|\gamma|\}\rightarrow S$, such that $\sigma(i) = \nu(i)$ for $1 \leq i \leq |\nu|$, and $\sigma(i) = \gamma(i - |\nu|)$ for $|\nu|+1 \leq i \leq |\nu|+|\gamma|$.
We define the projection of sequences on a set $T$ by induction as follows:
\begin{inparaenum}[\it (i)]
\item $\proj{\emptysequence}{T} = \emptysequence$;
\item $\proj{(\sequence{a}\concat\sigma)}{T} = \sequence{a}\concat\proj{\sigma}{T}$, if $a \in T$;
\item $\proj{(\sequence{a}\concat\sigma)}{T} = \proj{\sigma}{T}$, if $a\not\in T$.
\end{inparaenum}
Renaming sequence $\sigma$ with an injective function $r : S \rightarrow T$ is defined inductively by $\rho_r(\emptysequence) = \emptysequence$, and  $\rho_r(\sequence{a}\concat\sigma) = \sequence{r(a)}\concat\rho_r(\sigma)$. Renaming is extended to multisets of sequences as follows: given a multiset $m \in \mult{(S^*)}$, we define $\rho_r(m) = \sum_{\sigma\in\supp{m}} \sigma(m)\cdot\rho_r(\sigma)$. For example, $\rho_{\{x\mapsto a, y \mapsto b\}}(\sequence{x,y}^3) = \sequence{a,b}^3$.

\medskip\noindent
\textbf{Labeled Transition Systems.}\quad
To model the behavior of a system, we use \emph{labeled transition systems}.
Given a finite set $A$ of (action) labels, a \emph{(labeled) transition system} (LTS) over $A$ is a tuple \linebreak $\tsys{}{} = (S,A,s_0, \to)$, where $S$ is the (possibly infinite) set of \emph{states}, $s_0$ is the \emph{initial state} and \linebreak
$\to\ \subseteq (S\times (A \cup \{\tau\}) \times S)$ is the \emph{transition relation}, where $\tau\not\in A$ denotes the silent action~\cite{glabbeekoverviewii}. 
In what follows, we write $s \xrightarrow{a} s'$ for $(s,a,s') \in \to$.
Let $r : A \to (A' \cup \{\tau\})$ be a total function.
Renaming $\tsys{}{}$ with $r$ is defined as $\rho_{r}(\Gamma) = (S,  A'\cup\set{\tau}, s_0, \to')$ with $(s,r(a),s') \in \to'$ iff $(s,a,s') \in \to$.
Given a set $T$, hiding is defined as $\hide{T}(\Gamma) = \rho_{h}(\Gamma)$ with $h : A \rightarrow A \cup \{\tau\}$ such that $h(t) = \tau$ if $t \in T$ and $h(t)=t$ otherwise.
Given $a\in A$, $p \xdasharrow{~a~} q$ denotes a \emph{weak transition relation} that is defined as follows:
\begin{inparaenum}[\it (i)]
	\item $p \xdasharrow[->]{~a~} q$ iff $p (\xrightarrow{\tau})^* q_1\xrightarrow{a}q_2 (\xrightarrow{\tau})^* q$;
	\item $p \xdasharrow[->]{\ensuremath{~\tau~}} q$ iff $p (\xrightarrow{\tau})^* q$.
\end{inparaenum}
Here, $(\xrightarrow{\tau})^*$ denotes the reflexive and transitive closure of $\xrightarrow{\tau}$.

\begin{definition}[Strong and weak bisimulation]
\label{def:bisimulation}\label{def:strong-bisimulation}
Let $\tsys{1} = (\states_1,A,s_{01},\to_1)$ and $\tsys{2} = (\states_2,A,s_{02},\to_2)$ be two LTSs.
A relation $R\subseteq(\states_1 \times \states_2)$ is called a \emph{strong simulation}, denoted as $\tsys{1} \prec_R \tsys{2}$, if for every pair $(p,q)\in R$ and $a\in A \cup \{\tau\}$, it holds that if $p\xrightarrow{a}_1 p'$, then there exists $q'\in\states_2$ such that $q \xrightarrow{a}_2 q'$ and $(p',q')\in R$. Relation $R$ is a \emph{weak simulation}, denoted by $\tsys{1} \preccurlyeq_R \tsys{2}$, iff  for every pair $(p,q)\in R$ and $a\in A \cup \{\tau\}$ it holds that if $p\xrightarrow{a}_1 p'$, then either $a = \tau$ and $(p', q) \in R$, or there exists $q'\in\states_2$ such that $q \xdasharrow{~a~}_{\hspace{-.5ex} 2} ~q'$ and $(p',q')\in R$.

$R$ is called a strong (weak) \emph{bisimulation}, denoted by $\tsys{1} \sim_R \tsys{2}$ ($\tsys{1} \approx_R \tsys{2}$) if both $\tsys{1} \prec \tsys{2}$ ($\tsys{1} \preccurlyeq_R \tsys{2}$) and $\tsys{2} \prec_{R^{-1}} \tsys{1}$ ($\tsys{2} \preccurlyeq_{R^{-1}} \tsys{1}$).
The relation is called \emph{rooted} iff $(s_{01}, s_{02}) \in R$. A rooted relation is indicated with a superscript ${}^r$.
\end{definition}

\medskip\noindent
\textbf{Petri nets.}\quad
A weighted Petri net is a 4-tuple $(\places,\transitions,\flow,W)$ where
$\places$ and $\transitions$ are two disjoint sets of \emph{places} and \emph{transitions}, respectively, $F \subseteq ((P \times T) \cup (T\times P))$ is the \emph{flow relation}, and $W : ((P \times T) \cup (T\times P)) \rightarrow \naturals$ is a \emph{weight function} such that $W(f) > 0$ iff $f\in F$.
For $x\in\places\cup\transitions$, we write $\pre{x}=\set{y\mid (y,x)\in\flow}$ to denote the \emph{preset} of $x$ and $\post{x}=\set{y\mid (x,y)\in\flow}$ to denote the \emph{postset} of $x$.
We lift the notation of preset and postset to sets element-wise.
If for a Petri net no weight function is explicitly defined, we assume $W(f) = 1$ for all $f \in F$.
A \emph{marking} of $N$ is a multiset $m \in \mult P$, where $m(p)$ denotes the number of \emph{tokens} in place $p \in P$. If $m(p) > 0$, place $p$ is called \emph{marked} in marking $m$.
A \emph{marked Petri net} is a tuple $(N, m)$ with $N$ a weighted Petri net with marking $m$.
A transition $t\in T$ is enabled in $(N,m)$, denoted by $\enabled{(N,m)}{t}$ iff $W((p,t)) \leq m(p)$ for all $p \in \pre{t}$. An enabled transition can \emph{fire}, resulting in marking $m'$ iff $m'(p) + W((p,t)) = m(p) + W((t,p))$, for all $p \in P$, and is denoted by $\fire{(N,m)}{t}{(N,m')}$. We lift the notation of firings to sequences. A sequence $\sigma \in T^*$ is a \emph{firing sequence of $(N,m_0)$} iff $\sigma = \emptysequence$, or markings $m_0, \ldots, m_n$ exist such that $\fire{(N,m_{i-1})}{\sigma(i)}{(N,m_{i})}$ for $1 \leq i \leq |\sigma| = n$, and is denoted by $\fire{(N,m_0)}{\sigma}{(N,m_n)}$.
If the context is clear, we omit $N$, and just write $\fire{m_0}{\sigma}{m_n}$.
The set of reachable markings of $(N,m)$ is defined by $\reachable{N}{m} = \{ m' \mid \exists \sigma \in \transitions^* : \fire{m}{\sigma}{m'} \}$.
The semantics of a marked Petri net $(N,m_0)$ with $N = (\places, \transitions, \flow, W)$ is defined by the LTS $\Gamma_{N,m_0} = (\mult P, T, m_0, \to)$ with $(m,t,m') \in \to$ iff $\fire{m}{t}{m'}$.

\medskip
\noindent
\textbf{Workflow Nets.}\quad
A \emph{workflow net} (WF-net for short) is a tuple $N=(\places, \transitions, \flow,W,\inp,\outp)$ such that:
\begin{inparaenum}[\it (i)]
	\item $(\places, \transitions, \flow, W)$ is a weighted Petri net;
	\item $\inp,\outp\in\places$ are the source and sink place, respectively, with $\pre{\inp} =\post{\outp}= \emptyset$;
	\item every node in $\places \cup \transitions$ is on a directed path from $\inp$ to $\outp$.
\end{inparaenum}
$N$ is called \emph{$k$-sound} for some $k \in \naturals$ iff
\begin{inparaenum}[\it (i)]
\item it is proper completing, i.e., for all reachable markings $m \in \reachable{N}{[\inp^k]}$, if $[\outp^k]\leq m$, then $m = [\outp^k]$;
\item it is weakly terminating, i.e., for any reachable marking $m \in \reachable{N}{[\inp^k]}$, the final marking is reachable, i.e., $[\outp^k] \in \reachable{N}{m}$; and
\item it is quasi-live, i.e., for all transitions $t\in\transitions$, there is a marking $m\in\reachable{N}{[\inp]}$ such that $\enabled{m}{t}$.
\end{inparaenum}
The net is called \emph{sound} if it is $1$-sound.
If it is $k$-sound for all $k\in\naturals$, it is called \emph{generalized sound}~\cite{HeeSV03}.

\section{Typed Petri nets with identifiers}
\label{sec:pnids}
Processes and data are highly intertwined: processes manipulate data objects while objects govern processes.
For example, consider a retail shop with three types of objects: \emph{products}  sold through the shop, \emph{customers} that can order these products, and \emph{orders} that track products bought by customers.
This example already involves many-to-many relations between objects, e.g., a product can be ordered by many customers, while a customer can order many products.
Relations between objects can also be one-to-many, e.g., an order is always for a single customer, but a customer can have many orders.
In addition, objects may have life cycles, which themselves can be considered as processes.
Figure~\ref{fig:lifecycles} shows three life cycles of objects in the retail shop.
A product may be temporarily unavailable, while customers may be blocked by the shop, disallowing them to order products.
These life cycles are inherently intertwined. For instance, customers should not be allowed to order products that are unavailable.
Similarly, blocked customers should not be able to create new orders.

\begin{figure}[!h]
	\centering
	\includegraphics[width=\textwidth]{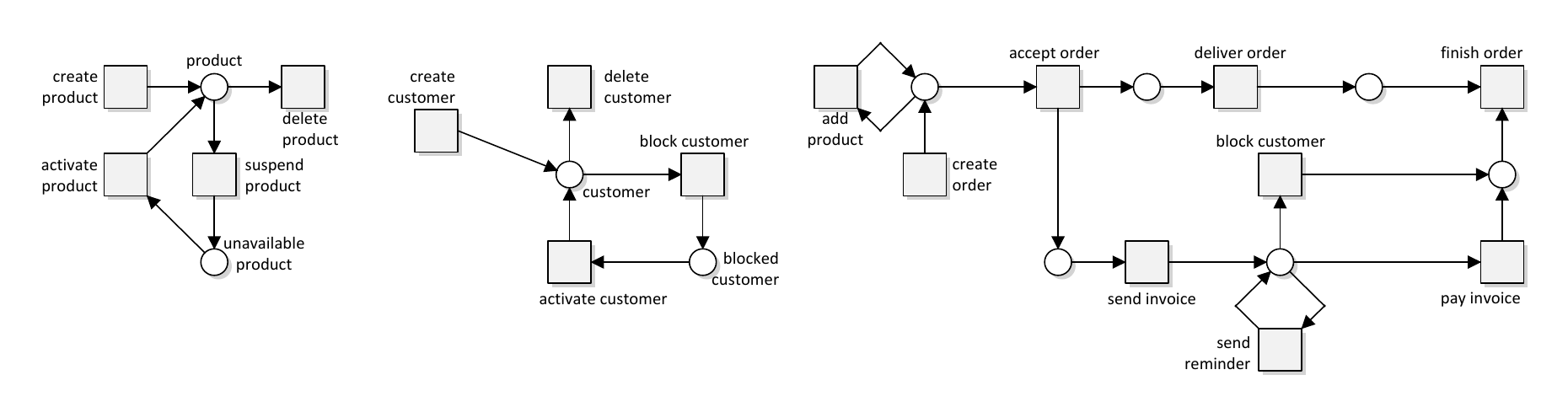}
	\caption{The life cycles of products, customers and orders in the retail shop.}\label{fig:lifecycles}
\end{figure}

Several approaches have been studied to model and analyze models that combine objects and processes.
For example, data-aware Proclets~\cite{Fahland19} allow describing the behavior of individual artifacts and their interactions.
Another approach is followed in $\nu$-PN~\cite{RVFE11}, in which a token can carry a single identifier~\cite{RosaMF06}.
In this formalism, markings map each place to a bag of identifiers, indicating how many tokens in each place carry the same identifier.
These identifiers can be used to reference entities in an information model.
However, referencing a fact composed of multiple entities is not possible in $\nu$-PNs.
In this paper, we study \emph{typed Petri nets with identifiers} (\tpnids), which build upon $\nu$-PNs~\cite{RVFE11} by extending tokens to carry vectors of  identifiers~\cite{PolyvyanyyWOB19,WerfP20}.
Vectors, represented by sequences, have the advantage that a single token can refer to multiple objects or entities that compose (a part of) a \emph{fact}, such as an order is for a specific customer.
Identifiers are typed, i.e., the countable, infinite set of identifiers is partitioned into a set of types, such that each type contains a countable, infinite set of identifiers.
Identifier types should not overlap, i.e., each identifier has a unique type.
Variables can take values of identifiers, and, thus, are typed as well and can only refer to identifiers of the associated type.
For example, the product, customer and order objects from the retail shop example make three object types.

\begin{definition}[Identifier Types]
\label{def:id-types}
Let $\I$, $\Lambda$, and $\V$ denote countable, infinite sets of  identifiers, type labels, and variables, respectively. We define:
	\begin{itemize}
\itemsep=0.9pt
		\item the \emph{domain assignment} function $I : \Lambda \rightarrow \powerset{\I}$, such that $I(\lambda_1)$ is an infinite set, and $I(\lambda_1) \cap I(\lambda_2) \neq \emptyset$ implies $\lambda_1 = \lambda_2$ for all $\lambda_1, \lambda_2 \in \Lambda$;
		\item the \emph{id typing} function $\type_{I}:\I\to\Lambda$ s.t. if $\type_{I}(\cname{id})=\lambda$, then $\cname{id}\in I(\lambda)$;
		\item a \emph{variable typing} function $\type_{\V}:\V\to\Lambda$, prescribing that $x\in\V$ can be substituted only by values from $I(\type_{\V}(x))$.
	\end{itemize}
	When clear from the context, we omit the subscripts of $type$.
\end{definition}
For ease of presentation, we assume the natural extension of the above typing functions to the cases of sets and vectors, for example, $\type_{\V}(x_1\cdots x_n)=\type_{\V}(x_1)\cdots\type_{\V}(x_n)$.

\medskip
In a \tpnid, each place is annotated with a \emph{place type}, which is a vector of types, indicating types of identifier tokens the place can carry.
A place with the empty place type, represented by the empty vector, is a classical Petri net place carrying indistinguishable (black) tokens.
Each arc of a \tpnid is inscribed with a multiset of vectors of variables, such that the types of the variables in the vector coincide with the place types.
This approach allows modeling situations where a transition may require multiple tokens with different identifiers from the same place.

\begin{figure}[!b]
\vspace*{-3mm}
	\centering
	\includegraphics[scale=.8]{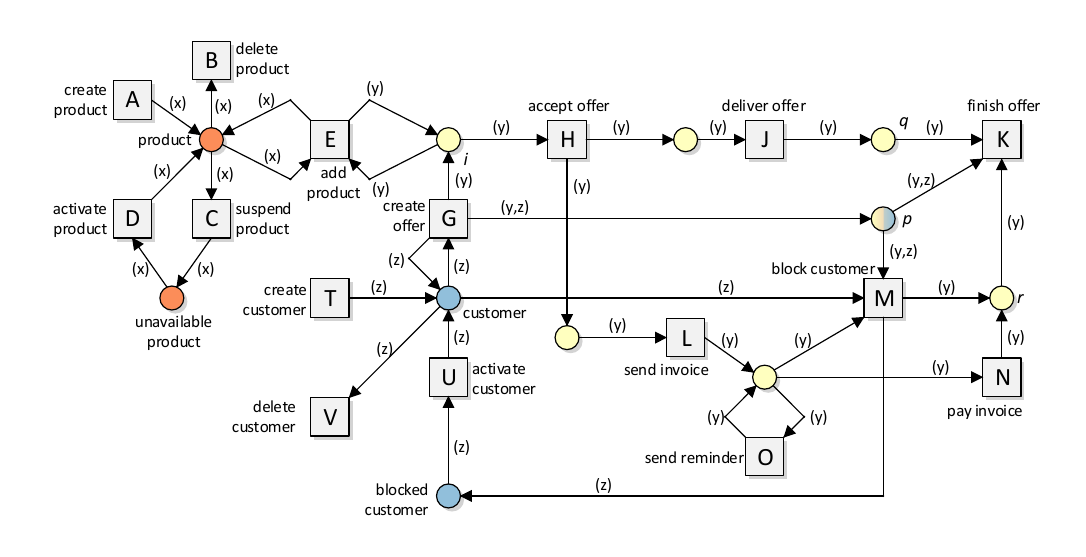}\vspace*{-4mm}
	\caption{\tpnid $N_{\mathit{rs}}$ for a retail shop that manipulates products, customers and orders. Each place is colored according to its type.
Place $p$ carries pairs of identifiers: an order and a customer.}\label{fig:overallModel}
\end{figure}

\begin{definition}[Typed Petri nets with identifiers]
\label{def:tpnid}
A \emph{Typed Petri net with identifiers} (\tpnid) $N$ is a tuple
$(\places,\transitions,\flow,\alpha,\beta)$, where:
\begin{itemize}
\itemsep=0.9pt
\item $(\places,\transitions,\flow)$ is a Petri net;
\item $\alpha:\places\to\Lambda^*$ is the \emph{place typing function};
\item $\beta : \flow\to \mult{(\V^*)}$ defines for each flow a multiset of \emph{variable vectors} such that
$\alpha(p) = \type(\vec{x})$ for any $\vec{x} \in \supp{\beta((p,t))}$ and $\type(\vec{y})=\alpha(p')$ for any $\vec{y} \in \supp{\beta((t,p'))}$ where $t\in\transitions$, $p\in\pre{t}$, $p'\in\post{t}$;
\end{itemize}
\end{definition}

Figure~\ref{fig:overallModel} shows a \tpnid, $N_{\mathit{rs}}$, of a retail shop.
Each place is colored according to its type.
The net intertwines the life cycles of Fig.~\ref{fig:lifecycles} and weakly simulates each of these life cycles.
In $N_{\mathit{rs}}$, places \emph{product} and \emph{unavailable product} are annotated with a vector $\langle\mathit{product}\rangle$, i.e., these places contain tokens that carry only a single identifier of type $\mathit{product}$.
Places \emph{customer} and \emph{blocked customer} have type $\langle\mathit{customer}\rangle$.
All other places, except for place $p$, are labeled with type $\langle\mathit{order}\rangle$.
Place $p$ maintains the relation between orders and customers, and is typed $\langle\mathit{order},\mathit{customer}\rangle$, i.e., tokens in this place are identifier vectors of size 2. 
$N_{\mathit{rs}}$ uses three variables: $x$ for $\mathit{product}$, $y$ for $\mathit{order}$ and $z$ for $\mathit{customer}$.

\smallskip
A marking of a \tpnid $N$ is a configuration of tokens over its places.
The set of all possible markings of $N$ is denoted by $\mathbb{M}(N)$.
Each token in a place should be of the correct type, i.e., the vector of identifiers carried by a token in a place should match the corresponding place type.
All possible vectors of identifiers a place $q$ may carry is defined by the set $\colset(q)$. 

\begin{definition}[Marking]
	Given a \tpnid $N = (\places, \transitions, \flow, \alpha, \beta)$,
	and place $p \in \places$, its \emph{id set} is 
	$\colset(p) = \prod_{1 \leq i \leq |\alpha(p)|} I(\alpha(p)(i))$.
	A \emph{marking} is a function $M \in \mathbb{M}(N)$, with $\mathbb{M}(N) = P \to \mult{(\I^*)}$, such that $M(p) \in \mult {\colset(p)}$, for each place $p \in P$.
	The set of identifiers used in $m$ is denoted by
	$\id{M} = \{ \cname{id} \mid \exists \vec{\cname{id}} \in \colset(p),$ $p \in P: \cname{id} \in \vec{\cname{id}} \wedge M(p)(\vec{\cname{id}}) > 0 \}$.
	The pair $(N, M)$ is called a \emph{marked \tpnid}.
\end{definition}

To define the semantics of a \tpnid, the variables need to be valuated with identifiers.
Variables may be used differently by transitions.
In Fig.~\ref{fig:overallModel}, transition $G$ uses variable $y$ to create an identifier of type $\mathit{order}$.
Transition $K$ uses the same variable $y$ to remove identifiers of type $\mathit{order}$ from the marking, as it has no outgoing arcs, and thus only consumes tokens.
We, therefore, first introduce some notation to work with variables and types in a \tpnid.
Variables used on the input arcs, i.e., variables on arcs from a place to a transition $t$ are called the input variables of $t$.
Similarly, variables on arcs from transition $t$ to a place are called the output variables of $t$.
A variable that only occurs in the set of output variables of a transition, is an \emph{emitting variable}.
Similarly, if a variable only appears as an input variable of a transition, it is called a \emph{collecting variable}.
As variables are typed, an emitting variable creates a new identifier of a corresponding type upon transition firing, whereas a collecting variable removes the identifier.

\begin{definition}[Variable sets, emitter and collector transitions, object types]
\label{def:notations}
Given a \tpnid $N=(\places,\transitions,\flow,\alpha,\beta)$, $t\in \transitions$ and $\lambda \in \Lambda$, we define the following sets of variables:
\begin{itemize}
\itemsep=0.8pt
	\item \emph{input variables} as $\smash{\invar{t} = \bigcup_{\vec x \in \supp{\beta((p,t))}, p\in\pre{t}} \bigcup_{x\in\vec x}x}$;
	\item \emph{output variables} as $\smash{\outvar{t} = \bigcup_{\vec x \in \supp{\beta((t,p))},p \in\post{t}} \bigcup_{x\in\vec x}x}$;
	\item \emph{variables} as $\var{t} = \invar{t} \cup \outvar{t}$;
	\item \emph{emitting variables} as $\newvar{t} = \outvar{t}\setminus\invar{t}$;
	\item \emph{collecting variables} as $\delvar{t} = \invar{t} \setminus \outvar{t}$.
\end{itemize}
Using the above notions, we introduce the sets of:
\begin{itemize}
\itemsep=0.9pt
	\item \emph{emitting transitions} (or simply referred to as emitters) as $E_N(\lambda) = \{ t \mid \exists x \in \newvar{t} \wedge \type(x) = \lambda \}$;
	\item \emph{collecting transitions} (or simply referred to as collectors) as $C_N(\lambda) = \{ t \mid \exists x \in \delvar{t} \wedge \type(x) = \lambda \}$.
	\end{itemize}
To properly account for \emph{place types used in $N$}, we introduce $\type_\places(N) = \{ \vec\lambda \mid \exists p \in P : \vec\lambda\in\alpha(p)\}$. Similarly, for objects, we introduce the set of \emph{object types in $N$} $\type_\Lambda(N) = \{ \lambda \mid \lambda\in \vec\lambda, \vec\lambda\in\type_\places(N)\}$.
\end{definition}


A firing of a transition requires a \emph{binding} that valuates variables to identifiers.
The binding is used to inject new fresh data into the net via variables that emit identifiers.
We require bindings to be an injection, i.e., no two variables within a binding may refer to the same identifier.
Note that in this definition, the freshness of identifiers is local to the marking, i.e., disappeared identifiers may be reused, as it does not hamper the semantics of the \tpnid.
Our semantics allow the use of well-ordered sets of identifiers, such as the natural numbers, as used in~\cite{PolyvyanyyWOB19,RosaMF06} to ensure that identifiers are globally new.
Here we assume local freshness over global freshness.

\begin{definition}[Firing rule]
Given a marked \tpnid $(N, M)$ with $N=(\places,\transitions,\flow,\alpha,\beta)$, a \emph{binding} for  transition $t\in T$ is an injective function  $\psi:\V\rightarrow \I$ such that
$\type(v) =\type(\psi(v))$ and
$\psi(v)\not\in \id{M}$ iff $v\in\newvar{t}$.
Transition $t$ is \emph{enabled} in $(N, M)$ under binding
$\psi$, denoted by $\enabled{(N, M)}{t,\psi}$ iff $\rho_\psi(\beta(p,t)) \leq M(p)$ for all $p\in\pre{t}$.
Its firing results in marking $M'$, denoted by $\fire{(N, M)}{t,\psi}{(N, M')}$, such that $M'(p) + \rho_\psi(\beta(p,t)) = M(p) + \rho_\psi(\beta(t,p))$. 
\end{definition}

Again, the firing rule is inductively extended to sequences $\eta\in (\transitions\times(\V\rightarrow \I))^*$.
A marking $M'$ is \emph{reachable} from $M$ if there exists $\eta\in (\transitions\times(\V\rightarrow \I))^*$ s.t. $\fire{(N,M)}{\eta}{(N,M')}$.
We denote with $\reachable{N}{M}$ the set of all markings reachable from $(N, M)$.

The execution semantics of a \tpnid is defined as an LTS that accounts for all possible executions starting from a given initial marking.

\begin{definition}[Induced transition system]
\label{def:induced-ts}
Given a marked \tpnid $(N, M_0)$ with $N=(P,T,F,\alpha,\beta)$, its \emph{induced transition system}  is $\tsys{N, M_0} = (\mathbb{M}(N),(T\times(\V\to\I)),M_0, \to)$ with $M\xrightarrow{(t,\psi)} M^\prime$ iff  $\fire{(N,M)}{t,\psi}{(N,M^\prime)}$.

\end{definition}

\tpnids are a vector-based extension of $\nu$-PNs~\cite{RVFE11}.
In other words, a $\nu$-PN can be translated into a strongly bisimilar \tpnid with a single type, and all place types are of length of at most $1$, which follows directly from the definition of the firing rule~\cite{RVFE11}.

\begin{corollary}
	\label{lemma:nupn}
	For any $\nu$-PN there exists a single-typed \tpnid  such that the two nets are strongly rooted bisimilar.
\end{corollary}

As a result, the decidability of reachability for $\nu$-PNs transfers to \tpnids~\cite{RVFE11}.

\begin{proposition}\label{prop:reachabilityundecidable}
Reachability is undecidable for \tpnids.
\end{proposition}

\section{Correctness criteria for \tpnids}
\label{sec:soundness}
Many criteria have been devised for assessing the correctness of systems captured as Petri nets.
Traditionally, Petri net-based criteria focus on the correctness of processes the systems can support.
Enriching the formalism with ability to capture object manipulation
while keeping analyzability is a delicate balancing act.

\medskip
For \tpnids, correctness criteria can be categorized as 
system-level and object-level.
Criteria at the system-level (Section~\ref{sec:systemlevel}) focus on traditional Petri net-based criteria to assess the system as a whole, whereas criteria at the object-level (Section~\ref{sec:objectlevel}) address the correctness of individual objects represented by identifiers.

\subsection{System-level correctness criteria}\label{sec:systemlevel}
Liveness is an example of a system-level correctness property.
It expresses that any transition is always eventually enabled again.
As such, a live system guarantees that its activities cannot eventually become unavailable.

\begin{definition}[Liveness]
	A marked \tpnid $(N, M_0)$ with $N = (\places, \transitions,\flow,\alpha,\beta)$ is \emph{live} iff for every marking $M\in\reachable{N}{M_0}$ and every transition $t \in \transitions$, there exists a marking $M'\in\reachable{N}{M}$ and a binding $\psi:\V\rightarrow \I$  such that $\enabled{M'}{t,\psi}$.
\end{definition}

Boundedness expresses that the reachability graph of a system is finite, i.e., that the system has finitely many possible states and state transitions.
Hence, boundedness is another example of a system-level correctness property.
Many systems can support an arbitrary number of simultaneously active objects; they are unbounded by design.
Similar to $\nu$-PN, we differentiate between various types of boundedness~\cite{RosaF10}.
Specifically, \emph{boundedness} expresses that the number of tokens in any reachable place does not exceed a given bound.
\emph{Width-boundedness} expresses that the modeled system has a bound on the number of simultaneously active objects. 

\begin{definition}[Bounded, width-bounded]
\label{def:boundedness}
Let $(N,M_0)$ be a marked \tpnid with $N=(\places, \transitions,\flow,\alpha,\beta)$. A place $p \in P$ is called:
\begin{itemize}
\itemsep=0.9pt
	\item \emph{bounded} if there is $k \in \naturals$ such that $|M(p)| \leq k$ for all $M \in \reachable{N}{M_0}$;
	\item \emph{width-bounded} if there is $k \in \naturals$ such that $|\supp{M(p)}|\! \leq\! k$ for all $M \in \reachable{N}{M_0}$;
\end{itemize}
If all places in $(N,M_0)$ are (width-) bounded, then $(N,M_0)$ is called (width-) bounded.
\end{definition}

As transitions $A$ and $T$ in Fig.~\ref{fig:overallModel} have no input places, these transitions are always enabled.
Consequently, places \emph{product} and \emph{customer} are not bounded, and thus no place in $N_{\mathit{rs}}$ is bounded.
Upon each firing of transition $A$ or $T$, a new identifier is created.
Hence, these places are also not \emph{width-bounded}.
In other words, the number of objects in the system represented by $N_{\mathit{rs}}$ is dynamic, without an upper bound.


\subsection{Object-level correctness criteria}\label{sec:objectlevel}
An object-level property assesses the correctness of individual objects.
In \tpnids, identifiers can be seen as references to objects: if two tokens carry the same identifier, they refer to the same object.
The projection of an identifier on the reachability graph of a marked \tpnid represents the life-cycle of the referenced object.
Boundedness of a system implies that the number of states of the reachability graph is finite.
\emph{Depth-boundedness} captures this idea for identifiers: in any marking, the number of tokens that refer to a single identifier is bounded.
In other words, if a marked \tpnid is depth-bounded, the complete system may still be unbounded, but the life-cycle of each object is finite.

\begin{definition}[Depth-boundedness]
	Let $(N,m_0)$ be a marked \tpnid with $N=(\places, \transitions,\flow,\alpha,\beta)$. A place $p \in P$ is called \emph{depth-bounded} if for each identifier $\cname{id}\in\I$ there is $k\in\naturals$ such that $m(p)(\vec{\cname{id}}) \leq k$
	for all $m \in \reachable{N}{m_0}$ and
	$\vec{\cname{id}}\in \colset(p)$ with $\cname{id}\in\vec{\cname{id}}$.
	If all places in $\places$ are depth-bounded, $(N,m_0)$ is called depth-bounded.
\end{definition}

Depth-boundedness is undecidable for $\nu$-PNs~\cite{RVFE11} and, thus, also for \tpnids.

\begin{proposition}
	Depth-boundedness is undecidable for \tpnids.
\end{proposition}

The idea of depth-boundedness is to consider a single identifier in isolation, and study its reachability graph.
Intuitively, an object of a given type ``enters'' the system via an emitter that creates a unique identifier that refers to the object.
The identifier remains in the system until the object ``leaves'' the system by firing a collecting transition (that binds to the identifier and consumes the last token in the net that refers to it).
In other words, if a type has emitters and collectors, it has a life-cycle, which can be represented as a process.
The process of a type is the model describing all possible paths for the type.
It can be derived by taking the projection of the \tpnid on all transitions and places that are ``involved'' in the type. Notably, the net obtained after the projection is just a regular Petri net.

\begin{definition}[Type projection]
Let $\lambda \in \Lambda$ be a type.
Given a \tpnid $N = (P_N, T_N, F_N, \alpha_N, \beta_N)$, its $\lambda$-projection $\pi_\lambda(N) = (P, T, F, W)$ is a Petri net defined by:
\begin{itemize}
\itemsep=0.85pt
	\item $P = \{ p \in P_N \mid \lambda \in \alpha_N(p) \}$;
	\item $T = \{ t \in T_N \mid \exists p \in P_N :
	\lambda \in \type_{\V}\bigl(\supp{\beta_N((p,t))}\bigr)
	\lor
	\lambda \in \type_{\V}\bigl(\supp{\beta_N((t,p))}\bigr)\}$
	\item $F = F_N \cap ((P \times T) \cup (T \times P))$;
	\item $W(f) = |\beta_N(f)|$ for all $f \in F$.
\end{itemize}
Give a marking $M \in \mathbb{M}(N)$, its $\lambda$-projection $\pi_\lambda(M)$ is defined by $\pi_\lambda(M)(p) = |M(p)|$.
\end{definition}

\begin{figure}[!h]
 \vspace*{-3mm}
	\centering
\hspace*{-5mm}	\begin{subfigure}{.48\textwidth}
		\centering
		\includegraphics[scale=.66]{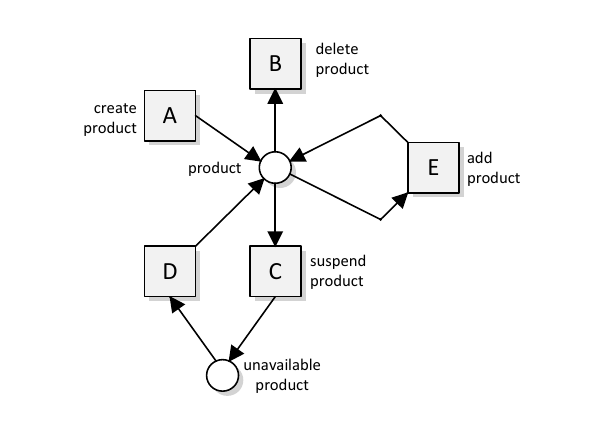}\vspace*{-2mm}
		\caption{Product projection}
	\end{subfigure}\hfil
\hspace*{-12mm}	\begin{subfigure}{.48\textwidth}
		\centering
		\includegraphics[scale=.66]{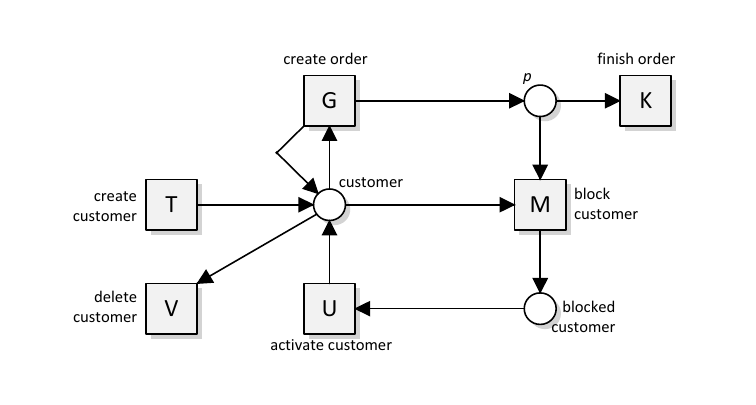}\vspace*{-2mm}
		\caption{Customer projection}\label{fig:cust-projection}
	\end{subfigure}
\hspace*{-5mm} \begin{subfigure}{\textwidth}
		\centering
		\includegraphics[scale=.66]{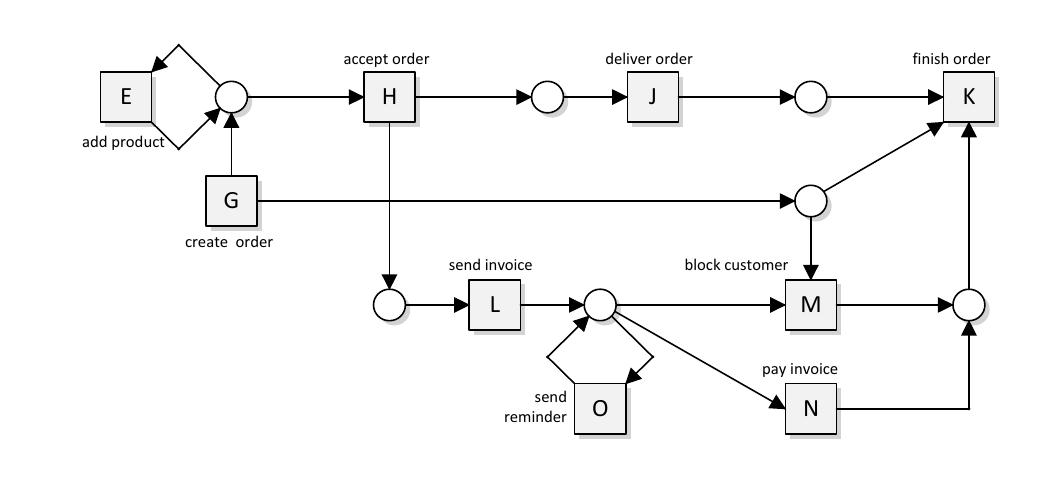}\vspace*{-2mm}
		\caption{Order projection}
	\end{subfigure}\vspace*{-2mm}
	\caption{Type projections of Figure~\ref{fig:overallModel}. The customer projection is not sound, as place $p$ is not
   bounded.}\label{fig:projections}\vspace*{-4mm}
\end{figure}

\begin{figure}[!h]
\vspace*{1mm}
	\centering
	\begin{subfigure}{.4\textwidth}
		\centering
		\includegraphics[scale=.5]{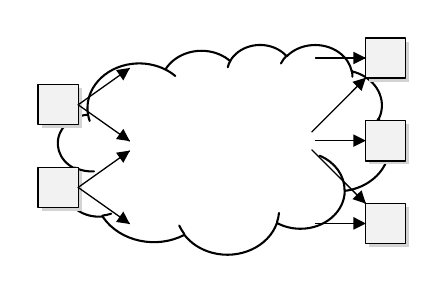}\vspace*{-1mm}
		\caption{}
	\end{subfigure}
	\begin{subfigure}{.4\textwidth}
	\centering
	\includegraphics[scale=.5]{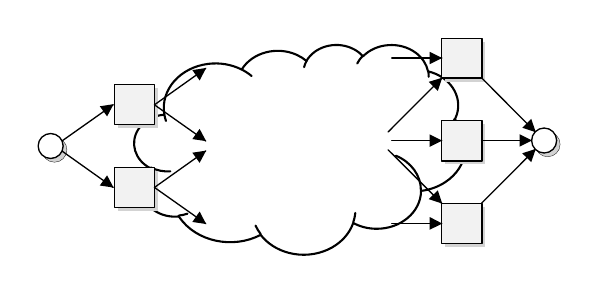}\vspace*{-1mm}
	\caption{}
 \end{subfigure}\vspace*{-3mm}
 \caption{A transition-bordered WF-Net (a) and its closure (b)~\cite{HeeSW13}.}\label{fig:tBorderedNet}\vspace*{-3mm}
\end{figure}

Figure~\ref{fig:projections} shows the three type projections of $N_{\mathit{rs}}$ from Figure~\ref{fig:overallModel}.
As an emitter of a type creates a new identifier, and a collector removes the created identifier, each type with emitters and collectors can be represented as a transition-bordered WF-net~\cite{HeeSW13}.
Instead of a source and a sink place, a transition-bordered WF-net has dedicated transitions that represent the start and finish of a process.
A transition-bordered WF-net is sound if its closure is sound~\cite{HeeSW13}.
As shown in Fig.~\ref{fig:tBorderedNet}, the closure is constructed by creating a new source place 
so that each emitting transition consumes from it, and a new sink place so that each collecting transition produces in it.
In the remainder of this section, we develop this intuition of soundness of type projections into the concept of identifier soundness of \tpnid{}s.

Many soundness definitions comprise two properties: proper completion and weak termination.
\emph{Proper completion} states that once a marking that has a token in the final marking is reached, it is actually the final marking.
For example, for the proper completion to hold in a WF-net, as soon as a token is produced in the final place, all other places should be empty.
Following the idea of transition-bordered WF-nets, identifiers should have a similar  property: once a collector consumes one or more identifiers, then no further tokens carrying those identifiers should persist in the marking obtained after the consumption.

\begin{definition}[Proper type completion]
\label{def:completion}
Given a type $\lambda\in\Lambda$, a marked \tpnid $(N, m_0)$ is called \emph{properly $\lambda$-completing} iff for all $t \in C_N(\lambda)$, bindings $\psi : \V\to\I$ and markings $m, m'\in\reachable{N}{m_0}$, if $\fire{m}{t,\psi}{m'}$, then for all identifiers $\cname{id} \in \rng{\restr{\psi}{\delvar{t}}} \cap \id{m}$ with $\type(\cname{id})=\lambda$, it holds that $\cname{id}\not\in\id{m'}$.\footnote{Here, we constrain $\psi$ to objects of type $\lambda$ that are  consumed.}
\end{definition}
Intuitively, from the perspective of a single identifier of type $\lambda$,
if a \tpnid $N$ that generated it is properly $\lambda$-completing,
then the points of consumption for this identifier are mutually exclusive (that is, it can be consumed from the net only by one of the collectors from $C_N(\lambda)$).

\medskip
As an example, consider \tpnid $N_{\mathit{rs}}$ in Fig.~\ref{fig:overallModel}.
For type $\mathit{customer}$, 
we have $C_{N_{\mathit{rs}}}(\mathit{customer}) = \{K, V\}$.
In the current -- empty -- marking,
transition $T$ is enabled with binding  $\psi = \{ z\mapsto \cname{c}\}$, which results in marking $m$ with $m(\textit{customer}) = [\cname{c}]$.
We can then create an offer by firing $G$ with binding $\psi = \{y\mapsto \cname{o}, z\mapsto \cname{c}\}$.
Next, transitions $H$, $J$, $L$ and $N$ can fire, all using the same binding,
producing marking $m'$ with $m'(p) = [\cname{o},\cname{c}]$, $m'(\mathit{customer}) =[\cname{c}]$ and $m'(q) = m'(r) = [\cname{c}]$.
Hence, transition $K$ is enabled with binding $\psi$.
However, firing $K$ with $\psi$ results in marking $m''$ with $m''(\mathit{customer}) = [\cname{c}]$.
Since for the proper type completion on type $\mathit{customer}$  we would like to achieve that all
tokens containing  $\cname{c}$ are removed,  $N_{\mathit{rs}}$ is not properly $\mathit{customer}$-completing.

\medskip
\emph{Weak termination} signifies that the final marking can be reached from any reachable marking.
Translated to identifiers, removing an identifier from a marking should always eventually be possible.

\begin{definition}[Weak type termination]
\label{def:termination}
Given a type $\lambda\in\Lambda$, a marked \tpnid $(N, m_0)$ is called \emph{weakly $\lambda$-terminating} iff
for every $m \in \reachable{N}{m_0}$ and identifier $\cname{id} \in I(\lambda)$ such that $\cname{id} \in \id{m}$, there exists a marking $m' \in \reachable{N}{m}$ with $\cname{id} \not\in \id{m'}$.
\end{definition}

\emph{Identifier soundness} combines the properties of proper type completion and weak type termination: the former ensures that as soon a collector fires for an identifier, the identifier is removed, whereas the latter ensures that it is always eventually possible to remove that identifier.

\begin{definition}[Identifier soundness]
\label{def:soundness}
A marked \tpnid $(N,m_0)$ is \emph{$\lambda$-sound} iff it is properly $\lambda$-completing and weakly $\lambda$-terminating.
It is \emph{identifier sound} iff it is $\lambda$-sound for every $\lambda\in \type_\Lambda(N)$.
\end{definition}

Two interesting observations can be made about the identifier soundness property. First, identifier soundness does not imply soundness in the classical sense: any classical net $N$ without types, i.e., $\type_\Lambda(N) = \emptyset$, is identifier sound, independently of the properties of $N$.
Second, identifier soundness implies depth-boundedness. In other words, if a marked \tpnid is identifier sound, it cannot accumulate infinitely many tokens carrying the same identifier.

\begin{lemma}\label{lemma:depth}
	If a \tpnid $(N, m_0)$ is identifier sound, then it is depth-bounded.
\end{lemma}

\begin{proof}
Suppose that $(N,m_0)$ is identifier sound, but not depth-bounded.
Then, at least for one place $p\in P$ and identifier $\vec{\cname{id}}\in\colset (p)$ of type $\vec\lambda$ there exists an infinite sequence of increasing markings $m_i$, all reachable in $(N,m_0)$, such that $m_i(p)(\vec{\cname{id}})<m_{i+1}(p)(\vec{\cname{id}})$.
Let $\lambda\in\vec\lambda$ and let $\cname{id}\in\vec{\cname{id}}$ be such that $\type(\cname{id})=\lambda$.
From the above assumption it follows that there are no such markings $m_i$ and $m_{i+1}$ in the infinite sequence of increasing markings for which it holds that $m_{i+1}\in\reachable{N}{m_i}$, $\cname{id}\in m_i(p)$ and $\cname{id}\not\in m_{i+1}(p)$.
Since $N$ is properly type completing, it must be possible to reach from $m_i$ a marking $m_i'$ (via some firing sequence $\sigma$) such that
$\fire{m_i'}{t,\psi}{m_i''}$, for a binding $\psi : \V\to\I$, $m_i''\in\reachable{N}{m_0}$, $\cname{id}\not\in\id{m_i''}$ and  $t \in C_N(\lambda)$.
Since marking $m_{i+1}$ contains at least one more $\cname{id}$, then for the same $t \in C_N(\lambda)$ we cannot apply the same reasoning from above. Specifically, we can reach a marking $m_{i+1}'$ from $m_{i+1}$ using the same firing sequence $\eta$ and, although it still holds that $\fire{m_{i+1}'}{t,\psi}{m_{i+1}''}$ (and $m_{i+1}''$ differs from  $m_{i}''$ by having one extra $\cname{id}$ in $p$), we have that $\cname{id}\in\id{m_i''}$. This contradicts the proper type completion. Hence, $(N,m_0)$ is depth-bounded.
\end{proof}

As identifier soundness relies on reachability, it is undecidable.
This also naturally follows from the fact that all non-trivial decision problems are undecidable for Petri nets in which tokens carry pairs of data values (taken from unordered domains) and in which element-wise equality comparisons are allowed over such pairs in transition guards~\cite{Lasota16}.

\tikzset{
	config/.style={
		circle,
        rounded corners=5pt,
        draw,
		very thick,
		minimum height=5mm,
		minimum width=5mm,
	},
    link/.style={
        -latex,
        thick,
    }
}

\begin{theorem}
\label{thm:id-soundness-undecidable}
Identifier soundness is undecidable for \tpnids.
\end{theorem}
\begin{proof}
We prove this result by reduction from the reachability problem for a 2-counter Minsky machine by following ideas of the proof of Theorem 4 in~\cite{GGMR22}.

\medskip
A 2-counter Minsky machine with two non-negative counters $c_1$ and $c_2$ is a finite sequence of numbered instructions $1:\mathtt{ins_1},\ldots, n:\mathtt{ins_n}$, where $\mathtt{ins_n}=\mathsf{HALT}$ and for every $1\leq i < n$ we have that $\mathtt{ins_i}$ has one of the following forms:\smallskip

\begin{compactitem}
	\item $\mathtt{inc}~~ c_j;\; \mathtt{goto}~~ k$
	\item $\mathtt{if}~~ c_j=0~~ \mathtt{then~~goto}~~ k ~~ \mathtt{else} ~~(\mathtt{dec}~~c_j;\; \mathtt{goto}~~ l)$
\end{compactitem}\smallskip

\noindent Here, $j\in\set{1,2}$ and $1\leq k,l\leq n$, and $\mathtt{inc}$ (resp., $\mathtt{dec}$) is an operation used to increment (resp., decrement) the content of counter $c_j$.
It is well-known that, for a Minksy $2$-counter machine that starts with both counters set to 0, checking whether it eventually reaches the instruction $\mathsf{HALT}$ is undecidable.

We then largely rely on the encoding of Minksy $2$-counter machines presented in~\cite{GGMR22}.
In a nutshell, that encoding shows how so called OA-nets (we rely on them in the proof of Proposition~\ref{prop:verification-safety}) can simulate an arbitrary $n$-counter Minsky machine.
Borrowing an idea from~\cite{Las16}, each counter is encoded using a ``ring gadget'', where counter value $m$ is represented via sets of $m+1$ linked pairs of identifiers $S=\set{(a_1,a_2),(a_2,a_3),\ldots,(a_{m+1},a_{1})}$ such that an identifier appears exactly twice in $S$. At the level of the net marking, there always must be only one ring. For more detail on this approach we refer to~\cite{GGMR22}.

Without loss of generality, we assume that the machine halts only when both $c_1$ and $c_2$ are zero. Notice that an arbitrary 2-counter machine can be transformed into a corresponding machine that only halts with counter zero by appending, at the end of the original machine, a final set of instructions that decrements both counters, finally halting when they both test to zero.

\medskip
Using the \tpnid components from Figure~\ref{fig:minsky}, we can construct a \tpnid faithfully simulating a $2$-counter Minsky machine. Counter operations are defined as in \cite{GGMR22}. The \tpnid has one special object type $\lambda$, which works as follows:\smallskip

\begin{compactitem}
\item a new instance for $\lambda$ can be only created when a black token is contained in the distinguished \emph{init} place;
\item the emission of an object for $\lambda$ consumes the black token from the \emph{init} place, and inserts it in the place $\p_{q_0}$ that corresponds to the first instruction of the 2-counter machine;
\item when such a 2-counter machine halts, such a black token is finally transferred into the place $\p_{q_n}$, which in turn enables the last collector transition for $\lambda$.\footnote{Notice that transitions of net components simulating instructions with $\mathtt{dec}$, according to Definition~\ref{def:notations}, are also collectors for $\lambda$.}
\end{compactitem}\smallskip

This implies that the \tpnid is $\lambda$-sound if and only if the 2-counter machine halts.
\end{proof}

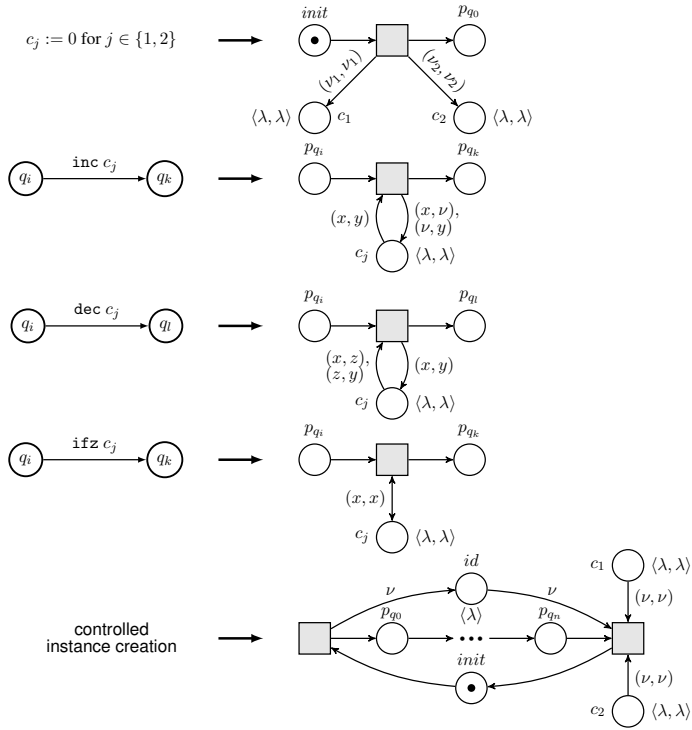
\begin{figure}[h!]
\vspace*{-2mm}
\centering
\resizebox{.59\textwidth}{!}{
\begin{tikzpicture}[->,>=stealth',auto,x=5mm,y=5mm,node distance=10mm and 10mm,thick]

	\node (a1src) {};
	\node[right=10mm of a1src] (a1tgt) {};
	\draw[ultra thick,-latex] (a1src) -- (a1tgt);

	\node[left= 5mm of a1src] (initialization) {$c_j:=0$ for $j \in \set{1,2}$};

	\node[place,tokens=1, right=5mm of a1tgt] (p0) {};
	\node[above=0mm of p0] {$\mathit{init}$};
	\node[transition,right=of p0] (init) {};
	\node[place,right=of init] (pqI) {};
	\node[above=0mm of pqI] {$p_{q_0}$};
	\node[place,below=of p0] (c1) {};
	\node[right=0mm of c1] {$c_1$};
	\node[left=0mm of c1] {$\tup{\lambda,\lambda}$};
	\node[place,below=of pqI] (cn) {};
	\node[left=0mm of cn] {$c_2$};
	\node[right=0mm of cn] {$\tup{\lambda,\lambda}$};
	
	\draw[->] (p0) -- (init);
	\draw[->] (init) -- (pqI);
	\draw[->] (init) -- node[above,sloped]{$(\nu_1,\nu_1)$} (c1);
	\draw[->] (init) -- node[above,sloped]{$(\nu_2,\nu_2)$} (cn);

	\node[below=28mm of a1src] (a2src) {};
	\node[right=10mm of a2src] (a2tgt) {};
	\draw[ultra thick,-latex] (a2src) -- (a2tgt);

	\node[config,left=5mm of a2src] (q2) {$q_k$};
	\node[config,left=23mm of q2] (q1) {$q_i$};
	\draw[link] (q1) --node {$\mathtt{inc}~c_j$} (q2);	
	
	\node[place,right=5mm of a2tgt] (pq1) {};
	\node[above=0mm of pq1] {$p_{q_i}$};
	\node[transition,right=of pq1] (t) {};
	\node[place,right=of t] (pq2) {};
	\node[above=0mm of pq2] {$p_{q_k}$};
	\node[place,below=of t] (pci) {};
	\node[left=0mm of pci] {$c_j$};
	\node[right=0mm of pci] {$\tup{\lambda, \lambda}$};
	\draw[->] (pq1) -- (t);
	\draw[->] (t) -- (pq2);
	\draw[->,out=120,in=-120]
		(pci)
		edge node[left] {$(x,y)$}
		(t);
	\draw[->,out=-60,in=60]
		(t)
		edge node[right] {$\begin{array}{@{}l@{}}(x,\nu),\\[-6pt](\nu,y)\end{array}$}
		(pci);

	\node[below=30mm of a2src] (a3src) {};
	\node[right=10mm of a3src] (a3tgt) {};
	\draw[ultra thick,-latex] (a3src) -- (a3tgt);

	\node[config,left=5mm of a3src] (q2) {$q_l$};
	\node[config,left=23mm of q2] (q1) {$q_i$};
	\draw[link] (q1) --node {$\mathtt{dec}~c_j$} (q2);	
	
	\node[place,right=5mm of a3tgt] (pq1) {};
	\node[above=0mm of pq1] {$p_{q_i}$};
	\node[transition,right=of pq1] (t) {};
	\node[place,right=of t] (pq2) {};
	\node[above=0mm of pq2] {$p_{q_l}$};
	\node[place,below=of t] (pci) {};
	\node[left=0mm of pci] {$c_j$};
	\node[right=0mm of pci] {$\tup{\lambda, \lambda}$};
	\draw[->] (pq1) -- (t);
	\draw[->] (t) -- (pq2);
	\draw[->,out=120,in=-120]
		(pci)
		edge node[left] {$\begin{array}{@{}l@{}}(x,z),\\[-6pt](z,y)\end{array}$}
		(t);
	\draw[->,out=-60,in=60]
		(t)
		edge node[right] {$(x,y)$}
		(pci);
				
	\node[below=27mm of a3src] (a4src) {};
	\node[right=10mm of a4src] (a4tgt) {};
	\draw[ultra thick,-latex] (a4src) -- (a4tgt);

	\node[config,left=5mm of a4src] (q2) {$q_k$};
	\node[config,left=23mm of q2] (q1) {$q_i$};
	\draw[link] (q1) --node {$\mathtt{ifz}~c_j$} (q2);	
	
	\node[place,right=5mm of a4tgt] (pq1) {};
	\node[above=0mm of pq1] {$p_{q_i}$};
	\node[transition,right=of pq1] (t) {};
	\node[place,right=of t] (pq2) {};
	\node[above=0mm of pq2] {$p_{q_k}$};
	\node[place,below=of t] (pci) {};
	\node[left=0mm of pci] {$c_j$};
	\node[right=0mm of pci] {$\tup{\lambda, \lambda}$};
	\draw[->] (pq1) -- (t);
	\draw[->] (t) -- (pq2);
	\draw[<->]
		(pci)
		edge node[left] {$(x,x)$}
		(t);
		
	\node[below=37mm of a4src] (a5src) {};
	\node[right=10mm of a5src] (a5tgt) {};
	\draw[ultra thick,-latex] (a5src) -- (a5tgt);
	
	\node[left=5mm of a5src] (q2) {$\begin{array}{@{}c@{}}\textsf{controlled}\\[-6pt]\textsf{instance creation}\end{array}$};
	\node[transition,right=5mm of a5tgt] (t1) {};
	\node[place,right=of t1] (pq0) {};
	\node[above=-1mm of pq0] {$p_{q_0}$};

	\node[right=of pq0] (dots) {\tiny{$\bullet \bullet \bullet$}};
	\node[place,right=of dots] (pqn) {};
	\node[above=-1mm of pqn] {$p_{q_n}$};
	\node[transition,right=of pqn] (t2) {};
	\node[place,above=9mm of t2] (c1){};
	\node[left=0mm of c1] {$c_1$};
	\node[right=0mm of c1] {$\tup{\lambda, \lambda}$};

	\node[place,below=9mm of t2] (c2) {};
	\node[left=0mm of c2] {$c_2$};
	\node[right=0mm of c2] {$\tup{\lambda, \lambda}$};

	\node[place,above=5.5mm of dots] (p1) {};
	\node[above=0mm of p1] {$id$};
	\node[below=-1mm of p1] {$\tup{\lambda}$};

	\node[place,below=5.5mm of dots,tokens=1] (init) {};
	\node[above=0mm of init] {$init$};
	\draw[->] (t1) -- (pq0);
	\draw[->] (pq0) -- ($(dots)-(0.8,0)$);
	\draw[->] ($(dots)+(0.8,0)$) -- (pqn);
	\draw[->] (pqn) -- (t2);
	\draw[->] (init) edge[bend left=12] (t1);
	\draw[->] (t2) edge[bend left=13] (init);
	\draw[->] (t1) edge[bend left=13] node[above] {$\nu$}  (p1);
	\draw[->] (p1) edge[bend left=13] node[above] {$\nu$}  (t2);		
	\draw[->] (c1) edge node[right,pos=.4] {$(\nu,\nu)$}  (t2);		
	\draw[->] (c2) edge node[right,pos=.4] {$(\nu,\nu)$}  (t2);		
\end{tikzpicture}
}
	\caption{Simulation of a Minksy $2$-counter machine via \tpnids. Here, $q_i$, $q_k$ and $q_l$ correspond to control states of the machine.}
	\label{fig:minsky}
\end{figure}

The above theorem shows that the identifier soundness is already undecidable for nets carrying identifier tuples of size $2$.
One may wonder whether the same result holds for \tpnids with singleton identifiers only.
To obtain this result one could, for example, study how the identifier soundness in this particular case relates to the notion of dynamic soundness -- an undecidable property of $\nu$-Petri nets studied in~\cite{MaV11}.

\begin{figure}[!h]
\vspace*{-2mm}
\centering
\begin{subfigure}{0.4\textwidth}
	\includegraphics[scale=.66]{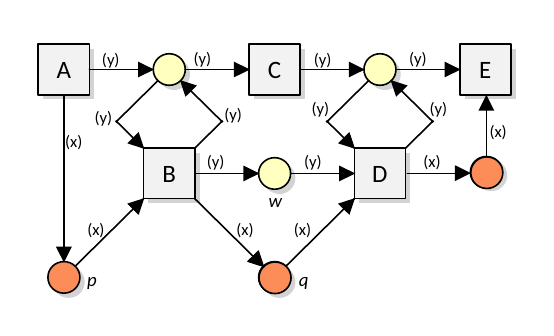}\vspace*{-1mm}
	\caption{$N_1$}\label{fig:projectionNotSufficientN1}
\end{subfigure}
\begin{subfigure}{0.4\textwidth}
	\includegraphics[scale=.66]{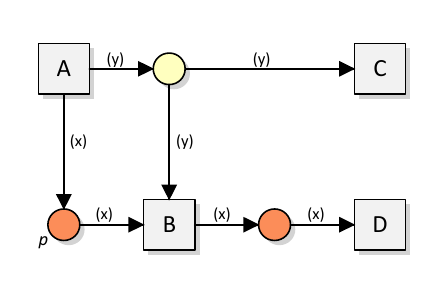}\vspace*{-1mm}
	\caption{$N_2$}\label{fig:projectionNotSufficientN2}
\end{subfigure}\vspace*{-1mm}
\caption{Two \tpnid{}s. Net $N_1$ is identifier sound, whereas net $N_2$ is not identifier sound.}\label{fig:projectionNotSufficient}
\end{figure}

The underlying idea of identifier soundness is that each type projection should behave well, i.e., each type projection should be sound.
Consider in \tpnid $N_{\mathit{rs}}$ of Fig.~\ref{fig:overallModel} and its $\mathit{customer}$-typed projection from Fig.~\ref{fig:cust-projection}.
The life cycle starts with transition $T$.
Transitions $K$ and~$V$ are two transitions that may remove the last reference to a $\mathit{customer}$.
Soundness of a transition-bordered WF-net would require that firing transition $K$ or transition $V$ would result in the final marking.
However, this is not necessarily the case. Consider the firing of transitions $T$, $G$ and $K$. Then, the token in place \emph{customer} remains, while the final transition $K$ already fired.
Hence, the \emph{customer} life cycle is not sound.
This raises the question whether we may conclude from this observation that $N_{\mathit{rs}}$ is not identifier sound.
Unfortunately, identifier soundness is not compositional, i.e., identifier soundness does not imply soundness of the type projections, and vice versa.
Consider the example \tpnids in Fig.~\ref{fig:projectionNotSufficient}.
The first net, $N_1$, is identifier sound.
However, taking the $\type_{\V}(y)$ projection of $N_1$ results in a bordered transition WF-net that is unsound: if transition $D$ fires fewer times than transition $B$, tokens will remain in place $w$.
The reverse is false as well.
Consider net $N_2$ in Fig.~\ref{fig:projectionNotSufficient}.
Each of the two projections are sound.
However, the resulting net reaches a deadlock after firing transitions $A$ and $C$ as a token generated by $A$ remains in place $p$. Hence, $N_2$ is not weakly $\type_{\V}(x)$-completing.

\begin{theorem}
	Let $\lambda\in \Lambda$ be a type, and let $(N, m_0)$ be a marked \tpnid.
	Then:
	\begin{enumerate}
\itemsep=0.85pt
		\item identifier soundness of $(N,m_0)$ does not imply soundness of $(\pi_\lambda(N),\pi_\lambda^N(m))$, and
		\item soundness of $(\pi_\lambda(N), \pi_\lambda^N(m))$ does not imply that $(N,m_0)$ is identifier sound.
	\end{enumerate}
\end{theorem}
\begin{proof}
We prove both statements by contradiction.
For the first statement, consider \tpnid $N_1$ depicted in Fig.~\ref{fig:projectionNotSufficientN1}. Though $N_1$ is identifier sound, its $\type_{\V}(y)$-projection is not sound.
Similarly, the $\type_{\V}(x)$-projection and $\type_{\V}(y)$-projection of $N_2$, depicted in Fig.~\ref{fig:projectionNotSufficientN2}, are sound, but $N_2$ is not identifier sound, as $N_2$ is not weakly $\type_{\V}(x)$-completing.
\end{proof}

Consequently, compositional verification of soundness of each of the projections is not sufficient to conclude anything about identifier soundness of the complete net, and vice versa.

In general, weak bisimulation does not guarantee identifier soundness, as it does not impose any relation on the identifiers in the nets.
However, if the bisimulation relation takes into account and preserves identifiers, then identifier soundness is preserved.
We formally demonstrate this property below.

\begin{lemma}[Weak bisimulation preserves proper type completion]\label{lm:weakbisimpreservespropertypecompletion}
	Let $(N_1,m_0^1)$ and $(N_2,m_0^2)$ be two marked \tpnids.
	Let $\lambda \in \Lambda$ be some type such that $C_{N_1}(\lambda) = C_{N_2}(\lambda)$, and let $Q \subseteq \mathbb{M}(N_1) \times \mathbb{M}(N_2)$ be a relation such that $I(\lambda) \cap \id{m_1} = I(\lambda) \cap \id{m_2}$ for all $(m_1,m_2) \in Q$ and
	$\tsys{N_1, m_0^1} \approx_Q \tsys{N_2, m_0^2}$.
	Then $N_1$ is properly $\lambda$-completing iff $N_2$ is properly $\lambda$-completing.
\end{lemma}
\begin{proof}
($\Rightarrow$) Suppose $N_1$ is properly $\lambda$-completing.
We need to show that $N_2$ is properly $\lambda$-completing.
Let $t \in C_{N_2}(\lambda)$ be a transition, let $\psi : \V\to\I$ be a binding and let $m_2, m'_2 \in \reachable{N_2}{m_0^2}$, such that $\fire{(N_2,m_2)}{t,\psi}{(N_2,m'_2)}$.
Let $\cname{id} \in \rng{\restr{\psi}{\delvar{t}}} \cap \id{m_2}$ with $\type(\cname{id})=\lambda$.
As $Q$ is a weak bisimulation relation, some markings $m_1,m'_1 \in \mathbb{M}{(N_1)}$ exist such that $(m_1,m_2) \in Q$ and $\fire{(N_1,m_1)}{t,\psi}{(N_1,m'_1)}$.
Then, by the definition of $Q$, it must be that $\cname{id} \in \rng{\restr{\psi}{\delvar{t}}} \cap \id{m_1}$.
As $C_{N_1}(\lambda) = C_{N_2}(\lambda)$ and $N_1$ is properly $\lambda$-completing, $\cname{id} \not\in \id{m'_1}$.
Since $Q$ is a weak bisimulation relation, $(m'_1,m'_2) \in Q$.
Thus, $\cname{id}\not\in\id{m'_2}$.
Hence, $N_2$ is properly $\lambda$-completing.

\smallskip\noindent
($\Leftarrow$) Follows from the commutativity of weak bisimulation.
\end{proof}

\begin{lemma}[Weak bisimulation preserves weak type termination]\label{lm:weakbisimpreservesweaktypetermination}
Let $(N_1,m_0^1)$ and $(N_2,m_0^2)$ be two marked \tpnids.
Let $\lambda \in \Lambda$ be some type, and let $Q \subseteq \mathbb{M}(N_1) \times \mathbb{M}(N_2)$ be a relation such that $I(\lambda) \cap \id{m_1} = I(\lambda) \cap \id{m_2}$ for all $(m_1,m_2) \in Q$ and
$\tsys{N_1, m_0^1} \approx_Q \tsys{N_2, m_0^2}$.
Then $N_1$ is weakly $\lambda$-terminating iff $N_2$ is weakly $\lambda$-terminating.
\end{lemma}
\begin{proof}
($\Rightarrow$) Suppose $N_1$ is weakly $\lambda$-terminating.
We need to show that $N_2$ is weakly $\lambda$-terminating.
Let $m_2 \in \reachable{N_2}{m_0^2}$ be some reachable marking and $\cname{id} \in I(\lambda)$ such that $\cname{id} \in \id{m_2}$.
As $Q$ is a weak bisimulation relation, some marking $m_1\in \mathbb{M}(N_1)$ exists with $(m_1,m_2) \in Q$.
Then, by the definition of $Q$, $\cname{id} \in \id{m_1}$.
As $N_1$ is weakly $\lambda$-terminating, a marking $m'_1 \in \reachable{N_1}{m_1}$ and firing sequence $\eta$ exist such that $\fire{(N_1,m_1)}{\eta}{(N_1,m'_1)}$ and $\cname{id} \not\in \id{m'_1}$.
As $Q$ is a weak bisimulation relation, a marking $m'_2 \in \mathbb{M}(N_2)$ exists such that $(m'_1, m'_2) \in Q$ and $\fire{(N_2,m_2)}{\eta}{(N_2,m'_2)}$.
As $I(\lambda) \cap \id{m'_1}= I(\lambda) \cap \id{m'_2}$, we have that $\cname{id} \not \in \id{m'_2}$, which proves the statement.

\smallskip\noindent
($\Leftarrow$) Follows from the commutativity of weak bisimulation.
\end{proof}

The two above lemmas are combined together to prove the following result.
\begin{theorem}[Weak bisimulation preserves $\lambda$-soundness]\label{thm:weakbisimpreservesidentifiersoundness}
Let $(N_1,m_0^1)$ and $(N_2,m_0^2)$ be two marked \tpnids.
Let $\lambda \in \Lambda$ be some type such that $C_{N_1}(\lambda) = C_{N_2}(\lambda)$, and let $Q \subseteq \mathbb{M}(N_1) \times \mathbb{M}(N_2)$ be a relation such that $I(\lambda) \cap \id{m_1} = I(\lambda) \cap \id{m_2}$ for all $(m_1,m_2) \in Q$ and
$\tsys{N_1, m_0^1} \approx_Q \tsys{N_2, m_0^2}$.
Then $N_1$ is  $\lambda$ sound iff $N_2$ is $\lambda$ sound.
\end{theorem}
\begin{proof}
Follows directly from Lm~\ref{lm:weakbisimpreservespropertypecompletion} and Lm~\ref{lm:weakbisimpreservesweaktypetermination}.
\end{proof}

\subsection{Towards verification of logical criteria}%
\label{sec:verification}

We now describe how existing results on the verification of safety and temporal properties over variants of Petri nets \cite{MontaliR16,GGMR22} and transition systems operating over relational structures \cite{CalvaneseGMP18,CalvaneseGMP22} can be lifted to the case of \emph{bounded} \tpnids. Boundedness is a sufficient requirement to make the verification of such properties decidable.

We start by considering safety checking of \tpnids, considering the recent results presented in~\cite{GGMR22}. A safety property is a property that must hold globally, that is, in every marking of the net.
Such a property is usually checked by formulating its unsafety dual and verifying whether a marking satisfying that unsafety property is reachable.

\begin{definition}[Unsafety property]
\label{def:property}
An \emph{unsafety property} over a \tpnid $N$
is a formula $\exists y_1,\ldots, y_k. \psi(y_1,\ldots, y_k)$, where $\psi$ is defined by the following grammar:
\[\psi::=p(x_1\cdots x_n)\geq c\,|\,x=y\,| x\neq y\,|\,p\geq c\,|\,\psi\land\psi,\]
Here:
\begin{inparaenum}[\it (i)]
\item $p$ is a place name from $N$,
\item  $x_i\in \V$ (for every $1\leq i\leq n$),
\item $x,y\in \V\cup\I$,
\item  $p\geq c$ and $p(x_1\cdots x_n)\geq c$ are atomic predicates defined over place markings with $c\in\naturals$.
\end{inparaenum}
\end{definition}
Given a marked \tpnid $(N,m_0)$, each atomic predicate is interpreted on all possible markings covering those from  $\reachable{N}{m_0}$. Like that, $p\geq c$ specifies that in place $p$ there are at least $c$ tokens, whereas $p(x_1\cdots x_n)\geq c$ indicates that in place $p$ there are at least $c$ tokens carrying an identifier vector that can valuate $x_1 \cdots x_n$.
We use elements from $x_1,\ldots, x_n$ as a filter selecting matching tokens in $p$, and use the variables from the same sequence to inspect different places by creating implicit joins between tokens stored therein. 
As it has been established in~\cite{GGMR22}, such unsafety properties can be used for expressing object-aware \emph{coverability} properties of \tpnids.

\smallskip
For example, as a property we may write that $\exists z,y. \mathit{created\_offer}(z,y)\geq 1 \land \mathit{customer}(y)\ge 1$ captures the (undesired) situation in which an offer has been made to customer $z$, but that customer is still available for receiving other offers.

\smallskip
The verification problem for checking unsafety properties is specified as follows: given an unsafety property $\psi$, a marked \tpnid $(N,m_0)$ is \emph{unsafe} w.r.t. $\psi$ if
$(N,m_0)$ can reach a marking in which $\psi$ holds.
If this is not the case, then we say that $(N, m_0)$ is \emph{safe} w.r.t.~$\psi$.
We show below that such verification problem is actually decidable.
\begin{proposition}\label{prop:verification-safety}
Verification of unsafety properties over bounded, marked \tpnids is decidable.
\end{proposition}
\begin{proof}
To prove this statement, we introduce the class of OA-nets~\cite{GGMR22} and establish their relation to \tpnids. To do so, we provide a modified version of Definition 1 from~\cite{GGMR22}.
Essentially, an OA-net is a tuple $(P,T,F_{in}, F_{out},\mathtt{color},\mathtt{guard})$, where:
\begin{itemize}
\itemsep=0.9pt
\item $P$ and $T$ are finite sets of places and transitions, s.t. $P\cap T =\emptyset$;
\item $\mathtt{color}:P\rightarrow\Lambda^*$ is a place typing function;
\item $F_{in} : P \times T \to \mult{\Omega_\V}$ is an input flow s.t. $\type_{\V}(F_{in}(p,t))=\mathtt{color}(p)$ for every $(p,t)\in P\times T$;\footnote{We denote by $\Omega_A$ the set of all possible  tuples of variables and identifiers over a set $A$.}\footnote{Without loss of generality, we assume that $\type_\V$ naturally extends to cartesian products.}
\item $F_{out} : T \times P \to \mult{\Omega_{\mathcal{X}}}$ is an output flow s.t. $\type_{\V}(F_{out}(t,p))=\mathtt{color}(p)$ for every $(t,p)\in T\times P$ and $\mathcal{X}=\V\uplus \hat\V$, where $\hat\V$ is the countably infinite set of fresh variables (i.e., variables used to provide fresh inputs only);\footnote{The original definition from~\cite{GGMR22} also allows for constants from $\I$ to appear in the output flow vectors. However, for simplicity's sake, we removed it here form the definition of $F_{out}$.}
\item $\mathtt{guard}:T\not\to \Phi$ is a partial guard assignment function, s.t. $\Phi$ is a set of conditions $\phi ::= y_1 = y_2 \,|\, y_1\neq y_2 \,|\, \phi\land\phi$, where $y_i\in \V\cup\I$, and
for each $\phi=\mathtt{guard}(t)$ and $t\in T$ it holds that $\var{\phi}\subseteq\invar{t}$.\footnote{Here, $\var{\phi}$ provides the set of all variables in $\phi$.}
\end{itemize}
It is easy to see that a \tpnid is an OA-net without guards. Moreover, using the notions from Definition~\ref{def:notations}, we get that $\newvar{t}\subset\hat\V$, for each $t\in T$.

Given the relation between \tpnids and OA-nets, the proof of the decidability follows immediately from Theorem 3 in~\cite{GGMR22}.
\end{proof}
\smallskip
One may wonder whether it is possible to go beyond safety and check
other properties expressible on top of \tpnids using more sophisticated temporal logics.
We answer to this question affirmatively, by proving that bounded \tpnids
induce transition systems that enjoy the so-called \emph{genericity} property~\cite{CalvaneseGMP18}. Such property, combined with \tpnid boundedness (which corresponds to the notion of state-boundedness used in~\cite{CalvaneseGMP18}) guarantees decidability of model checking for sophisticated variants of first-order temporal logics~\cite{HaririCGDM13,CalvaneseGMP18,CalvaneseGMP22}.

\smallskip
In generic transition systems, the behaviour does not depend on the actual
data present in the states, but only on how they relate to each other. This essentially reconstructs the well-known notion of genericity in databases, which expresses that isomorphic databases return the same answers to the same query, modulo renaming of individuals~\cite{AbiteboulHV95}.

We lift the notion of genericity to the case of transition systems induced by
\tpnids. To proceed, we first need to define a suitable notion of isomorphism
between two markings of a net.
\begin{definition}[Marking isomorphism]
\label{def:marking-isomorphism}
Given a \tpnid $N$ and two markings $m_1,m_2 \in \mathbb{M}(N)$,
we say that $m_1$ and $m_2$ are \emph{isomorphic},
written $m_1\sim_h m_2$,
if there exits a bijection (called \emph{isomorphism}) $h:Id(m_1)\to Id(m_2)$ such that for every $p\in \places$, it holds that $(\cname{id_1}\cdots\cname{id_n})^k\in m_1$ iff $(h(\cname{id_1})\cdots h(\cname{id_n}))^k\in m_2$, for every $k \in \naturals$.
\end{definition}
Intuitively, two markings are called isomorphic if they have the same amount of tokens and tokens correspond to each other modulo consistent renaming of identifiers.

\begin{definition}[Generic transition system]
\label{def:generic-ts}
Let $\tsys{} = (S,A,s_0, \to)$ 
be the transition system induced by some \tpnid. Then $\tsys{}$ is \emph{generic} if
for every markings $m_1,m_1',m_2\in S$ and every bijection $h:\I\to \I$,
if $m_1\sim_h m_2$ and $\fire{m_1}{t,\psi}{m_1'}$
(for some $t\in\transitions$ and binding $\psi:\V\to\I$),
then there exists $m_2'\in S$ and $\psi':\V\to\I$ such that
$\fire{m_2}{t,\psi'}{m_2'}$, $m_1'\sim_h m_2'$ and $\psi(v)=h(\psi'(v))$, for every $v\in\V$.
\end{definition}
As one can see from the definition, genericity requires that if two marking are isomorphic, then they induce the same transitions modulo isomorphism (i.e., the transition names are the same, and the variable assignments are equivalent modulo renaming). This implies that they induce isomorphic successors.

\begin{remark}
Let $(N,m_0)$ be a marked \tpnid. Then its induced transition system $\tsys{N,m_0}$ is generic.
\end{remark}
The above result can be easily shown by considering the transition system construction
described in Definition~\ref{def:induced-ts}, by considering isomorphism between its markings given by a simple renaming function and checking the conditions of Definition~\ref{def:generic-ts}.

\medskip
As it has been demonstrated in \cite{CalvaneseGMP18,CalvaneseGMP22},
model checking of sophisticated first-order variants of $\mu$-calculus and LTL
becomes decidable for so-called state-bounded generic transition systems.
Since the transition systems induced by \tpnids are also generic, and boundedness of a \tpnid correspond to state-boundedness of its induced transition system, we directly obtain decidability of
model checking for the same logics considered there, extended in our case over atomic predicates of the form $p\odot c$ and $p(x_1\cdots x_n)\odot c$, where $\odot\in\set{<,>,=,\leq,\geq}$, in the style of
similar logics introduced in~\cite{MontaliR16}.
We shall refer to such logics as $\mu\L{\text{-}FO}^{\text{PNID}}$ and $LTL\text{-}FO_p^{\text{PNID}}$, but in this work we omit their \linebreak definition.
\begin{theorem}
\label{thm:model-checking}
Model checking of $\mu\L{\text{-}FO}^{\text{PNID}}$ and $LTL\text{-}FO_p^{\text{PNID}}$ formulae is decidable for bounded, marked \tpnids.
\end{theorem}

Whereas the boundedness condition may appear restrictive at the first sight, we recall that according to Definition~\ref{def:boundedness}, a \tpnid is $k$-bounded if every marking reachable from the initial one does not assign more than $k$ tokens to every place of the net.
This however does not impede the net at hand from reaching infinitely many states as tokens may, along its run, carry infinitely many distinct objects. Notice also that this condition is less restrictive than identifier boundedness (which essentially forces the identifier domain to be finite) made use of in~\cite{PolyvyanyyWOB19}.
In Section~\ref{sec:soundness:with:resources} we  discuss a class of \tpnids for which boundedness still allows to explore lifecycles of potentially infinitely many objects.

\section{Correctness by construction}
\label{sec:soundness:by:construction}

As shown in the previous section, identifier soundness is undecidable.
However, we are still interested in ensuring correctness criteria over the modeled system.
In this section, we propose a structural approach to taming the undecidability
and study  sub-classes of \tpnids that are identifier sound by construction.

\subsection{EC-closed workflow nets}
\label{sec:ec-wf-closure}

\begin{figure}[t]
	\centering
	\includegraphics[scale=.8]{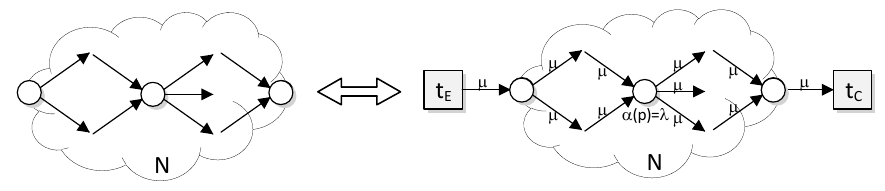}
	\caption{EC-closure of a WF-net $N$.}\label{fig:EC-closure}
\end{figure}

WF-nets are widely used to model business processes.
The initial place of the WF-net signifies the start of a \emph{case}, the final place represents the goal state, i.e., the process case completion. A firing sequence from initial state to final state represents the activities that are performed for a single case.
Thus, a WF-net describes all possible sequences of a single case.
Process engines, like Yasper~\cite{HeeOPSW06} simulate the execution of multiple cases in parallel by coloring the tokens with the case identifier (a similar idea is used for resource-constrained WF-net variants of $\nu$-PNs in~\cite{MonR16}).
In other words, they label each place with a case type, and inscribe each arc with a variable.
To execute it, the WF-net is closed with an emitter and a collector, as shown in Figure~\ref{fig:EC-closure}.
We generalize this idea to any place label, i.e., any finite sequence of types may be used to represent a case. We use the technical results obtained in this section further on in Section~\ref{sec:wf-refinement}.

\begin{definition}[EC-Closure]
\label{def:ec-closure}
Given a WF-net $N$, place type $\vec{\lambda} \in \Lambda^*$ and a variable vector $\vec{v}\in\V^*$ such that $\type_{\V}(\vec{v}) = \vec{\lambda}$,
its \emph{EC-closure} is a \tpnid $\mathcal{W}(N, \vec{\lambda}, \vec v) = (\places_N, \transitions_N \cup \{t_E, t_C\}, \flow_N \cup \{(t_E, \inp), (\outp, t_C)\}, \alpha, \beta)$, with:
\begin{itemize}
\itemsep=0.9pt
	\item $\alpha(p) =\vec{\lambda}$ for all places $p \in \places_N$;
	\item $\beta(f) = \vec{v}\,{}^{W(f)}$ for all flows $f \in \flow_N$, and $\beta((t_e,\inp))=\beta((\outp, t_c)) = [\,\vec{v}\,]$.
\end{itemize}
\end{definition}

The EC-closure of a WF-net describes all cases that run simultaneously at any given time.
In other words, any reachable marking of the EC-closure is the ``sum'' of all simultaneous cases.
Lemma~\ref{lemma:weak-bisim-wrappednet} formalizes this idea by establishing  weak bisimulation  between the projection on a single case and the original net. 

\begin{lemma}[Weak bisimulation for each identifier]\label{lemma:weak-bisim-wrappednet}
Let $N$ be a WF-net, $\vec\lambda \in \Lambda^*$ be a place type and $\vec{v} \in \V^*$ be a variable vector s.t. $\type_{\V}(\vec{v}) = \vec\lambda$.
Then, for any $\vec{\cname{id}} \in \I^{|\vec\lambda|}$,  $\rho_{r}(\Gamma_{\mathcal{W}(N,\vec\lambda,\vec{v}),\emptyset}) \approx \Gamma_{N,[\inp]}$, where $r$ in renaming $\rho_{r}$ is such that $r(t,\psi) = t$, if $\psi(\vec{v}) = \vec{\cname{id}}$, and $r((t,\psi))=\tau$, otherwise.
\end{lemma}
\begin{proof}
Let $N' = \mathcal{W}(N,\vec\lambda,\vec{v})$.
Define $R = \{(M, m) \mid \forall p \in P: M(p)(\vec{\cname{id}}) = m(p) \}$. We need to show that $R$ is a weak bisimulation.

\smallskip\noindent
($\Rightarrow$) Let $M, M'$ and $m$ be such markings that $(M,m) \in R$ and $\fire{(N',M)}{t,\psi}{(N',M')}$, with $t \in \transitions$ and $\psi : \V \to \I$.
By Definition~\ref{def:ec-closure}, $\psi(\vec{v}) = \vec{\cname{u}}$, for some $\vec{\cname{u}}\in\I^{|\vec\lambda|}$.
From the firing rule, we obtain $M'(p) + [\vec{\cname{u}}^{W((p,t))}] = M(p) + [\vec{\cname{u}}^{W((t,p))}]$, for any $p \in P$.
If $\vec{\cname{u}} \neq \vec{\cname {id}}$, then $r(t,\psi) = \tau$, and $(M', m) \in R$.
If $\vec{\cname{u}} = \vec{\cname {id}}$,  there exists such marking $m'$ that $\fire{m}{t}{m'}$ (since $m(p) = M(p)( \vec{\cname {id}})$ and thus $m(p)\geq W((p,t))$) and $m'(p) + W((p,t))= M(p)( \vec{\cname {id}}) + W((t,p))$.
Then, by construction, $m'(p)=M'(p)(\vec{\cname {id}})$ and $(M',m') \in R$.

\smallskip\noindent
($\Leftarrow$)
Let $M$, $m$, and $m'$ be markings that $(M,m) \in R$ and $\fire{(N,m)}{t}{(N,m')}$ with $t \in T$.
We choose binding $\psi$ such that $\psi(\vec v) = \vec{\cname{id}}$.
Then $\rho_\psi(\beta(p,t)) = [\vec{\cname{id}}^{W((p,t))}] \leq M(p)$, since $W((p,t))\leq m(p) = M(p)(\vec{\cname{id}})$.
Thus, a marking $M'$ exists such that  $\fire{(N',M)}{t,\psi}{(N',M')}$.
Then $M'(p) + [\vec{\cname{id}}^{W((p,t))}] = M(p) + [\vec{\cname{id}}^{W((t,p))}]$.
Hence, $M'(p)(id) = m'(p)$ and thus $(M',m') \in R$.
\end{proof}

A natural consequence of this weak bisimulation result is that any EC-closure of a WF-net is identifier sound if the underlying WF-net is sound.

\begin{theorem}
Given a WF-Net $N$, if $N$ is sound, then $\mathcal{W}(N, \vec \lambda, \vec v)$ is identifier sound and live, for any place type $\vec \lambda \in \Lambda^*$ and variable vector $\vec{v} \in \V^*$ with $type(\vec{v}) = \vec \lambda$.
\end{theorem}
\begin{proof}
Let $N' = \mathcal{W}(N, \vec\lambda, \vec v) = (\places, \transitions, \flow, \alpha, \beta)$.
By definition of $\mathcal{W}$, $\delvar{t} = \emptyset$ for any transition $t \in \transitions\setminus\set{t_C}$.
Hence, only transition $t_C$ can remove identifiers, and thus, by construction, $\W$ is properly type completing on all $\lambda \in \vec\lambda$.

\medskip
Next, we need to show that $N'$ is weakly type terminating for all types $\lambda \in \vec\lambda$.
Let $M \in \reachable{N'}{\emptyset}$, with firing sequence $\eta \in (T \times (\V\to\I))^*$, i.e., $\fire{(N',\emptyset)}{\eta}{(N',m)}$.
Let $\vec{\cname{id}} \in \colset(p)$ such that $M(p)(\vec{\cname{id}}) > 0$ for some $p \in \places$.
We then construct a sequence $\omega$ by stripping the bindings from $\eta$ s.t. it contains only transitions of $T$. 
Using Lemma~\ref{lemma:weak-bisim-wrappednet}, we obtain a marking $m \in \reachable{N}{[in]}$ such that $\fire{[\inp]}{\psi}{m}$ and $m(p) = M(p)(\vec{\cname{id}})$.
Since $N$ is sound, there exists a firing sequence $\omega'$ such that $\fire{m}{\omega'}{[\outp]}$.
Again by Lemma~\ref{lemma:weak-bisim-wrappednet}, a firing sequence $\eta'$ exists such that $\fire{M}{\eta'}{M'}$ and $(M', [\outp]) \in \restr{\reachable{\W}{\emptyset}}{\vec{\cname{id}}}$, where $\restr{\reachable{\W}{\emptyset}}{\vec{\cname{id}}}$ is the set of all reachable markings containing $\vec{\cname{id}}$.
Hence, if $M'(p)(\vec{\cname{id}}) > 0$, then $p = \outp$. Thus, transition $t_C$ is enabled with some binding $\psi$ such that $\psi(\vec{v}) = \vec{\cname{id}}$, and a marking $M''$ exists such that $\fire{M'}{t_c,\psi}{M''}$, which removes all identifiers in $\vec{\cname{id}}$ from $M'$.
Hence, $N'$ is identifier sound.

\medskip
As transition $t_e$ is always enabled and $N$ is quasi live, $\mathcal{W}(N,\vec{\lambda},\vec{v})$ is live.
\end{proof}

\subsection{Typed Jackson nets}\label{sec:typedJN}
A well-studied class of processes that guarantee soundness are block-structured nets.
Examples include Process Trees~\cite{Leemans13}, Refined Process Structure Trees~\cite{Weidlich11} and Jackson Nets~\cite{HeeHHPT09}.
Each of the techniques have a set of rules in common from which a class of nets can be constructed that guarantees properties like soundness.
In this section, we introduce Typed Jackson Nets (t-JNs), extending the ideas of Jackson Nets~\cite{HeeHHPT09,HeeSW13} to \tpnids, that guarantee both identifier soundness and liveness.
The six reduction rules presented by Murata in~\cite{Murata89} form the basis of this class of nets.
The rules for t-JNs are depicted in Figure~\ref{fig:constructionrulestJN}.

\begin{figure}[!ht]
\vspace*{-5mm}
  \centering
	\begin{subfigure}{.3\textwidth}
		\centering
		\includegraphics[width=\textwidth]{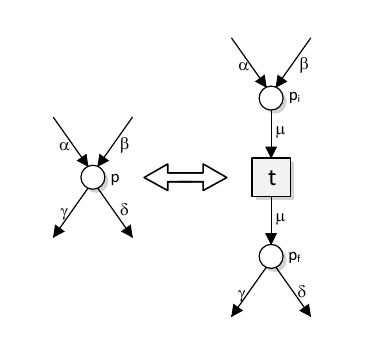}\vspace*{-1mm}
        \caption{Place Expansion}\label{fig:placeExpansion}
	\end{subfigure}
	\begin{subfigure}{.3\textwidth}
		\centering
		\includegraphics[width=\textwidth]{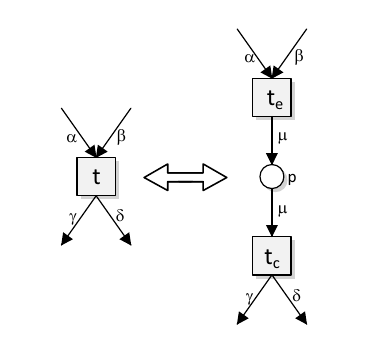}\vspace*{-1mm}
        \caption{Transition Expansion}\label{fig:transitionExpansion}
	\end{subfigure}
	\begin{subfigure}{.3\textwidth}
		\centering
		\includegraphics[width=\textwidth]{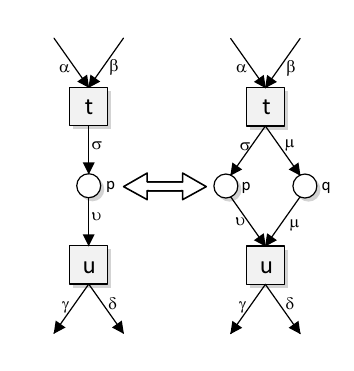}\vspace*{-1mm}
        \caption{Place Duplication}\label{fig:placeDuplication}
	\end{subfigure}
\vspace*{-5mm}
	\begin{subfigure}{.3\textwidth}
		\centering
		\includegraphics[width=\textwidth]{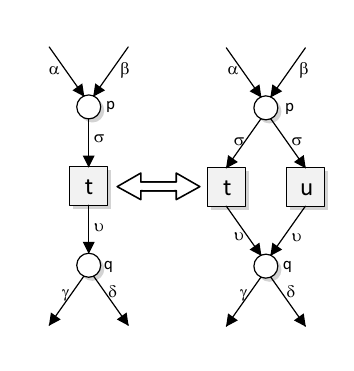}\vspace*{-3mm}
        \caption{Transition Duplication}\label{fig:transitionDuplication}
	\end{subfigure}
	\begin{subfigure}{.3\textwidth}
		\centering
		\includegraphics[width=\textwidth]{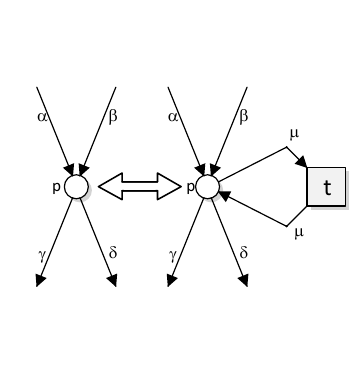}\vspace*{-3mm}
        \caption{Self-loop Addition}\label{fig:selfLoopAddition}
	\end{subfigure}
	\begin{subfigure}{.3\textwidth}
		\centering
		\includegraphics[width=\textwidth]{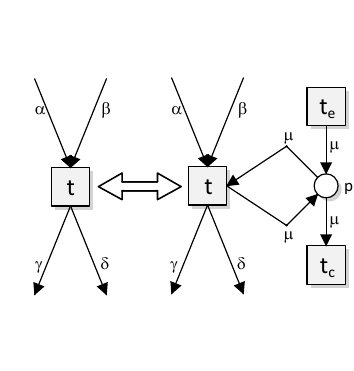}\vspace*{-3mm}
        \caption{Identifier Creation}\label{fig:identifierCreation}
	\end{subfigure}\vspace*{1mm}
\caption{Construction rules of the typed Jackson Nets.}\label{fig:constructionrulestJN}
\end{figure}

\newcounter{tjncount}
\newcommand{\ruletitle}[1]{\medskip	\noindent\textbf{\fbox{Rule \addtocounter{tjncount}{1}\thetjncount:} #1}}

\subsubsection{Place Expansion}
The first rule is based on \emph{fusion of a series of places}.
As shown in Figure~\ref{fig:placeExpansion}, a single place $p$ is replaced by two places $p_i$ and $p_f$ that are connected via transition $t$.
All transitions that originally produced in $p$, produce in $p_i$ in the place expansion, and similarly, the transitions that consumed from place $p$, now consume from place $p_f$.
In fact, transition $t$ can be seen as a transfer transition: it needs to move tokens from place $p_i$ to place $p_f$, before the original process can continue.
This is also reflected in the labeling of the places: both places have the same place type, and the input and output arc of transition $t$ are inscribed with the same variable vector $[ \,\vec\mu\,]$ that matches the type of place $p$.

\begin{definition}[Place expansion]
\label{def:place-expansion}
Let $(N,m)$ be a marked \tpnid with $N = (\places, \transitions, \flow, \alpha, \beta)$, $p \in \places$ be a place and $\vec\mu \in \V^*$ be a variable vector s.t. $\type_{\V}(\vec\mu) = \alpha(p)$.
The \emph{place expanded \tpnid} is defined by the relation $R_{p,\vec\mu}(N,m) = ((\places', \transitions',\flow', \alpha',\beta'),m')$, where:
\begin{itemize}
\itemsep=0.9pt
\item $\places' = (\places \setminus \{p\}) \cup \{p_i, p_f\}$ with $p_i, p_f \not\in \places$; and
$\transitions' = \transitions \cup \{t\}$ with $t \not\in \transitions$;
\item
$\flow' = (\flow \setminus ((\{p\} \times \post{p})\cup (\pre{p} \times \{p\}) )
\cup
(\pre{p} \times \{p_i\})
\cup \{(p_i,t),(t,p_f)\} \cup
(\{p_f\} \times \post{p})$;

\item 
$\alpha'(q) = \alpha(p)$, if $q \in \{p_i, p_f\}$, and $\alpha'(q) = \alpha(q)$, otherwise.
\item 	$\beta'(f)\! =\! [\vec\mu]$, if $f \in \{(p_i,t),(t,p_f)\}$, $\beta'((u,p_i))\!=\!\beta((u,p))$, if  $u \in \pre{p}$, $\beta'((p_f,u))\!=\!\beta((p,u))$, if $u \in \post{p}$, and $\beta'(f) = \beta(f)$, otherwise.
	\item $m'(q) = m(q)$ for all $q \in P\setminus\{p\}$, $m'(p_f) = 0$, and $m'(p_i) = m(p)$.
\end{itemize}
\end{definition}

Inscription $\vec\mu$ cannot alter the vector identifier on the tokens, as the type of $\vec\mu$ should correspond to both place types $\alpha(p)$ and $\alpha(q)$.
Hence, the transition is enabled with the same bindings as any other transition that consumes a token from place $p$, modulo variable renaming.
As such, transition $t$ only ``transfers'' tokens from place $p_i$ to place $p_f$.
Hence, as the next lemmas shows, place expansion yields a weakly bisimilar \tpnid and preserves identifier soundness.

\begin{lemma}\label{lm:bisimplaceexpansion}
Let  $(N,m_0)$ be a marked \tpnid with $N = (\places, \transitions, \flow, \alpha, \beta)$, $p\in \places$ be a place to expand and  $\vec\mu \in \V^*$ be a variable vector s.t. $\type_{\V}(\vec\mu) = \alpha(p)$. Then $\tsys{N,m_0} {\approx^{r}} \hide{\{t\}}(\tsys{R_{p,\vec\mu}(N,m_0)})$, with transition $t$ added by $R_{p,\vec\mu}$.
\end{lemma}
\begin{proof}
Let $(N',m_0') = R_{p,\mu}(N, m_0)$.
We define $Q \subseteq \mathbb{M}(N) \times \mathbb{M}(N')$ such that $(m,m') \in Q$ iff $m(q) = m'(q)$ for all places $q \in \places\setminus\{p\}$ and $m'(p_i) + m'(p_f) = m(p)$. Then $(m_0,m_0') \in Q$, hence the relation is rooted.

\smallskip \noindent ($\Rightarrow$)
Let $(m,m') \in Q$ and $\fire{(N,m)}{u,\psi}{(N,\bar{m})}$.
We need to show that there exists marking $\bar{m}'$ such that  $\smash{m' \xdasharrow{~(u,\psi)~} \bar{m}'}$ and $(\bar{m},\bar{m}') \in Q$.

\smallskip
Suppose $p \not\in\pre{u}$. Then $m'(q) = m(q)$ and  $m(q)\geq \rho_\psi(\beta((p,u)))$ (note that $\rho_\psi(\beta((p,u))) = \rho_\psi(\beta'((p,u)))$).
By the firing rule, a marking $\bar{m}'$ exists with $\fire{(N,m')}{u,\psi}{(N',\bar{m}')}$, $\bar{m}(q) = \bar{m}'(q)$ for all $q \in \places'$.
Thus, $(\bar{m},\bar{m}')\in Q$.
Suppose $p \in\pre{u}$. Then $\rho_{\psi}(\beta((p_f,u))) \leq m(p) = m'(p_i)  + m'(p_f)$.
If $\rho_\psi(\beta(p_f, u))) \leq m'(p_f)$, then transition $u$ is enabled, and a marking $\bar{m}'$ exists with $\fire{(N,m')}{u,\psi}{(N,\bar{m}')}$ and $(\bar{m},\bar{m}') \in Q$.

\smallskip
Otherwise, $\rho_\psi(\beta(p_f, u))) \leq m'(p_i)$.
Construct a binding $\psi'$ by letting $\psi'(\mu(i)) = \psi(\beta(p,u)(i))$, for all $1 \leq i \leq |\mu|$.
Then, $\rho_{\psi'}(\mu) = \rho_{\psi}(\beta(p,u))$, and transition $t$ is enabled with binding $\psi'$.
Hence, a marking $m''$ exists with $\fire{(N,m')}{t,\psi'}{(N,m'')}$ and $\rho_{\psi}(\beta((p',u))) \leq m''(p')$.
Then $(m,m'') \in Q$ and $t$ is labeled $\tau$ in $\hide{\{t\}}(R_{(p,\vec\mu)}(N))$.
Now, either transition $u$ is enabled, or transition $t$ is again enabled with binding $\psi'$.
In all cases, $\smash{m' \xdasharrow{~(t,\psi)~} \bar{m}'}$ and $(m', \bar{m}') \in Q$.

\smallskip \noindent
($\Leftarrow$)
Let $(m,m') \in Q$ and $\fire{(N',m')}{u,\psi}{(N',\bar{m}')}$.
We need to show that either a $\bar{m}$ exists such that $\fire{(N,m)}{u,\psi}{(N,\bar{m})}$ and $(\bar{m},\bar{m}')\in Q$ or
$u = \tau$ and $(m, \bar{m}') \in Q$.

\smallskip
Suppose $u = t$, i.e., $u$ is labeled $\tau$ in $\hide{\{t\}}(R_{(p,\vec\mu)}(N))$.
Then, $\pre{u} = p_i$ and $\post{u} = p_f$.
By the firing rule, $\bar{m}'(p_i) + \bar{m}'(p_f) = m(p_i) + m(p_f) = m(p)$. Hence $(m, \bar{m}') \in Q$.

\smallskip
If $u \neq t$, we need to show that there exists marking $\bar{m}$ such that  $\smash{m \xdasharrow{~(u,\psi)~} \bar{m}}$ and $(\bar{m},\bar{m}') \in Q$.
Let $q \in \pre{u}$.
If $q \neq p$, then $m(q) \leq \rho_{\psi}(\beta(p,u))$.
If $q = p$, then $m(q) \geq m(p_f)$ and thus, $m(q) \leq \rho_{\psi}(\beta(p,u))$.
Hence, transition $u$ is enabled in $m$ and a marking $\bar{m}$ exists such that $\fire{(N,m)}{t,\psi'}{(N,\bar{m})}$.
By the firing rule, we have
$\bar{m}(p) = m(p) - \rho_{\psi}(\beta(p,u)) + \rho_{\psi}(\beta(u,p)) = m'(p_i) + m'(p_f) - \rho_{\psi}(\beta(p_i,u)) - \rho_{\psi}(\beta(p_f,u)) + \rho_{\psi}(\beta(u,p_i)) + \rho_{\psi}(\beta(u,p_f)) = \bar{m}'(p_i) + \bar{m}'(p_f)$, since $\rho_{\psi}(\beta(p_i,u)) = \rho_{\psi}(\beta(u,p_f)= \emptyset$.
Hence, $(\bar{m}, \bar{m}')\in Q$, which proves the statement.
\end{proof}

\begin{lemma}
Let  $(N,m_0)$ be a marked \tpnid with $N = (\places, \transitions, \flow, \alpha, \beta)$, $p\in \places$ be a place to expand and  $\vec\mu \in \V^*$ be a variable vector s.t. $\type_{\V}(\vec\mu) = \alpha(p)$.
Then, $R_{p,\vec\mu}(N,m_0)$ is identifier sound iff $(N,m_0)$ is identifier sound.
\end{lemma}
\begin{proof}
Let $(N',m'_0) = R_{p,\vec\mu}(N,m_0)$.
Define $Q \subseteq \mathbb{M}(N) \times \mathbb{M}(N')$ such that $(m,m') \in Q$ iff $m(q) = m'(q)$ for all places $q \in \places\setminus\{p\}$ and $m'(p_i) + m'(p_f) = m(p)$, i.e., $Q$ is the bisimulation relation as defined in the previous lemma.
Then, the statement is a direct consequence of $\id{m} = \id{m'}$ for all $(m,m') \in Q$ and the bisimulation.
\end{proof}

\begin{lemma}\label{lm:placeexpansionpreservesidentifiersoundness}
	Let  $(N,m_0)$ be a marked \tpnid with $N = (\places, \transitions, \flow, \alpha, \beta)$, $p\in \places$ be a place to expand and  $\vec\mu \in \V^*$ be a variable vector.
	Then $(N,m_0)$ is identifier sound iff $R_{p,\vec{\mu}}(N,m_0)$ is identifier sound.
\end{lemma}
\begin{proof}
Let $(N',m'_0) = R_{p,\vec\mu}(N,m_0)$, and let $Q \subseteq \mathbb{M}(N) \times \mathbb{M}(N')$ be the bisimulation relation of Lm.~\ref{lm:bisimplaceexpansion}.
Let $t$ be the transition added by the place extension rule.
Then $t \not\in C_N(\lambda)$ for any type $\lambda \in \type_\Lambda(N)$.
As $\id{m_1} = \id{m_2}$ for all $(m_1,m_2) \in Q$, the statement directly follows from Thm.~\ref{thm:weakbisimpreservesidentifiersoundness}.
\end{proof}

\subsubsection{Transition Expansion}
The second rule is transition expansion, which corresponds to Murata's \emph{fusion of series transitions}.
As shown in Fig.~\ref{fig:transitionExpansion}, transition $t$ is divided into two transitions, $t_e$ that consumes the tokens, and a second transition $t_c$ that produces the tokens.
The two transitions are connected with a single, fresh place $p$.
Place $p$ should ensure that all variables consumed by the original transition $t$, are passed to transition $t_c$, to ensure that $t_c$ can produce the same tokens as transition $t$ in the original net.
In other words, the type of each input place of $t$ is included in the type of the newly added place $p$.
Moreover, transition $t_e$ is also allowed to emit new, fresh identifiers that, however, will be eventually consumed by $t_c$.

\begin{definition}[Transition expansion]
Let $(N,m)$ be a marked \tpnid with $N = (\places, \transitions, \flow, \alpha, \beta)$, let $t \in \transitions$, and let $\vec{\lambda} \in \Lambda^*$ and $\vec{\mu} \in (\V\setminus \newvar{t})^*$ such that $\type_\V(x) \in \vec\lambda$ and $x\in\vec{\mu}$, for all $x\in \invar{t}$, and $\type_{\V}(\vec{\mu}) = \vec{\lambda}$.
The \emph{transition expanded \tpnid} is defined by $R_{t,\vec{\lambda},\vec{\mu}}(N,m) = ((\places', \transitions',\flow', \alpha',\beta'),m)$, where:
\begin{itemize}
\itemsep=0.85pt
	\item $\places' = \places \cup \{p\}$ with $p \not\in \places$; and
	$\transitions' = (\transitions \setminus \{t\}) \cup \{t_e, t_c\}$ with $t_e, t_c \not\in \transitions$;
	\item
	$\flow' = (\flow \setminus ((\pre{t} \times \{t\}) \cup (\{t\}\times\post{t}) )) \cup
	(\pre{t} \times \{t_e\})
	\cup \{(t_e,p),(p,t_c)\}
	\cup
	(\{t_c\} \times \post{t})$;
	\item
	$\alpha'(p) = \vec{\lambda}$ and $\alpha'(q) = \alpha(q)$ for all $q \in \places$;
	\item
	$\beta'(f) = [\,\vec{\mu}\,]$ if $f \in \{(t_e,p),(p,t_c)\}$, $\beta'((q,t_e))=\beta((q,t))$ for $q \in \pre{t}$, $\beta'((t_c,q))=\beta((t,q))$ for $q \in \post{t}$, and $
	\beta'(f) = \beta(f)$ otherwise.
\end{itemize}
\end{definition}

Transition $t_e$ is allowed to introduce new variables, but key is that inscription $\vec{\mu}$ contains all input variables of transition $t$. Consequently, $\vec{\mu}$ encodes the binding of transition $t$.
We use this to prove weak bisimulation between a \tpnid and it transition expanded net.
The idea behind the simulation relation $Q$ is that the firing of $t_e$ is postponed until $t_c$ fires.
In other words, $Q$ encodes that tokens remain in place $q$ until transition $t_c$ fires.

\begin{lemma}\label{lm:bisimtransexpansion}
Given marked \tpnid $(N,m_0)$ with $N = (\places, \transitions, \flow, \alpha, \beta)$, transition $t\in \transitions$, $\vec\lambda\in\Lambda^*$ and $\mu \in \V^*$. Let $t_e, t_c$ be the transitions added by the expansion.\\
Then $\tsys{N,m_0} \approx^{r} \rho_{r}(\tsys{R_{t,\vec\lambda,\mu}(N,m_0)})$ with $r = \{(t_e, \tau), (t_c,t) \}$.
\end{lemma}

\begin{proof}
Let $(N',m_0') = R_{t,\lambda,\mu}(N,m_0)$. Then $m_0' = m_0$.
Define relation $Q \subseteq \mathbb{M}(N) \times \mathbb{M}(N')$ such that $(m,m') \in Q$ iff $m(q) = m'(q)$ for all places $q \in \places\setminus \pre{t}$ and $m(q) = m'(q) + \sum_{b \in \supp{M'(p)}} M'(p)(b)\cdot \rho_{\mu(b)}\beta((q,t))$, where $\mu(b)$ is a shorthand for the binding $\psi : \V\to\I$ with $\psi(x) = b(i)$ iff $\mu(i) = x$ for all $1 \leq i \leq |\mu|$. Then $(m_0,m_0') \in Q$.

\vspace*{1.8mm}\noindent ($\Rightarrow$)
Follows directly from the firing rule, and the construction of $\mu$.

\vspace*{1.8mm}\noindent($\Leftarrow$)
Let $(m,m') \in Q$ and $\fire{(N',m')}{u,\psi}{(N',\bar{m}')}$. We need to show a marking $\bar{m}$ exists such that $m \xdasharrow{~(t,\psi)~} \bar{m}$ and $(\bar{m},\bar{m}') \in Q$.
If $t_e \neq u \neq t_c$, the statement holds by definition of the firing rule.
Suppose $u = t_e$, i.e., $r(u) = \tau$. Hence, we need to show that $(m, \bar{m}') \in Q$. Let $q \in \pre{t}$.
Since $(m,m') \in Q$, we have
$m(q) = m'(q) + \sum_{b \in \supp{m'(p)}} m'(p)(b)\cdot \rho_{\mu(b)}\beta((q,t))$.
By the firing rule, we have $\bar{m}'(p) = m'(p) + [\rho_\psi(\mu)]$ and $m'(q) = \bar{m}'(q) + \rho_{\psi}(\beta((q,t)))$.
By construction, $\rho_{\psi}$ and $\rho_{\mu([\rho_\psi(\mu)])}$ are identical functions.
Rewriting gives $m(q) = \bar{m}'(q) + \sum_{b \in \supp{\bar{m}'(p)}} m'(p)(b)\cdot \rho_{\mu(b)}\beta((q,t))$, and thus $(m,\bar{m}') \in Q$.

\smallskip
Suppose $u = t_c$, i.e., $r(u) = t$ and $[\rho_\psi(\mu)]\leq m'(p)$.
Let $q \in \pre{t}$. Then
$m(q) = m'(q) + \sum_{b \in \supp{m'(p)}} m'(p)(b)\cdot \rho_{\mu(b)}\beta((q,t))$. Since $\bar{m}'(p) + [\rho_{\psi}(\mu)] = m'(p)$ and $\rho_{\psi}(\beta((q,u))) = \rho_{\mu([\rho_\psi(\mu)])}(\beta((q,u)))$, we obtain
$m(q) = m'(q) + \left(\sum_{b \in \supp{\bar{m}'(p)}} \bar{m}'(p)(b)\cdot \rho_{\mu(b)}\beta((q,t))\right) +\linebreak
\rho_{\psi}\beta((q,t))$.
Hence, a marking $\bar{m}$ exists such that $\fire{(N,m)}{t,\psi}{(N,\bar{m})}$ and $(\bar{m},\bar{m}') \in Q$.
\end{proof}

\begin{lemma}\label{lm:transitionexpansionpreservesidentifiersoundness}
Given marked \tpnid $(N,m_0)$ with $N = (\places, \transitions, \flow, \alpha, \beta)$, transition $t\in \transitions$, $\vec\lambda\in\Lambda^*$ and $\mu \in \V^*$.
Then $(N,m_0)$ is identifier sound iff $R_{t,\vec\lambda,\mu}(N,m_0)$ is identifier sound.
\end{lemma}

\begin{proof}
Define $(N',m'_0) = R_{t,\vec\lambda,\mu}(N,m_0)$ with $N' = (\places',\transitions',\flow',\alpha',\beta')$, let $t_e, t_c \in\transitions' \setminus \transitions$ be the two added transitions and let $q \in \places'\setminus\places$ be the added place.
Let $Q \subseteq \mathbb{M}(N) \times \mathbb{M}(N')$ be the weak bisimulation relation as defined in Lm.~\ref{lm:bisimtransexpansion}.

\smallskip
Suppose $\lambda \in \Lambda$.
If $\lambda \in \type_\Lambda(N)$, then $C_N(\lambda) = C_{N'}(\lambda)$ by definition of the transition expansion. As $I(\lambda) \cap \id{m_1} = I(\lambda) \cap \id{m_2}$ for all $(m_1, m_2) \in Q$, the statement directly follows from Thm.~\ref{thm:weakbisimpreservesidentifiersoundness}.

\smallskip
Otherwise, if $\lambda \in \type_\Lambda(N')\setminus \type_\Lambda(N)$,
then, for all places $p\in P'$, having $\lambda \in \alpha(p)$ implies that $p = q$, i.e., $q$ is the only place that contains tokens carrying identifiers of type $\lambda$, and $t_c \in C_{N'}(\lambda)$.
Suppose that there exist a marking $m' \in \reachable{N'}{m'_0}$, firing sequence $\eta$, vector $\vec{\cname{id}} \in \I^*$ and $\cname{id} \in I(\lambda)$
such that $\fire{(N',m'_0)}{\eta}{(N',m')}$, $m(q)(\vec{\cname{id}}) > 0$ and $\cname{id}\in\vec{\cname{id}}$.
Thus, $\cname{id} \in \id{m'}$.
Then a binding $\psi : \V \rightarrow \I$ exists such that $\cname{id} \in \rng{\psi}$ and $(t_e,\psi) \in \eta$.
As $t_e$ is an emitting transition for $\lambda$,
we have that $|\beta'((t_e,\psi))| = 1$, and thus $m'(q)(\vec{\cname{id}}) = 1$.
By the firing rule and the construction of $N'$, we have that $m'(q)(\psi(\beta'((q,t_c)))) > 0$, i.e., $\enabled{(N',m')}{t_c,\psi}$.
Hence, a marking $m'' \in \mathbb{M}(N')$ exists with $\fire{(N',m')}{t_c,\psi}{(N',m'')}$ such that  $\cname{id} \not\in \id{m''}$.
Hence, $N'$ is weakly $\lambda$-terminating.
As $\post{q} = \{t_c\}$, transition $t_c$ is the only transition which can remove identifiers of type $\lambda$, and thus $N'$ is also proper $\lambda$ completing.
\end{proof}

\subsubsection{Place Duplication}
Whereas the previous two rules introduced ways to extend sequences, the third rule introduces parallelism by duplicating a place, as shown in Figure~\ref{fig:placeDuplication}.
It is based on the \emph{fusion of parallel transitions} reduction rule of Murata.
For \tpnids, duplicating a place has an additional advantage: as all information required for passing the identifiers is already guaranteed, the duplicated place can have any place type.
Transition $t$ can emit new identifiers, provided that transition $u$ does not already emit\linebreak these.

\begin{definition}[Duplicate place]
	Let $(N,m)$ be a marked \tpnid with $N = (\places, \transitions, \flow, \alpha, \beta)$, let $p \in \places$, such that $m(p) = \emptyset$, and some transitions $t, u\in \transitions$ exist with $\pre{p} =\{t\}$, $\post{t} = \{p\}$, $\post{p} = \{u\}$ and $\pre{u} = \{p\}$. Let $\vec{\lambda} \in \Lambda^*$ and $\vec{\mu} \in (\V\setminus \newvar{u})^*$ such that $\type_\V(\mu) = \lambda$.
	Its \emph{duplicated place \tpnid} is defined by $D_{p,\lambda,\mu}(N,m) = ((\places', \transitions,\flow', \alpha',\beta'),m)$, where:

	\begin{itemize}
\itemsep=0.85pt
		\item $\places' = \places \cup \{ q \}$, with $q\not\in P$, and $\flow' = \flow \cup \{ (t, q), (q, u) \}$;
		\item $\alpha' = \alpha\cup \{ q \mapsto \vec{\lambda} \}$ and $\beta' = \beta \cup \{ (t,q) \mapsto [\mu], (q,u) \mapsto [\mu] \}$.
	\end{itemize}
\end{definition}

As the duplicated place cannot hamper the firing of any transition, all behavior is preserved by a strong bisimulation on the identity mapping.

\begin{lemma}\label{lm:bisimplaceduplication}
Given a marked \tpnid $(N,m_0)$ with $N = (\places, \transitions, \flow, \alpha, \beta)$, place $p\in \places$, $\vec\lambda\in\Lambda^*$ and $\mu \in \V^*$.
Then $\tsys{N,m_0} \sim^r \tsys{D_{p,\vec\lambda,\mu}(N,m_0)}$.
\end{lemma}

\begin{proof}
Let $(N',m'_0) = D_{p,\vec\lambda,\mu}(N,m_0)$.
Define relation $Q \subseteq \mathbb{M}(N) \times \mathbb{M}(N')$ such that $(m, m') \in Q$ iff $m(q)\! =\! m'(q)$ for all places $q\in \places$. The bisimulation relation trivially follows from the firing rule.
\end{proof}

\begin{lemma}\label{lm:placeduplicationpreservesidentifiersoundness}
	Let  $(N,m_0)$ be a marked \tpnid with $N = (\places, \transitions, \flow, \alpha, \beta)$, place $p\in \places$, $\lambda\in\Lambda^*$ and $\mu \in \V^*$. Then $(N,m_0)$ is identifier sound iff $D_{p,\vec\lambda,\mu}(N,m_0)$ is identifier sound.
\end{lemma}
\begin{proof}
	Let $(N',m'_0) = D_{p,\vec\lambda,\mu}(N,m_0)$ with $N' = (\places',\transitions',\flow',\alpha',\beta')$, let $q \in \places'\setminus\places$ be the place added by the place duplication rule, and let $Q \subseteq \mathbb{M}(N) \times \mathbb{M}(N')$ be the bisimulation relation of Lm.~\ref{lm:bisimtransduplication}.
	
\smallskip\noindent 	($\Rightarrow$)
	Let $\lambda \in \type_\Lambda(N)$ be a type.
	Then, $\lambda \in \type_\Lambda(N')$ and $C_N(\lambda) = C_{N'}(\lambda)$ by definition of the place duplication rule.
	As $I(\lambda) \cap \id{m_1} = I(\lambda) \cap \id{m_2}$ for all $(m_1,m_2) \in Q$, the statement directly follows from Thm.~\ref{thm:weakbisimpreservesidentifiersoundness}.
	
\smallskip\noindent	($\Leftarrow$)
	Let $\lambda \in \type_\Lambda(N')$ be a type.
	If $\lambda \in \type_\Lambda(N)$, then the statement directly follows from Thm.~\ref{thm:weakbisimpreservesidentifiersoundness} as  $C_N(\lambda) = C_{N'}(\lambda)$ and $I(\lambda) \cap \id{m_1} = I(\lambda) \cap \id{m_2}$ for all $(m_1,m_2) \in Q$.
	Otherwise $\lambda \not\in \type_\Lambda(N)$.
	Then,  by definition of the place duplication, it must be that $\lambda \in \alpha(q)$.
	Then \mbox{$E_{N'}(\lambda) = {u}$} and $C_{N'}(\lambda) = {u}$, where $(t,q),(q,u) \subseteq \flow'$.
	Suppose there there exist a marking $m' \in \reachable{N'}{m'_0}$, firing sequence $\eta$, an identifier vector $\vec{\cname{id}} \in \I^*$ and identifier $\cname{id} \in I(\lambda)$ such that $\fire{(N',m'_0)}{\eta}{(N',m')}$, $\cname{id} \in \vec{\cname{id}}$ and $m'(q)(\vec{\cname{id}}) > 0$.
	Then a binding $\psi : \V \rightarrow \I$ exists such that $\cname{id} \in \rng{\psi}$  and $(t,\psi) \in \eta$.
	As $t$ is an emitting transition for $\lambda$, then $|\beta'((t,\psi))| = 1$, i.e., $m'(q)(\vec{\cname{id}}) = 1$.
	By the firing rule and the construction of $N'$, it holds that $m'(q)(\psi(\beta'((q,u)))) > 0$, and thus $\enabled{(N',m')}{u,\psi}$.
	Hence, there exists a marking $m'' \in \mathbb{M}(N')$ such that $\fire{(N',m')}{u,\psi}{(N',m'')}$.
	Then $\cname{id} \not\in \id{m''}$.
	Hence, $N'$ is weakly $\lambda$-terminating.
	As $\post{q} = \{u\}$, transition $u$ is the only transition which can remove identifiers of type $\lambda$, and hence $N'$ is also proper $\lambda$-completing.
\end{proof}

\subsubsection{Transition Duplication}
As already recognized by Berthelot~\cite{Berthelot78}, if two transitions have an identical preset and postset, one of these transitions can be removed while preserving liveness and boundedness. Murata's fusion of parallel places is a special case of this rule, requiring that the preset and postset are singletons.
For t-JNs, this results in the duplicate transition rule: any transition may be duplicated, as shown in Figure~\ref{fig:transitionDuplication}.
As duplication should not hamper the behavior of the original net, we require that the inscriptions of the duplicated transition are identical to the original transition.

\begin{definition}[Duplicate transition]
	Let $(N,m)$ be a marked \tpnid with $N = (\places, \transitions, \flow, \alpha, \beta)$, and let $t \in \transitions$ such that some places
      $p, q\in \places$ \ exist with $\pre{t} =\{p\}$ and $\post{t} = \{q\}$.
	Its \emph{duplicated transition \tpnid} is defined by $D_{t}(N,m) = ((\places, \transitions',\flow', \alpha,\beta'),M)$, where:

	\begin{itemize}
       \itemsep=0.9pt
		\item $\transitions' = \transitions \cup \{u \}$, with $u\not\in \transitions$, and $\flow' = \flow \cup \{ (p, u), (u, q) \}$;
		\item $\beta'((p,u)) = \beta((p,t))$, $\beta((u,q)) = \beta((t,q))$ and $\beta'(f) = \beta(f)$ for all $f\in \flow$.
	\end{itemize}
\end{definition}

As the above rule only duplicates $t\in\transitions$, the identity relation on markings is a strong rooted bisimulation. The proof is straightforward from the definition.

\begin{lemma}\label{lm:bisimtransduplication}
Given a marked \tpnid $(N,m_0)$ with $N = (\places, \transitions, \flow, \alpha, \beta)$, and transition $t\in\transitions$.
Then $\tsys{N,m_0} \sim^r \rho_{\{(u,t)\}}(\tsys{D_{t}(N,m_0)})$.
\end{lemma}
\begin{proof}
	Let $(N',m'_0) = D_{t}(N,m_0)$.
	Define relation $Q \subseteq \mathbb{M}(N) \times \mathbb{M}(N')$ such that $(m, m') \in Q$ iff $m(p)\! =\! m'(p)$ for all places $p \in \places$. The bisimulation relation trivially follows from the firing rule.
\end{proof}

\begin{lemma}\label{lm:transitionduplicationpreservesidentifiersoundness}
	Let  $(N,m_0)$ be a marked \tpnid with $N = (\places, \transitions, \flow, \alpha, \beta)$ and transition $t\in\transitions$. Then $(N,m_0)$ is identifier sound iff $D_{t}(N,m_0)$ is identifier sound.
\end{lemma}
\begin{proof}
Let $(N',m'_0) = D_{t}(N,m_0)$, and let $Q \subseteq \mathbb{M}(N) \times \mathbb{M}(N')$ be the bisimulation relation of Lm.~\ref{lm:bisimtransduplication}.
Then $\id{m_1} = \id{m_2}$ for all $(m_1,m_2) \in Q$.
If $t \not\in C_N(\lambda)$ then the statement directly follows from Thm.~\ref{thm:weakbisimpreservesidentifiersoundness}.
Otherwise, i.e., $t \in C_N(\lambda)$, then $C_N(\lambda) = C_N(\lambda) \cup \{u\}$.
By Lm~\ref{lm:weakbisimpreservesweaktypetermination}, $N$ is weakly $\lambda$ terminating.
As $\beta_{N'}((t,p)) = \beta_{N'}((u,p))$ and $\beta_{N'}((p,t)) = \beta_{N'}((p,u))$ for all places $p \in P_{N'}$, proper type completion cannot distinct firing transition $t$ from transition $u$.
Hence, proper $\lambda$ completion follows from the proof of Lm~\ref{lm:weakbisimpreservespropertypecompletion}, and thus, $N$ is identifier sound.
\end{proof}

\subsubsection{Adding Identity Transitions}
In \cite{Berthelot78}, Berthelot classified a transition $t$ with an identical preset and postset, i.e., $\pre{t} = \post{t}$ as irrelevant, as its firing does not change the marking. The reduction rule \emph{elimination of self-loop transitions} is a special case, as Murata required these sets to be singletons.
We now introduce the fifth rule allowing the addition of a self-loop transition, as depicted in  Figure~\ref{fig:selfLoopAddition}.

\begin{definition}[Self-loop addition]
	Let $(N,m)$ be a marked \tpnid with $N = (\places, \transitions, \flow, \alpha, \beta)$, and let $p \in \places$.
	Its \emph{self-loop added \tpnid} is defined by $A_{p}(N,m) = ((\places, \transitions',\flow', \alpha,\beta'), m)$, where:
	\begin{itemize}
		\item $\transitions' = \transitions \cup \{ t \}$, with $t\not\in \transitions$, and
		 $\flow' = \flow \cup \{ (p, t), (t, p) \}$;
		\item $\beta'((p,t)) = \beta'((t,p)) = [\,\vec\mu\,]$ with $\vec\mu \in \V^*$ such that $\type_\V(\mu) = \alpha(p)$, and $\beta'(f) = \beta(f)$ otherwise.
	\end{itemize}
\end{definition}

Similar to the duplicate transition rule, the self-loop addition rule does not introduce new behavior, except for silent self-loops. Hence, the identity relation on markings is a weak rooted bisimulation.

\begin{lemma}\label{lm:selflooptrans}
Given a marked \tpnid $(N,m_0)$ with $N = (\places, \transitions, \flow, \alpha, \beta)$, and place $p\in\places$.
Then $\tsys{N,m_0} \approx^r \hide{\{t\}}(\tsys{A_{p}(N,m_0)})$ with $t$ the added self-loop transition.
\end{lemma}

\begin{proof}
	Let $(N',m'_0) = A_{p}(N,m_0)$.
	Define relation $Q \subseteq \mathbb{M}(N) \times \mathbb{M}(N')$ such that $(m, m') \in Q$ iff $m(p)\! =\! m'(p)$ for all places $p \in \places$. The bisimulation relation trivially follows from the firing rule.
\end{proof}

\begin{lemma}\label{lm:identitytransitionpreservesidentifiersoundness}
	Let  $(N,m_0)$ be a marked \tpnid with $N = (\places, \transitions, \flow, \alpha, \beta)$ and place $p\in\places$. Then $(N,m_0)$ is identifier sound iff $A_{p}(N,m_0)$ is identifier sound.
\end{lemma}

\begin{proof}
	Let $(N',m'_0) = A_{p}(N,m_0)$ and let $Q \subseteq \mathbb{M}(N) \times \mathbb{M}(N')$ be the bisimulation relation of Lm.~\ref{lm:bisimtransduplication}.
	Note that $C_N(\lambda) = C_{N'}(\lambda)$ for all $\lambda \in \type_\Lambda(N)$, since the added self-loop transition does not remove any identifier.
	As $\id{m_1} = \id{m_2}$ for all $(m_1,m_2) \in Q$, the statement directly follows from Thm.~\ref{thm:weakbisimpreservesidentifiersoundness}.
\end{proof}

\subsubsection{Identifier Introduction}
The first five  rules preserve the criteria of block-structured WF-nets.
Murata's 
\emph{elimination of self-loop places} states that adding or removing a marked place with identical preset and postset does preserve liveness and boundedness.
This rule is often used to introduce a fixed resource to a net, i.e., the number of resources is determined in the initial marking.
Instead, identifier introduction adds dynamic resources, as shown in Figure~\ref{fig:identifierCreation}:
transition $t_e$ emits new identifiers as its inscription uses only ``new''  variables (i.e., those that have not been used in the net), and place $p$ works like a storage of the available resources, which can be removed by firing transition
$t_c$.

\begin{definition}[Identifier Introduction]
	Let $(N,m)$ be a marked \tpnid with $N = (\places, \transitions, \flow, \alpha, \beta)$, let $t \in \transitions$, let $\vec\lambda \in
    (\Lambda\setminus \type_P(N))^*$ and $\vec\mu \in \V^*$ such that $\type_{\V}(\vec\mu) = \vec\lambda$.
	The \emph{Identifier introducing \tpnid} is defined by $A_{t,\vec\lambda,\vec\mu}(N,m) = ((\places', \transitions',\flow', \alpha',\beta'), m)$, where:
	\begin{itemize}
\itemsep=0.9pt
		\item $\places' = \places' \cup \{ p \}$ and $\transitions' = \transitions \cup \{ t_e, t_c \}$, for $p \not\in \places$ and $t_e, t_c \not\in \transitions$,
        and\newline
		 $\flow' = \flow \cup \{ (p, t), (t, p), (t_e,p), (p, t_c) \}$;
		\item $\alpha' = \alpha \cup \{ p \mapsto \vec\lambda \}$ and $\beta' = \beta \cup \{ (p,t) \mapsto [\vec\mu], (t,p) \mapsto [\vec\mu], (t_e, p)
        \mapsto [\vec\mu], (p, t_c) \mapsto [\vec\mu] \}$;
	\end{itemize}
\end{definition}

\begin{lemma}\label{lm:selfloopplace}
	Given a marked \tpnid $(N,m_0)$ with $N = (\places, \transitions, \flow, \alpha, \beta)$, transition $t\in\transitions$, $\vec\lambda \in \Lambda^*$ and $\vec\mu \in (\V \setminus \var{t})^*$.
	Then $\tsys{N,m_0} \approx^r \hide{\{t_e,t_c\}}(\tsys{A_{t,\vec\lambda,\vec\mu}(N,m_0)})$ with $t_e, t_c$ being the added transitions.
\end{lemma}

\begin{proof}
	Let $N' = (\places',\transitions',\flow',\alpha',\beta')$.
	Define $Q \subseteq \mathbb{M}(N) \times \mathbb{M}(N')$ such that $(m,m') \in Q$ iff $m(q) = m'(q)$ for all $q \in \places$.
	
\smallskip\noindent($\Rightarrow$)
	Suppose $\fire{(N,m)}{u,\psi}{(N,\bar{m}')}$ and $(m,m')\in Q$. If $u \neq t$, the statement directly follows from the firing rule. Same holds for the case when $u = t$ and $t$ is enabled in $m'$.
	If $u = t$ and $t$ is not enabled in $m'$, then a marking $m''$ and binding $\psi'$ exist such that $\fire{(N',m')}{t_e,\psi'}{(N,m'')}$.
	Then $m''(p) > \emptyset$, $(m,m'')\in Q$, and $\enabled{(N',m'')}{u,\psi}$.
	 Hence, markings $\bar{m}''$ and $\bar{m}'$ exist such that $\enabled{(N',m'')}{t,\psi}\fire{(N',\bar{m}'')}{t_c,\psi'}{(N,\bar{m}')}$, and $(m,\bar{m}''), (m',\bar{m}') \in Q$.
	
\smallskip\noindent($\Leftarrow$) Follows directly from the firing rule.
\end{proof}

As shown in~\cite{RVFE11}, unbounded places are width-bounded, i.e., they can carry only boundedly many distinct identifiers, or depth-bounded, i.e., for each identifier, the number of tokens carrying that identifier is bounded, or both.
The place added by the identifier creation rule is by definition width-unbounded, as it has an empty preset. However, it is identifier sound, and thus depth-bounded, as shown in the next lemma.

\begin{lemma}\label{lm:identplaceidentsound}
	Given a marked \tpnid $(N,m)$ with $N = (\places,\transitions,\flow,\alpha,\beta)$. Then $A_{t,\vec\lambda,\vec\mu}(N,m)$ is identifier sound iff $(N,m)$ is identifier sound.
\end{lemma}
\begin{proof}
Let $(N',m') = A_{t,\vec\lambda,\vec\mu}(N,m)$, let $p \in \places'\setminus \places$, and let $\lambda \in \type_\Lambda(N')$.

\smallskip\noindent($\Rightarrow$)
Suppose $(N,m)$ is identifier sound.
Let $\bar m \in \mathbb M(N')$ and $\eta \in (\transitions'\times(\V\rightarrow \I))^*$ such that $\fire{(N',m')}{\eta}{(N',\bar{m})}$.
Let $\cname{id} \in \id{\bar{m}} \cap I(\lambda)$.
If $\lambda \in \type_\Lambda(N)$, weak $\lambda$-termination and proper $\lambda$-completion follow from Lm.~\ref{lm:selfloopplace}.
Suppose $\lambda \not\in \type_\Lambda(N)$, i.e., $\lambda \in \vec\lambda$.
By construction of $N'$, we have $\lambda \in \alpha(q)$ implies $p = q$ for all places $q \in \places'$, $E_{N'}(\lambda) = \{t_e\}$ and  $C_{N'}(\lambda) = \{t_c\}$.
By the firing rule, we have $\cname{id} \in \vec{a}$ and $\cname{id} \in \vec{b}$ imply $\vec a = \vec b$ for all $\vec a, \vec b \in \supp{m(p)}$.
Again by the firing rule, $m(p)(\vec{a}) \leq 1$ for all $\vec a \in \supp{m(p)}$.
In other words, there is only one token carrying identifier $\cname{id}$.
Let $\vec{\cname{id}} \in \colset(p)$ such that $\cname{id} \in \vec{\cname{id}}$ and $m(p)(\vec{\cname{id}}) > 0$. Then $m(p)(\vec{\cname{id}}) = 1$.
Thus, a binding $\psi$ exists such that $(t_e,\psi) \in \eta$ and $\rho_{\vec{\mu}}(\psi) = \vec{\cname{id}}$.
By construction of $N'$, a marking $\bar{m}'$ exists such that $\fire{(N',\bar{m})}{t_c,\psi}{(N',\bar{m}')}$.
Then $\cname{id}\notin\id{\bar{m}'}$.
Hence, $(N',m')$ is weakly $\lambda$-terminating.
It is proper $\lambda$-completing since there is only one token carrying identifier $\cname{id}$.

\noindent($\Leftarrow$)
Suppose $(N',m')$ is identifier sound.
If $\lambda \in \type_\Lambda(N)$, weak $\lambda$-termination and proper $\lambda$-completion follow from Thm.~\ref{thm:weakbisimpreservesidentifiersoundness}.
In case $\lambda \not\in \type_\Lambda(N)$, it is weakly $\lambda$-terminating, since $E_N(\lambda) = \emptyset$, and properly $\lambda$-completing since $C_N(\lambda) = \emptyset$.
\end{proof}

\subsubsection{Soundness for Typed Jackson Nets}
Any net that can be reduced to a net with a single transition using these rules is called a typed Jackson Net (t-JN).


\begin{definition}
The class of \emph{typed Jackson Nets} $\mathcal{T}$ is inductively defined  by:
\begin{itemize}
\item $((\emptyset,\{t\},\emptyset,\emptyset,\emptyset),\emptyset) \in \mathcal{T}$;
\item if $(N,M) \in \mathcal{T}$, then $R_{p,\vec\mu}(N,M) \in \mathcal{T}$;
\item if $(N,M) \in \mathcal{T}$, then $R_{t,\vec\lambda,\vec\mu}(N,M) \in \mathcal{T}$;
\item if $(N,M) \in \mathcal{T}$, then $D_{p,\vec\lambda,\vec\mu}(N,M_0) \in \mathcal{T}$;
\item if $(N,M) \in \mathcal{T}$, then $D_{t}(N,M) \in \mathcal{T}$;
\item if $(N,M) \in \mathcal{T}$, then $A_{p}(N,M) \in \mathcal{T}$;
\item if $(N,M) \in \mathcal{T}$, then $A_{t,\vec\lambda,\vec\mu}(N,M) \in \mathcal{T}$.
\end{itemize}
\end{definition}

As any t-JN reduces to a single transition, and each construction rule goes hand in hand with a bisimulation relation, any liveness property is preserved. Consequently, any t-JN is identifier sound and live.

\begin{theorem}\label{thm:tJNisSound}
Any typed Jackson Net is identifier sound and live.
\end{theorem}
\begin{proof}
We prove the statement by induction on the structure of t-JNs. The statement holds trivially for the initial net, $((\emptyset,\{t\},\emptyset,\emptyset,\emptyset),\emptyset)$.
Suppose $(N', M') \in \mathcal{T}$ is identifier sound.
We show that applying any of the construction rules on $(N',M')$ preserves identifier soundness:
\begin{itemize}
	\item Suppose $(N,M) = R_{p,\vec\mu}(N',M')$.
	The statement follows directly from Lm.~\ref{lm:placeexpansionpreservesidentifiersoundness}.
	\item
	Suppose $(N,M) = R_{t,\vec\lambda,\vec\mu}(N',M')$.
	The statement follows directly from Lm.~\ref{lm:transitionexpansionpreservesidentifiersoundness}.
	\item
	Suppose $(N,M) = D_{p,\vec\lambda,\vec\mu}(N',M')$.
	The statement follows directly from Lm.~\ref{lm:placeduplicationpreservesidentifiersoundness}.
	\item
	Suppose $(N,M) = D_{t}(N',M')$.
	The statement follows directly from Lm.~\ref{lm:transitionduplicationpreservesidentifiersoundness}.
	\item
	Suppose $(N,M) = A_{p}(N',M')$.
	The statement follows directly from Lm.~\ref{lm:identitytransitionpreservesidentifiersoundness}.
	\item
	Suppose $(N,M) = A_{t,\vec\lambda,\vec\mu}(N',M')$.
	The statement follows directly from Lm.~\ref{lm:identplaceidentsound}.
\end{itemize}

\vspace*{-6mm}
\end{proof}

\begin{figure}[h!]
\vspace*{-1mm}
	\centering
	\includegraphics[scale=.8]{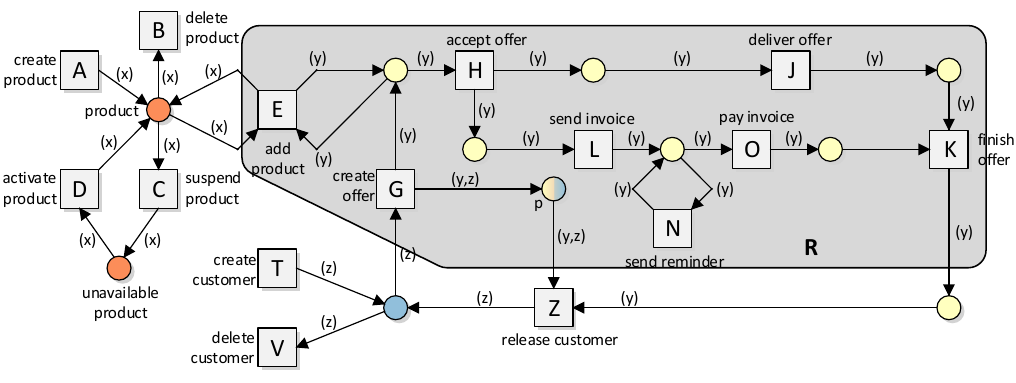}
	\caption{The example of the retailer shop as a typed Jackson Net.}\label{fig:overallModelRepaired}

\vspace*{7mm}
	\centering
	\begin{subfigure}{.25\textwidth}
		\centering
		\includegraphics[scale=.8]{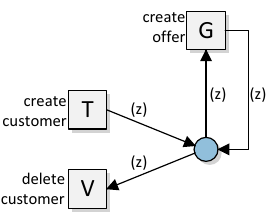}
		\caption{Step 2}\label{fig:overallModelRepaired-step2}
	\end{subfigure}
	\begin{subfigure}{.35\textwidth}
		\centering
		\includegraphics[scale=.8]{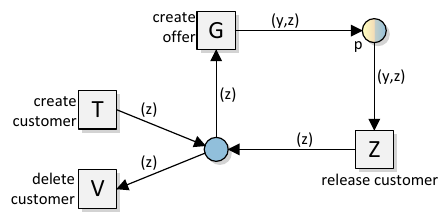}
		\caption{Step 3}\label{fig:overallModelRepaired-step3}
	\end{subfigure}
	\begin{subfigure}{.4\textwidth}
		\centering
		\includegraphics[scale=.8]{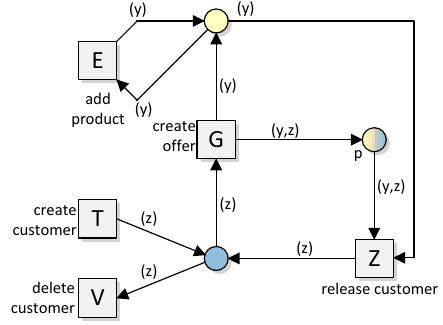}
		\caption{Step 5}\label{fig:overallModelRepaired-step5}
	\end{subfigure}
	\begin{subfigure}{.55\textwidth}
		\centering
		\includegraphics[scale=.8]{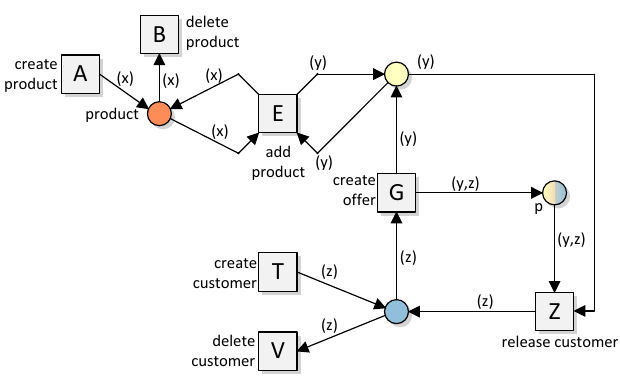}
		\caption{Step 6}\label{fig:overallModelRepaired-step6}
	\end{subfigure}\vspace*{-4mm}
\caption{Several intermediate steps while creating the typed Jackson Net of
       Fig.~\ref{fig:overallModelRepaired}.}\label{fig:overallModelRepaired-steps}\vspace*{-5mm}
\end{figure}

To solve the problem of the running example, several solutions exist.
One solution is shown in Figure~\ref{fig:overallModelRepaired}, which is a t-JN. Several intermediate construction steps are shown in Fig.~\ref{fig:overallModelRepaired-steps}.
The modeler starts with transition $T$, ``create customer''.
The net has no identifiers yet.
Next, transition $T$ is expanded using place type $\langle\mathit{customer}\rangle$, i.e., a place and transition $V$ are added.
A self loop is added to the newly created place (transition $G$), which results in the net depicted in Fig.~\ref{fig:overallModelRepaired-step2}.
The next step introduces place $p$, and is shown in Fig.~\ref{fig:overallModelRepaired-step3}: transition $G$ is expanded using place type $\langle\mathit{order},\mathit{customer}\rangle$.
\mbox{Duplicating} place $p$ allows to create a place with place type $\langle\mathit{customer}\rangle$.
The net depicted in Fig.~\ref{fig:overallModelRepaired-step5} shows the net after adding another self-loop (transition $E$).
The identifier introduction rule allows the modeler to add a product life cycle, so that transition $E$ can add actual products, resulting in the net depicted in Fig.~\ref{fig:overallModelRepaired-step6}.
Now, all required identifiers are present, and transitions $H$, $K$, $J$, $L$, $O$ and $N$ are added using the place type $\langle\mathit{customer}\rangle$, which results in the net depicted in Fig.~\ref{fig:overallModelRepaired}.
As only the types Jackson rules are used, the net is guaranteed to be identifier sound and live.

\subsection{Workflow refinement}
\label{sec:wf-refinement}
A well-known refinement rule is workflow refinement~\cite{HeeSV03}.
In a WF-net, any place may be refined with a generalized sound WF-net.
If the original net is sound, then the refined net is sound as well.
In this section, we present a similar refinement rule.
Given a \tpnid, any place may be refined by a generalized sound WF-net.
In the refinement, each place is labeled with the place type of the refined place, and all arcs in the WF-net are inscribed with the same variable vector.

\begin{definition}[Workflow refinement]\label{def:refinement}
Let $L = (\places_L,\transitions_L,\flow_L,\alpha_L,\beta_L)$, be a \tpnid,  $p \in P_L$ a place, and $N = (\places_N, \transitions_N,\flow_N, W_N,\inp, \outp)$ be a WF-net.
\emph{Workflow refinement} is defined by $L \oplus_{p} N = (\places,\transitions,\flow,\alpha,\beta)$, where:\smallskip

\begin{compactitem}
\itemsep=0.95pt
\item $\places = (\places_L \setminus \{p\}) \cup \places_N$ and $\transitions = \transitions_L \cup \transitions_N$;
\item $\flow = (\flow_L \cap ( ( \places \times \transitions  )  \cup ( \transitions \times \places ) )) \cup \flow_N \cup \{(t,\inp)\mid t \in \pre{p} \} \cup \{(\outp, t) \mid t\in\post{p}\}$;
\item $\alpha(q) = \alpha_L(q)$ for $q \in \places_L\setminus\{p\}$, and $\alpha(q) = \alpha_L(p)$ for $q \in P_N$;
\item $\beta(f) = \beta_L(f)$ for $f\in \flow_L$, $\beta(f) = [\vec\mu]^{(W(f))}$ for $f \in \flow_N$ and $\type_\V(\vec\mu)=\alpha(p)$, $\beta((t,\inp)) = \beta((t,p))$ for $t \in \pre{p}$ and $\beta((\outp,t)) = \beta((p,t))$ for $t \in \post{p}$.
\end{compactitem}
\end{definition}

Generalized soundness of a WF-Net ensures that any number of tokens in the initial place are ``transferred'' to the final place.
As shown in Section~\ref{sec:ec-wf-closure}, the EC-closure of a sound WF-net is identifier sound and live.
A similar approach is taken to show that the refinement is weakly bisimilar to the original net.
Analogously to~\cite{HeeSV03}, the bisimulation relation is the identity relation, except for place $p$. The relation maps all possible token configurations of place $p$ to any reachable marking in the WF-net, given $p$'s token configuration.

\begin{lemma}
Let $L =  (\places_L,\transitions_L,\flow_L,\alpha_L,\beta_L)$ be a \tpnid with initial marking $m_0$, let $p \in P_L$ be a place s.t. $m_0(p)=\emptyset$, and let $N = (\places_N, \transitions_N,\flow_N, W_N,\inp_N, \outp_N)$ be a WF-net. If $N$ is generalized sound, then $\tsys{L,m_0} \approx^{r}
\hide{T_N}(\tsys{L\oplus_p N,m_0})$.
\end{lemma}
\begin{proof}
For simplicity, we start by defining a type extension of $N$ as a \tpnid $N'= (\places_{N'}, \transitions_{N'}, \flow_{N'}, \alpha, \beta)$, where $\type(\vec{v})=\vec\lambda$, $\alpha(p) =\vec{\lambda}$ for all places $p \in \places_N$, and $\beta(f) = \vec{v}^{W(f)}$ for all $f \in \flow_N$, and $\beta((t_e,\inp))=\beta((\outp, t_c)) = [\vec{v}]$.

\medskip
To prove bisimilarity, we define $R = \set{ (M, M'+m) \mid M \in \reachable{L}{M_0}, M'\in\A(M), m\in\B(M)}$
	where
	\begin{itemize}
\itemsep=0.9pt
	\item $\A(M):=\set{M' \mid M'\in\reachable{L}{M_0}, M'(p)=\emptyset \text{ and } \forall q\in P_L\setminus\set{p}: M'(q)=M(q) }$, 
	and
	\item $\B(M):=\set{m\mid m\in\reachable{N'}{m''_0}, m''_0(\inp)=M(p) \text{ and }  \forall q\in P_{N'}\setminus\set{p}: m''_0(q)=\emptyset}$. 
	\end{itemize}
	Intuitively, $\B(M)$ is essentially the $\vec{\lambda}$-typed set of reachable markings of $N$ for a fixed $k$-tokens in $\inp$, where such tokens in $N'$ are provided by $M$ (more specifically, by $M(p)$).

\vspace*{1.8mm}	\noindent($\Rightarrow$) Let $(M,M'+m)\in R$ and $\fire{M}{t,\psi}{\bar M}$.
	We need to show that there exists $\bar{M}'$ and $\bar{m}$ s.t. 	
	$(M'+m) \xdasharrow{~(t,\psi)~} (\bar{M}'+\bar{m})$ and $(\bar{M},\bar{M}'+\bar{m})\in R$. To this end, we consider the following cases.
	\begin{enumerate}[(i)]
		\item If $t\not\in\pre{p}$ (or $p\not\in\pre{t}$), then $M'(q)=M(q)$ for all $q\in P_L\setminus{p}$ (follows from the definition of $\A(M)$), and thus $t$ is also enabled in $M'(q)$ and $M(q)\geq \beta((q,t))$. Then by the firing rule there exists $\bar{M}'$ s.t. $\fire{(M'+m)}{t,\psi}{(\bar{M}'+m)}$ and $\bar{M}'(q)=\bar{M}(q)$ for all $q\in\places_L$.
	Thus, $(M',\bar{M}'+m)\in R$.
		\item If $t\in\pre{p}$, then, since $M'(q)=M(q)$ for all $q\in P_L\setminus{p}$, $t$ must be enabled in $M'$. by the refinement construction from Definition~\ref{def:refinement}, $t$ is enabled regardless the marking of $\inp$.
		By the firing rule, there exists $\bar{M}'$ and $\bar{m}$ such that
		$\fire{(M'+m)}{t,\psi}{(\bar{M}'+\bar{m})}$, $\bar{M}'(q)=\bar{M}(q)$ for all $q\in\places_L\setminus{p}$, and $\bar{m}(\inp)=\bar{M}(p)+m(\inp)$.
		Moreover, by the definition of $R$, $\inp$ can be marked with arbitrarily many tokens from $M(p)$.  Thus,  $(M',\bar{M}'+\bar{m})\in R$.
		\item  If $p\in \pre t$ and  $\rho_\psi(\beta((p,t)))=\vec{\cname{id}}$, then,
	given that $N$ is generalized sound and by applying Lemma~\ref{lemma:weak-bisim-wrappednet},  there exists a firing sequence $\eta$ for $N'$ that carries identifier $\vec{\cname{id}}$ to $\outp$. This means that, by construction, $M'(q)=M(q)$, for all $q\in\places_L$, and $m(\outp)(\vec{\cname{id}})=M(p)(\vec{\cname{id}})$. Hence, $t$ is enabled in $(M'+m)$ under binding $\psi'$ that differs from $\psi$ everywhere but on place $\outp$. By the firing rule, there exists $(\bar{M}'+\bar{m})$ s.t. $\fire{(M'+m)}{t,\psi'}{(\bar{M}'+\bar{m})}$ and $(M,\bar{M}'+\bar{m})\in R$.
	\item If $t\in \pre{p}\cap\post{p}$, then $M(p)\neq \emptyset$ (since $\fire{M}{t,\psi}{\bar M}$).
	Assume that $\rho_\psi(\beta((t,p)))=\vec{\cname{id}_1}$
	and $\rho_\psi(\beta((p,t)))=\vec{\cname{id}_2}$.
	By construction, we know that $M'(q)=M(q)$ for all $q\in\places_L$ and $m$ marks some of the places in $P_N$. Since $N$ is generalized sound and by Lemma~\ref{lemma:weak-bisim-wrappednet}, we can safely assume that $\vec{\cname{id}_2}\in m(p)$ (otherwise, we can apply the reasoning from the previous  case). Then it easy to see that, by construction, $t$ is enabled in $(M'+m)$ under the same binding $\psi$. Thus, by the firing rule there exists $\bar{M}'$ s.t. $\fire{(M'+m)}{t,\psi}{(\bar{M}'+\bar{m})}$, where $\bar{M}'(q)=\bar{M}(q)$ for all $q\in\places_L$, $\bar{m}(\inp)=m(\inp)+[\vec{\cname{id}_1}^{\beta((t,\inp))}]$, $\bar{m}(\outp)=m(\outp)-[\vec{\cname{id}_2}^{\beta((\outp,t))}]$ and $m(w)=\bar{m}(w)$ for all $w\in\places_N$.
	It is easy to see that $(M',\bar{M}'+m)\in R$.
	\end{enumerate}

\smallskip	\noindent($\Leftarrow$) Let $\fire{(M'+m)}{t,\psi}{(\bar{M}'+\bar{m})}$ and $(M,M'+m)\in R$.
	If $t\in T_L$, then this can be proven by analogy with the previous cases (that is, we need to consider all possible relations of $t$ and $p$). If $t\in T_N$, then  $\fire{(M'+m)}{t,\psi}{(M'+\bar{m})}$ and $(\bar{M},M'+\bar{m})\in R$, where  $\bar{M}(q)=M'(q)$, for all $q\in P_L$, and $\bar{M}(p)=M(p)$.
\end{proof}

As a consequence of the bisimulation relation, the refinement is identifier sound and live if the original net is identifier sound.

\begin{theorem}
	Let  $(L,M)$ be a marked \tpnid and $N$ be a generalized sound WF net. Then $(L, M)$ is identifier sound and live iff $(L\oplus N, M)$ is identifier sound and live.
\end{theorem}

The refinement rule allows to combine the approaches discussed in this section.
For example, a designer can first design a net using the construction rules of Section~\ref{sec:typedJN}, and then design generalized WF-nets for specific places.
In this way, the construction rules and refinement rules ensure that the designer can model systems where data and processes are in resonance.

\newcommand{\obj}[1]{#1^o}
\newcommand{\res}[1]{#1^r}
\newcommand{\restype}{\eta}

\section{Enriching PNIDs with resources}
\label{sec:soundness:with:resources}
In the previous section we discussed pattern-based correctness criteria, which allow to construct PNID models that are sound by design. We now consider arbitrary PNIDs, and study how they can be enriched with \emph{resources}, introducing a dedicated property, called \emph{conservative resource management}, which captures that the net suitably employs resources. Following a similar approach, in spirit, to that of Section~\ref{sec:soundness:by:construction}, we define a modelling guideline, called \emph{resource closure}, which takes as input a PNID and indicates how to enrich it with resources through a well-principled approach. We then show that, by construction, if the input PNID is sound, then all its possible resource closures do not only maintain soundness, but they also guarantee that resources are conservatively managed. In addition, we prove that such resource closures are also bounded, and discuss the implications on the analysis of this class of PNIDs.

\subsection{Resource-aware PNIDs}

As customary for Petri nets, we model resource types as (special) places. However, differently from typical approaches like \cite{HSV06,LoBJ22,MontaliR16}, where resources are represented as indistinguishable (black) tokens populating such places, we assign identifiers to resources. This allows one to explicitly track how resources participate to the execution, and in particular how they relate to the different objects. At the same time, this poses a conceptual question: are different copies of the same identifier in distinct tokens representing different actual resources, or distinct references to the same resource? We opt for the latter approach, as it is the one that fully complies with this \emph{named approach} to resource management. As a consequence of this choice, we blur in the section the distinction between resource and resource identifier, using the two terms interchangeably.

\medskip
Technically, from now on we assume that $\Lambda$ is partitioned into two sets: $\obj{\Lambda}$ for object types, and $\res{\Lambda}$ for resource types. Given a resource type $\restype \in \res{\Lambda}$, we call its identifiers \emph{($\restype$-)resources}. We then simply define a \emph{resource-aware} \tpnid as a \tpnid with some distinguished places, each being of a certain resource type.

\begin{definition}[Resource-aware \tpnid] A \tpnid $N=(\places,\transitions,\flow,\alpha,\beta)$ is \emph{resource-aware} if there exists at least one place $p \in \places$ such that $\alpha(p) \in \res{\Lambda}$. We refer to  the non-empty subset $\res{\places} = \set{p \in \places \mid \alpha(p) \in \res{\Lambda}}$ of $\places$ as the set of \emph{resource places} of $N$.
\end{definition}

The initial marking of a resource-aware \tpnid hence identifies which resources are available per resource type. Consistently with the named approach to resources, every resource should be present at most once in the initial marking.

Places typed by the combination of one or more object types and one resource type are used to establish relations between (tuples of) objects and corresponding resources, which we can interpret as \emph{resource assignments}. For example, given an object type $\mathit{Order}$ and a resource type $\mathit{Clerk}$, a token carrying pair $\tup{o,c}$ with $o \in \mathit{Order}$ and $c \in \mathit{Clerk}$ represents that order $o$ is assigned to clerk $c$.

\medskip
In an unrestricted \tpnid, resources and resource assignments can be freely manipulated, generating new resources along the execution, assigning the same resource to multiple objects, and establishing arbitrary relations between resources and objects/other resources. To determine whether a \tpnid employs resources properly, we hence introduce a dedicated property that, intuitively, combines two requirements:\smallskip

\begin{compactitem}
\item \emph{Resource preservation} - only resources present in the initial marking can be used throughout the execution;
\item \emph{Resource exclusive assignment} - in a given marking, each resource can be assigned to at most one object, indicating that the resource is currently responsible for that tuple only.\footnote{An analogous treatment of resources can be defined over tuples of objects, instead of single objects.}
\end{compactitem}\medskip

\noindent The first requirement dictates that no resource can be newly generated during the execution; the second one stipulates that at every step, a resource can be responsible for at most one object, possibly carried by multiple tokens.\smallskip

We formalize these two requirements as follows.

\begin{definition}[Conservative resource management]
\label{def:conservative-resource-management}
A resource-aware marked \tpnid $(N,m_0)$ with $N=(\places,\transitions,\flow,\alpha,\beta)$ is \emph{managing resources conservatively} if the following two conditions hold.\smallskip

\begin{compactitem}
\item \emph{Resource preservation}: for every marking $m\in\reachable{N}{m_0}$, resource type $\restype \in \res{\Lambda}$, and resource $r \in I(\restype) \cap \id{m}$, we have that $r \in \id{m_0}$.
\item \emph{Resource exclusive assignment}: for every marking $m\in\reachable{N}{m_0}$, resource type $\restype \in \res{\Lambda}$, and resource $r \in I(\restype) \cap \id{m}$, there is exactly one tuple $\vec{r}\in \supp{m}$ s.t.~either $\vec{r} = (r)$ or $\vec{r} = \tup{o,r}$ for some object $o$. 
\end{compactitem}

\end{definition}
Consider a resource for $r$ present in the initial marking $\marking_0$, and a reachable marking $\marking$. Two observations are in place regarding Definition~\ref{def:conservative-resource-management}. First, resource exclusive assignment requires that at most one tuple of the form $(o,r)$ exists in the support of $\marking$, to express that multiple tokens carrying the same pair $(o,r)$ may indeed exist, while it is not possible to have in the same marking a different tuple of the form $(o_2,r)$ for some $o_2 \neq o$ (which would indicate the simultaneous assignment of $r$ to $o$ and $o_2$). Second, an active resource $r$ in $\marking$ can then appear in one and only one of the following forms: either\smallskip

\begin{compactitem}
\item $(o,r)$ for some object $o$ -- indicating that $r$ is currently assigned to $o$), or
\item $(r)$ - indicating that $r$ is active and not assigned to any object.
\end{compactitem}

\begin{figure}[!h]
\vspace*{-3mm}
\centering
	\begin{subfigure}{.32\textwidth}
		\centering
		\includegraphics[width=\textwidth]{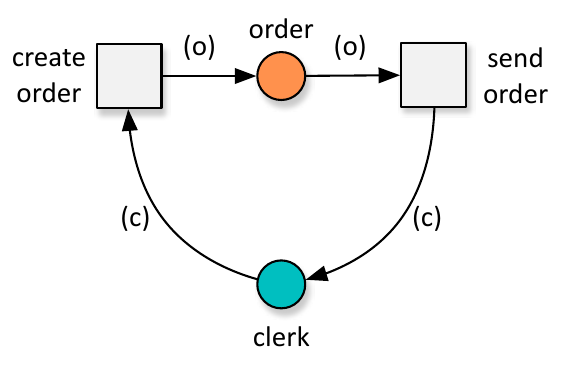}
        \caption{Resource not preserved}\label{fig:resource-not-preserved}
	\end{subfigure}
	\begin{subfigure}{.32\textwidth}
		\centering
		\includegraphics[width=\textwidth]{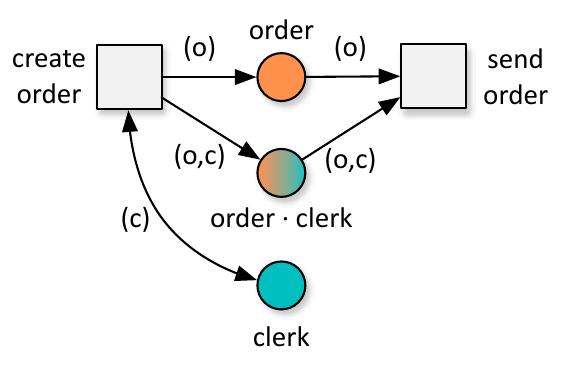}
        \caption{Resource not exclusive}\label{fig:resource-not-exclusive}
	\end{subfigure}
	\begin{subfigure}{.32\textwidth}
		\centering
		\includegraphics[width=\textwidth]{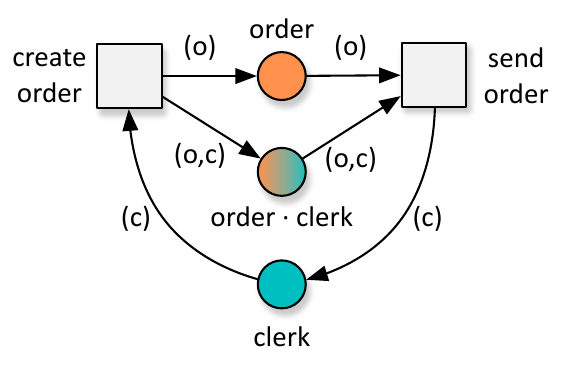}
        \caption{Conservative resource}\label{fig:resource-conservative}
	\end{subfigure}\vspace*{-2mm}
\caption{Three examples of resource-aware \tpnid{s}. Nets (a) and (b) are not managing resources conservatively, as they respective violate the property of resource preservation and that of resource exclusive assignment. Net (c) is instead a positive example that satisfies both properties.}\label{fig:resource-management}\vspace*{-1mm}
\end{figure}

\begin{example}
\label{ex:resource-management} Figure \ref{fig:resource-management} shows three examples of resource-aware \tpnid{s}, where place \emph{order} contains objects, and place \emph{clerk} resources.

\medskip
The \tpnid in Figure \eqref{fig:resource-not-preserved} attempts to model a setting where every order is managed by a clerk. The main issue here is that there is no information stored in the net about which clerk handles which order. In fact, starting from a marking that indicates which clerks are available, the net  violates the property of resource preservation, as when \emph{send order} fires, it brings into the \emph{clerk} place a freshly generated resource (not matching the one previously consumed in \emph{create order} - recall, in fact, that the scope of variables is that of a single transition).

The \tpnid in Figure \eqref{fig:resource-not-exclusive} explicitly keeps track of the assignments of clerks to orders in a dedicated ``synchronization'' place. Starting from a marking that indicates which clerks are available, it satisfies resource preservation, as no new clerk identifier is generated, but it violates the property of resource exclusive assignment, since two different order creations may lead to select the same resource twice, assigning it to two different orders.

The \tpnid in Figure \eqref{fig:resource-conservative} properly handles the assignment of clerks to orders. Starting from a marking that indicates which clerks are available, every time a new order is created, an existing clerk is exclusively assigned to that order. The fact that the same clerk is not reassigned is guaranteed by the fact that the clerk is consumed upon creating the order, and recalled in the assignment place. When the order is sent, its exclusively assigned clerk is released back into the place of (available) clerks, and can be later exclusively assigned to a different order.
\end{example}

By recalling that \tpnid{s} evolving tokens that carry pairs of identifiers are Turing-powerful (see \cite{GGMR22}, and also the proof of Theorem~\ref{thm:id-soundness-undecidable}), and hence every non-trivial property defined over them is undecidable to check, we obtain the following.

\begin{remark}
Verifying whether a marked \tpnid manages resources conservatively is in general undecidable.
\end{remark}

To mitigate this negative result, we introduce an approach that drives the modeller in enriching an input \tpnid via resources following a well-principled approach. The approach generalizes the idea introduced in \cite{MontaliR16} to the more sophisticated case of \tpnid{s}, and does so by following the modelling strategy used in Figure~\ref{fig:resource-conservative}. In particular, it aims at capturing the following modelling principles: \medskip

\begin{compactenum}
\itemsep=1.3pt
\item every object type $\lambda$ is associated to a dedicated resource type $\restype_\lambda$;
\item each such resource type is used in two places - one just typed with $\restype_\lambda$, to indicate which resources of that type are currently available, the other typed by $\lambda \cdot restype_\lambda$, to keep track of which resources are currently assigned to which objects;
\item every object of type $\lambda$ gets assigned a resource of type $\restype_\lambda$ upon creation, and until the consumption of the object, its resource cannot be assigned to any other object;
\item transitions applied to an object may (or may not) require its resource in isolation from the others (if so, implicitly introducing serialization);
\item upon consumption of an object, its resource may either be permanently consumed as well, or freed and become again available to further assignments.
\end{compactenum}\medskip

Technically, we substantiate these intuitive principles through the notion of resource closure.

\begin{definition}[Resource closure]
\label{def:closure}
Let $N=(\places,\transitions,\flow,\alpha,\beta)$ and
$N'=(\places',\transitions',\flow',\alpha',\beta')$
be two \tpnid{s}, and $\lambda\in\type_\Lambda(N)$ be an object type. We say that $N'$ is a
 \emph{$\lambda$-resource closure} of $N$ if the following conditions hold: \medskip

\begin{compactenum}
\item $\places' = \places\cup\set{p_r,p_s}$, where and $p_r$ and $p_s$ are respectively called the \emph{resource} and \emph{assignment} places, and  $p_r,p_s\not\in\places$.\smallskip
\item
$
\alpha'=
    \alpha
    \cup
    \set{p_r\mapsto \restype \mid \restype \in \res{\Lambda} \setminus  \type(N)} \cup \set{p_s\mapsto \lambda\concat \alpha'(p_r)}$ extends $\alpha$ by typing $p_r$ with a resource type $\restype$ not already used in $N$, and $p_s$ with the combination of the object type $\lambda$ and the resource type $\restype$ of $p_r$.\smallskip
\item $F' = F \cup F_r^{out} \cup F_r^{in} \cup F_s^{in} \cup F_s^{syn} \cup F_s^{out}$, where: \smallskip

\begin{compactenum}
\item $F_r^{out} = \set{(p_r,t)\mid t\in E_N(\lambda)}$, the \emph{output resource flow relation}, indicates that every emitter transition for $\lambda$ consumes a resource from $p_r$;
\item $F_r^{in} \subseteq \set{(t,p_r)\mid t\in C_N(\lambda)}$, the \emph{input resource flow relation}, indicates that every collector transition for $\lambda$ \emph{may} return a resource to $p_r$;
\item $F_s^{in} = \set{(t,p_s)\mid t\in E_N(\lambda)}$, the \emph{input assignment flow relation}, indicates that every emitter transition generates an assignment in $p_s$;
\item $F_s^{out} = \set{(p_s,t)\mid t\in C_N(\lambda)}$, the \emph{output assignment flow relation}, indicates that every collector transition consumes an assignment from $p_s$;
\item $F_s^{syn} \subseteq\set{(p_s,t),(t,p_s)\mid t\in T\setminus (E_N(\lambda)\cup C_N(\lambda))}$, the \emph{synchronization assignment flow relation}, indicates that every ``internal'' (i.e., non-emitting and non-consuming) transition for $\lambda$ \emph{may} check for the presence of an assignment in $p_s$.
\end{compactenum} \smallskip

\item $\beta'$ is the extension of $\beta$ satisfying the following conditions:\smallskip

\begin{compactenum}
	\item $\beta'(a,b)=\beta(a,b)$, if $(a,b)\in\flow$.
	\item For every $\lambda$-emitter $t \in E_N(\lambda)$, we define $\beta'(p_r,t)$ and $\beta'(t,p_s)$ as described next. Let $X = \set{x_1,\ldots,x_n}$ be the set of distinct variables of type $\lambda$ mentioned in the inscriptions of outgoing arcs of $t$, that is, $X = \set{x | x \in \beta(t,p) \text{ for some } p \in P}$. These variables denote the $n$ distinct objects of type $\lambda$ created upon firing $t$. We then need to consume $n$ distinct resources of type $\restype$ and establish the corresponding assignments: $\beta(p_r,t) = (r_1)+ \ldots +(r_n)$ and $\beta(t,p_s) = (x_1,r_1) + \ldots + (x_n,r_n)$, where $r_1,\ldots,r_n$ are $n$ distinct (resource) variables of type $\restype$.
    \item A symmetric approach is used to define $\beta'(p_s,t)$ and (in case $(t,p_r) \in F_r^{in}$ is defined) $\beta'(t,p_r)$ for every $\lambda$-collector $t \in C_N(\lambda)$ (based on the incoming arcs of $t$).\vspace{1mm}
    \item For every (internal) transition $t\in T\setminus (E_N(\lambda)\cup C_N(\lambda))$ such that $\set{(p_s,t),(t,p_s)} \subseteq F_s^{syn}$, we define $\beta'(p_s,t)$ and $\beta'(t,p_s)$ as described next. Let $X = \set{x_1,\ldots,x_n}$ be the set distinct variables of type $\lambda$ mentioned in the inscriptions of incoming arcs of $t$, that is, $X = \set{x | x \in \beta(p,t) \text{ for some } p \in P}$. These variables denote the $n$ distinct objects of type $\lambda$ accessed upon firing $t$. We then need to check for the presence of the $n$ distinct assigned resources to this objects, by defining $\beta'(p_s,t) = \beta'(t,p_s) = (x_1,r_1) + \ldots + (x_n,r_n)$, where $r_1,\ldots,r_n$ are $n$ distinct (resource) variables of type $\restype$.\vspace{0.5mm}
\end{compactenum}\medskip
\end{compactenum}

A \tpnid is called a \emph{(full) resource closure of $N$} if it is obtained by recursively constructing, starting from $N$, $\lambda_i$-resource closures for every $\lambda_i\in\type_\Lambda(N)$.
\end{definition}

The closure(s) of a marked \tpnid is defined by closing the \tpnid as per Definition~\ref{def:closure}, and enriching the initial marking by populating the resource places with some resources.

\begin{definition}
\label{def:marked-closure}
Let $(N,m_0)$ and
$(N',m'_0)$
be two marked \tpnid{s}, and $\lambda\in\type_\Lambda(N)$ be an object type. We say that $(N',m'_0)$ is a
 \emph{$\lambda$-marked resource closure} of $N$ if the following conditions hold:

\begin{enumerate}
\item $N'$ is a $\lambda$-resource closure of $N$;
\item $m'_0$ extends $m_0$ by assigning to the resource place $p_r$ introduced in $N'$ a \emph{finite subset} of identifiers of type $\type(p_r)$.
\end{enumerate}

A marked \tpnid is called a \emph{(full) marked resource closure of $(N,m_0)$} if it is obtained by recursively constructing, starting from $(N,m_0)$, $\lambda_i$-marked resource closures for every $\lambda_i\in\type_\Lambda(N)$.
\end{definition}

Notice that, in Definition~\ref{def:marked-closure}, we obey to the named approach described in the opening of this section by assigning to the resource place a set (and not a multi-set) of resources.

\medskip
It is easy to see that, for a \tpnid $N$ and a type $\lambda$, there are only finitely many distinct $\lambda$-closures of $N$, obtained by choosing which $\lambda$-consumers actually return resources in the resource place of the closure, and which internal $\lambda$-transitions access the assignment place of the closure. Marked \tpnid{s} defined on these finitely many distinct closures only differ by the set of identifiers they initially assign to each resource place.

\begin{example}
\label{ex:simple-resource-closure}
Figure~\eqref{fig:resource-conservative} shows an (order-)resource closure for the \tpnid manipulating orders via the \emph{create order} and \emph{send order} transitions, and the orange place in between. In particular, the resource type of clerk is chosen for the closure. The green place (in fact, typed \emph{clerk}) at the bottom of the picture is the resource place of the closure, while the mixed-colored place (in fact, typed \emph{order$\cdot$clerk}) is the assignment place of the closure.
The inscriptions attached to the input/output arcs connecting these two places to the emitter and consumer transitions for orders indicate how identifiers are being  matched when consuming/returning resources, while recalling the objects they get assigned to.
\end{example}

\begin{figure}[t]
\centering
	\begin{subfigure}{.6\textwidth}
		\centering
		\includegraphics[width=\textwidth]{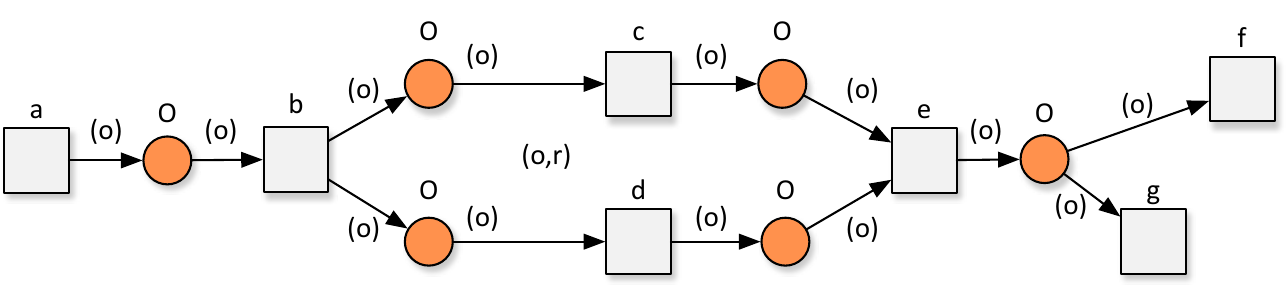}
        \caption{A sound \tpnid manipulating an object type $O \in \obj{\Lambda}$}\label{fig:sound-net}
	\end{subfigure}
	\begin{subfigure}{.6\textwidth}
		\centering
		\includegraphics[width=\textwidth]{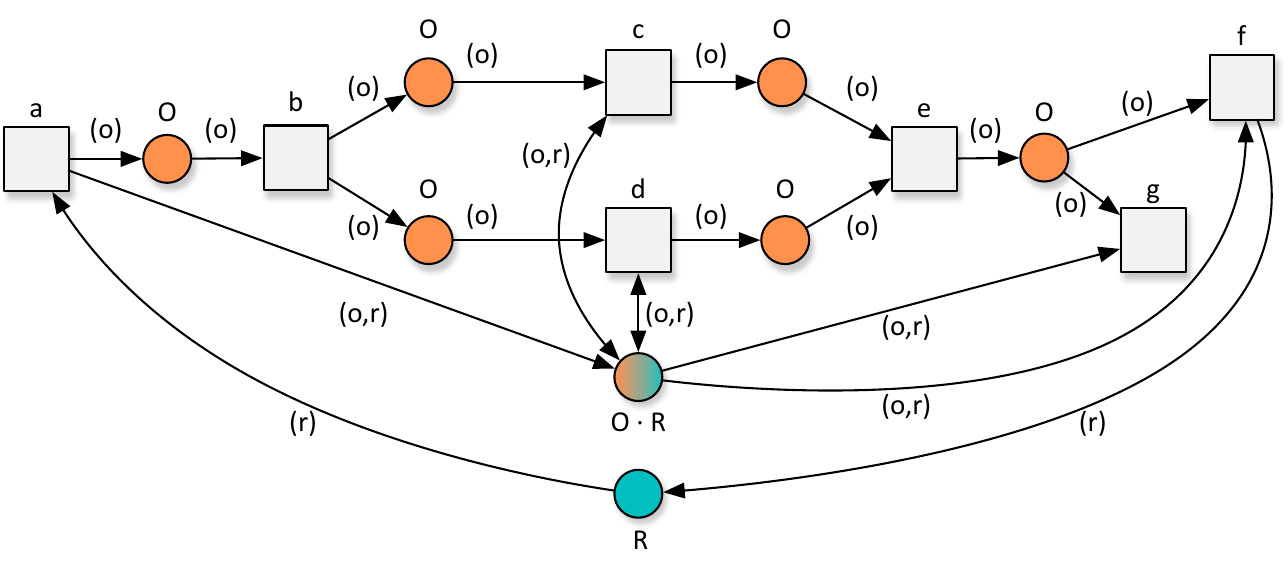}
        \caption{An $O$-resource closure of the \tpnid (a) \protect\label{fig:sound-net-closure-1}}
	\end{subfigure}
	\begin{subfigure}{.6\textwidth}
		\centering
		\includegraphics[width=\textwidth]{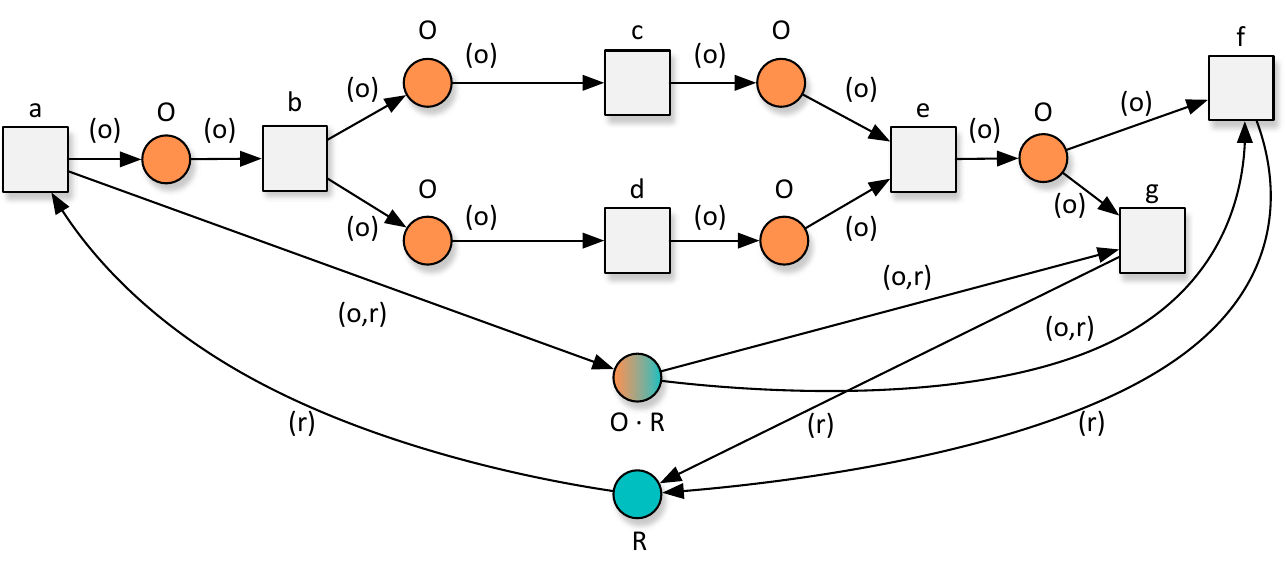}
        \caption{Another $O$-resource closure of the \tpnid (a) \protect\label{fig:sound-net-closure-2}}
	\end{subfigure}
\caption{An identifier-sound \tpnid and two alternative resource closures}\label{fig:resource-closures}
\end{figure}

\begin{example}
\label{ex:resource-closures}
Figure~\eqref{fig:sound-net} illustrates a \tpnid with a single object type $O$, manipulated using sequential and concurrent transitions, and finally consumed through one of two (mutually exclusive) consumers. Two resource closures of this net are shown in Figures~\eqref{fig:sound-net-closure-1} and~\eqref{fig:sound-net-closure-2}, using $R$ as resource type.

The closure in Figure~\eqref{fig:sound-net-closure-1} shows that for an object $o$ that is consumed by transition $f$, its assigned resource is returned to the resource place, becoming again available for a further assignment. If $o$ is instead consumed by transition $g$, the resource is also consumed (fetching it from the assignment place without returning it to the resource place). In addition, the two transitions $c$ and $d$, which are concurrent in the original \tpnid of Figure~\eqref{fig:sound-net}, are now declared to require the resource \emph{in isolation}, therefore implicitly requiring serialization (in whatever order).

The closure in Figure~\eqref{fig:sound-net-closure-2} depicts a slightly different scenario. On the one hand, both consumers now return the assigned resource upon consuming an object. On the other hand, no internal transition of the original \tpnid are linked to the assignment place, hence transitions $c$ and $d$ continue to be truly concurrent even after the application of the resource closure.
\end{example}

Examples~\ref{ex:simple-resource-closure} and~\ref{ex:resource-closures} show different examples of well-behaved resource closures, which indeed technically substantiate the informal modelling principles listed above and, even more, actually satisfy the property of conservative resource management, as per Definition~\ref{def:conservative-resource-management}. A natural question is whether this holds when resource closure is applied to an arbitrary input \tpnid. It is easy to show on even very minimalistic examples that provide a negative answer to this question.

\begin{figure}[!h]
\vspace*{-2mm}
\centering
	\begin{subfigure}{.49\textwidth}
		\centering
		\includegraphics[width=0.97\textwidth]{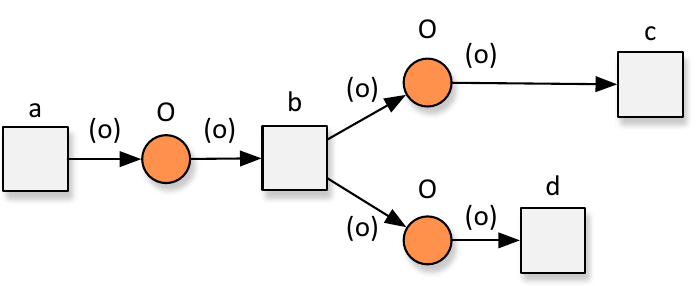}
        \vspace*{2.3cm}
        \caption{An unsound \tpnid with an object type $O \in \obj{\Lambda}$}\label{fig:unsound-net}
	\end{subfigure}
	\begin{subfigure}{.49\textwidth}
		\centering
		\includegraphics[width=0.97\textwidth]{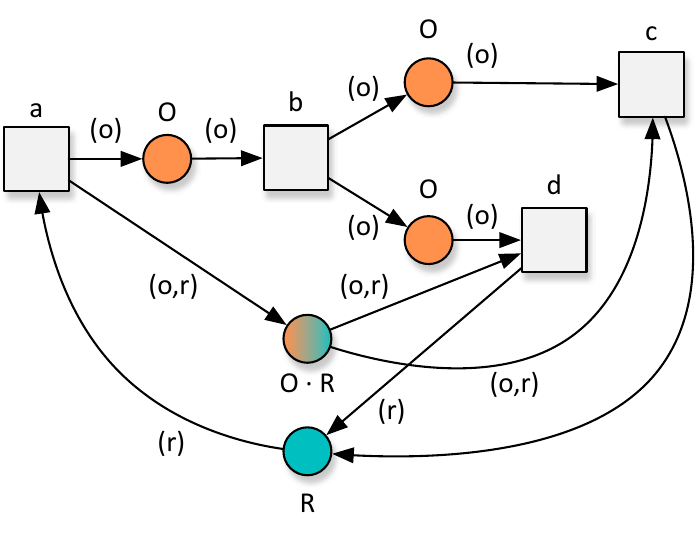}
        \caption{An $O$-resource closure of the \tpnid (a)}\label{fig:unsound-net-closure}
	\end{subfigure}\vspace*{-1mm}
\caption{An identifier-unsound \tpnid and one of its resource closures}\label{fig:resource-unwanted-closure}\vspace*{-3mm}
\end{figure}

\begin{example}
\label{ex:resource-unwanted-closure}
Figure~\ref{fig:resource-unwanted-closure} illustrates a \tpnid and one of its resource closures. One can immediately see that the resource closure does not behave as expected. Upon creation of a new object of type $O$, say, $o_1$, a resource, say, $r_5$ of type $R$ is taken from the resource place and assigned to $o_1$, keeping track of the assignment $(o_1,r_5)$ in the assignment place. Object $o_1$ then flows through the \tpnid, leading to the generation of two tokens carrying $o_1$, concurrently enabling the consumer transitions $c$ and $d$. Upon the consumption of the first of such two tokens, the assignment token $(o_1,r_5)$ is removed from the assignment place and used to return $r_5$ to the resource place. This means that the second token carrying $o_1$ stays forever stuck, and cannot be consumed, as there is no assignment token matching $o_1$ that can be used to fire the consumer transition.
\end{example}

By more closely inspecting the negative example of resource closures discussed in Example~\ref{ex:resource-unwanted-closure}, one can notice that the modelling glitch is not natively caused by the resource closure itself, but actually originates from the fact that the input \tpnid of Figure~\eqref{fig:unsound-net} is not identifier-sound. In particular, considering this \tpnid, initially marked by the empty marking (or, equivalently, the \tpnid of Figure~\eqref{fig:unsound-net-closure}, initially marked by a marking that inserts some resources in the resource place), the violated property is the one of proper completion: once one of the two concurrent consumer transitions $c$ and $d$ is fired to consume a previously created object, the identifier of the object still persists in the token enabling the other consumer transition. The resource closure actually inherits the lack of proper completion, but also suffers of lack of weak termination.

This leads to the following, natural follow-up question: how does identifier-soundness of an arbitrary input \tpnid impact on the properties of the \tpnid{s} resulting from the application of resource closure? We answer by showing three key properties:\medskip

\begin{compactenum}
\itemsep=1.2pt
\item \emph{Resource closure guarantees conservative resource management} - the application of resource closure to a \tpnid leads to a \tpnid that indeed satisfies the property of conservative resource management.
\item \emph{Resource closure preserves soundness} - Every resource closure of an identifier-sound \tpnid is an identifier-sound \tpnid.
\item \emph{Resource closure of a sound net induces boundedness} - Every resource closure of an identifier-sound \tpnid is a bounded \tpnid.
\end{compactenum}\medskip

We start with conservative resource management.

\begin{theorem}
\label{thm:conservative-closure}
Every marked resource closure of a marked \tpnid manages resources conservatively, in the sense of Definition~\ref{def:conservative-resource-management}.
\end{theorem}

\begin{proof}
Consider a marked \tpnid $(N,m_0)$, and a marked resource closure $(N',m'_0)$ of it. Fix an object type $\lambda$. First, notice that the resource type associated to $\lambda$ does not have any emitter transition. This proves resource preservation.

\medskip
We then consider resource exclusive assignment.
By definition, in $m'_0$, every active $\lambda$-resource $r$ is referenced by a single token carrying the unary tuple $(r)$ in the resource place for $\lambda$. The content of this place is left unaltered until a $\lambda$-emitter transition fires. Consider now the firing a $\lambda$-emitter generating a new $\lambda$ object, say $o$. As stated in Definition~\ref{def:closure} (items 3a, 3c, and 4b), this can only occur if there is a resource, say, $r$, in the resource place, and firing leads to consume the only token carrying $r$ therein, while producing a single token carrying $(o,r)$ in the assignment place for $\lambda$. Internal $\lambda$-transitions executed for $o$ may only access such a token in a read-only mode, thus leaving the content of the assignment place unaltered, as stated in Definition~\ref{def:closure} (items 3e, and 4d). The only way of removing such a token $(o,r)$ is by firing a $\lambda$-collector transition, which, according to Definition~\ref{def:closure} (items 3b, 3d,  and 4c), consumes $(o,r)$ and can possibly produce a token carrying $r$, inserted in the resource place for $\lambda$. All in all, for every resource $r$ associated to the resource place for $\lambda$ in $\marking'_0$, and for every reachable marking $\marking \in \reachable{N'}{\marking'_0}$, we have that there exists at most one token in $\marking$  either carrying $(r)$ (and contained in the resource place for $\lambda$) or $(o,r)$ for some object $o$ of type $\lambda$ (and contained in the assignment place for $\lambda$). This proves that $(N',m'_0)$ satisfies the property of resource exclusive assignment.
\end{proof}

We continue with soundness preservation. The crux here is that the enrichment with a \tpnid with resources through resource closure does not alter the evolution of emitted objects, but only constrains \emph{when} new objects can be created.

\begin{theorem}
\label{thm:soundness-preservation}
Let $(N,m_0)$ be a marked \tpnid, and $(N',m'_0)$ one of its marked resource closures. If $(N,m_0)$ is identifier-sound, then $(N',m'_0)$ is identifier-sound as well.
\end{theorem}

\begin{proof}
By definition of resource closure, the reachability graph $\reachable{N'}{m'_0}$, projected on the original places of $N$, is a subset of $\reachable{N}{m_0}$. In fact, the only effect of resource closure on the original marked net is to prevent the possibility of creating new objects if the resource place attached to the corresponding emitter is empty. Since proper termination is a universal property over markings, it is preserved by subsets of $\reachable{N}{m_0}$ and, noticing that resource and assignment places do not affect the status of proper completion, hence also by $\reachable{N'}{m'_0}$. This proves that $(N',m'_0)$ properly completes.

\medskip
As for weak termination of $\reachable{N'}{m'_0}$, assume by absurdum that one object $o$ of type $\lambda$ cannot progress to consumption. Since the original \tpnid is weakly terminating, this can only happen due to the absence of (the only) assignment token referencing $o$ and being present in the assignment place associated to $\lambda$. Such a token was by construction generated upon the creation of $o$, and consequently its absence can be only due to a previous firing of some collector transition that consumed a token referencing $o$. This would however mean that $\reachable{N'}{m'_0}$ is not properly completing, which contradicts what was proven before.
\end{proof}

Finally, we turn to boundedness of resource closures of sound \tpnid{s}. Intuitively, this is due to the good interaction between resources and objects:\vspace*{1.6mm}

\begin{compactitem}
\item there are boundedly many resources available;
\item new objects can be created as long as there are still available resources;
\item when no resource is available, a new object can be created only upon destruction of a currently existing one.\vspace*{1.6mm}
\end{compactitem}

\noindent This reconstructs, in the more sophisticated case of \tpnid{s}, the boundedness result for resource and instance-aware workflow nets, at the core of \cite{MontaliR16}.

\begin{theorem}
\label{thm:closure-boundedness}
Let $(N,m_0)$ be a marked \tpnid, and $(N',m'_0)$ one of its marked resource closures. If $(N,m_0)$ is identifier-sound, then $(N',m'_0)$ is bounded.
\end{theorem}

\begin{proof}
Boundedness immediately holds for resource identifiers in $\reachable{N'}{m'_0}$, thanks to Theorem~\ref{thm:conservative-closure} and the fact that $(N',m'_0)$ manages resources conservatively, and to the fact that, by construction of resource closure, every resource originally present in $m'_0$ is either carried by one token present in its corresponding resource place, or by one token present in its corresponding assignment place.

\medskip
We then prove the theorem by showing that $(N',m'_0)$ is width- and depth-bounded when considering object identifiers.

Depth boundedness is immediately obtained by recalling that, by Theorem~\ref{thm:soundness-preservation}, identifier-soundness of $(N,m_0)$ implies identifier-soundness of $(N',m'_0)$, which is consequently depth-bounded by Lem\-ma~\ref{lemma:depth}.

\medskip
Width-boundedness instead follows from the observation, already used in the proof of Theorem~\ref{thm:conservative-closure}, that for every marking $m \in \reachable{N'}{m'_0}$, if $o \in \id{m}$, then there is exactly one token referencing $o$ in the assignment place associated to the type of $o$. Hence the maximum number of object identifiers is bounded by the number of resources present in $m'_0$.
\end{proof}

By combining Theorems~\ref{thm:conservative-closure} and~\ref{thm:model-checking}, we thus obtain the following conclusive result, that relates to the verification results discussed in Section~\ref{sec:verification}.

\begin{corollary}
\label{cor:model-checking-closure}
Model checking of $\mu\L{\text{-}FO}^{\text{PNID}}$ and $LTL\text{-}FO_p^{\text{PNID}}$ is decidable for marked resource closures of identifier-sound \tpnids.
\end{corollary}

\subsection{Discussion}
We now briefly comment on the interestingness and generality of our approach to resource management. First and foremost, the conservative management of resources is a quite general notion, which is for example naturally guaranteed in processes operating over material objects (called \emph{material handling systems} in \cite{FaDV21}). 
In every snapshot of such system, each object can either be under the responsibility of one and only one resource (e.g., a baggage being inspected by an operator), or moving from one resource to another (e.g., a baggage waiting for inspection in a queue). In this light, resource closure captures the prototypical case where an object is associated to a single resource alongside its entire lifecycle. While taking verbatim this notion is hence unnecessarily restrictive, it allowed us to highlight the essential features needed to suitably control the way a \tpnid interacts with resources:\vspace*{1.8mm}

\begin{compactenum}
\item resources should be managed conservatively;
\item at every point in time, resources should control the number of simultaneously active objects of each type (so that a new object can be created only if there is an available resource to assign to it).
\end{compactenum}\smallskip

\noindent The latter requirement can be achieved by more fine-grained notions of resource closures, applied to sub-nets operating over the same object type. For example, if an object type is manipulated via two subnets that are composed sequentially, concurrently, or in mutual exclusion, one may apply resource closure over each subnet (possibly using distinct resource types), while retaining all the good properties introduced here. This calls for a follow-up investigation: infusing resource closure within the constitutive blocks of typed Jackson nets, ensuring that each block operates over a dedicated resource type satisfying the two features 1. and 2.~recalled above.

\smallskip
Last but not least, such two features yield the key property that the resulting net is bounded, as stated in Theorem~\ref{thm:closure-boundedness}. This should by no means be interpreted as the fact that the \tpnid overall operates over boundedly many objects: in fact, unboundedly many objects can be created and handled, provided that they are not all simultaneously active within the system, but distribute over time depending on the amount of available resources.

\vspace{-1mm}
\section{Related work}\label{sec:relatedwork}
This work belongs to the line of research that aims at augmenting pure control-flow description of processes with data, and study formal properties of the resulting, integrated models.
When doing so, it becomes natural to move from case-centric process models whose analysis focuses on the evolution of a single instance in isolation, to so-called \emph{object-centric process models} where multiple related instances of the same or different processes co-evolve. This is relevant for process modeling, analysis, and mining~\cite{Aalst19}.

Different approaches to capture the control-flow backbone of object-oriented processes have been studied in literature, including declarative \cite{ArtaleKMA19} and database-centric models \cite{MonC16}. In this work, we follow the Petri net tradition, which comes with three different strategies to tackle object-centric processes.

A first strategy is to represent objects implicitly. The most prominent example in this vein is constituted by proclets \cite{Fahland19}. Here, each object type comes with a Petri net specifying its life cycle. Special ports, annotated with multiplicity constraints, are used to express generation and synchronization points in the process, operating over tokens that are implicitly related to co-referring objects.
Correctness analysis of proclets is an open research topic.

\medskip
A second strategy is to represent objects explicitly. Models adopting this strategy are typically extensions of $\nu$-PNs \cite{RVFE11}, building on their ability to generate (fresh) object identifiers and express guarded transitions relating multiple objects at once.
While $\nu$-PNs attach a single object to each token,
Petri nets with identifiers (PNIDs)~\cite{HSVW09} use vectors of identifiers on tokens, representing database transactions.
Representing and evolving relationships between objects call for extending this to tuples of objects, in the style of~\cite{HSVW09}. For such Petri nets with identifiers (PNIDs),~\cite{HSVW09} provides patterns capturing different types of database transactions.
The ISML approach \cite{PolyvyanyyWOB19} equips Petri nets with identifiers (PNIDs)~\cite{HSVW09} with the ability of manipulating populations of objects defining the extensional level of an ORM data model.
Transitions can be executed if they do not lead to violating the constraints captured in the data model.
For such models, correctness properties are assessed by imposing that the overall set of object identifiers is finite, and fixed a-priori. This ensures that the overall state space is indeed finite, and can be analyzed using conventional methods.
Catalog-nets~\cite{GhilardiGMR20} extend PNIDs with the ability of querying a read-only database containing background information.
Correctness properties are checked parametrically to the content of read-only database.
Decidability and other meta-properties, as well as actual algorithms for verification based on SMT model-checking, are given for safety properties, whereas (data-aware) soundness can only be assessed for state-bounded systems 
\cite{BCMD14,MonC16}.

\medskip
The third, final strategy for modeling object-centric processes with Petri nets is to rely on models that highlight how multiple objects of different types may flow through shared transitions, without considering object identifier values. This approach is followed in \cite{AalstB20}, where object-centric nets are extracted from event logs, where logged events might come with sets of object identifiers. Soundness for this model is studied in \cite{LoMR21}, where the authors propose to check related correctness criteria for object types, without considering concrete object identifiers, or for single objects in isolation that are still allowed to interact with the system environment along their life-cycles.
Similarly to the approach studied in this paper, the authors in~\cite{LoMR21} assume that the system model can have any number of objects being simultaneously active.

\medskip
The approach studied in this paper focuses on the essence of Petri net-based object-centric processes adopting the explicit approach, that is, grounded on PNIDs. We provide, for the first time, a notion of \emph{identifier soundness} that conceptually captures the intended evolution of objects within a net, show that such a property is undecidable to check in general, and provide a pattern-based construction technique that guarantees to produce identifier-sound models.
Other works that propose more high-level extensions of classical WF-nets as well as the related notion(s) of soundness.
In~\cite{LeoniFM18,FelliLM19,LeoniFM20}, the authors investigated data-aware soundness for data Petri nets, in which a net is extended with guards manipulating a finite set of variables associated with the net.
That type of soundness was shown to be decidable.
In~\cite{MontaliR16}, the authors proposed both a workflow variant of $\nu$-Petri nets and its resource-aware extension. The authors also defined a suitable notion of soundness for such nets and demonstrated that it is decidable by reducing the soundness checking problem to a verification task over another first order logic-based formalism.
\cite{HaarmannW20,CaseModels2020} considered the soundness property for BPMN process models with data objects that can be related to multiple cases.
The approach consists of several transformation steps: from BPMN to a colored Petri net and then to a resource-constrained workflow net. The authors then check k-soundness against the latter.

\medskip
Resource-constrained workflow nets pose different requirements on soundness. In \cite{BBS07} the authors studied a specific class of WFR-nets for which soundness was shown to be decidable.
In \cite{HeeSSV05,SidorovaS13} a more general class of Resource-Constrained Workflow Nets (RCWF-nets) was defined.
The constraints are imposed on resources and require that all resources that are initially available are present again after all cases terminate, and that for any reachable marking, the number of available resources does not override the number of initially available resources.
In \cite{HeeSSV05} it was proven that for RCWF-nets with a single resource type generalized soundness can be effectively checked in polynomial time.
Decidability of generalized soundness for RCWF-nets with an arbitrary number of resource places was shown in~\cite{SidorovaS13}.

\section{Conclusions}
\label{sec:conclusions}

Achieving harmony in models that describe how processes data objects manipulate is challenging. In this paper, we use typed Petri nets with Identifiers (\tpnids) to model these complex interactions of multiple objects, referred through their identifiers. We propose identifier soundness as a correctness criterion that conceptually captures the expected evolution of each object.
Identifier soundness consists of two conditions: weak termination, i.e., that any identifier that is created is eventually removed, and proper type completion, i.e., when a collecting transition fires for an identifier of a type, the type should be removed from the resulting marking.
Identifier soundness is in general undecidable for \tpnids.
For two subclasses we show that identifier soundness is guaranteed, and that the overall model remains live.
On top of that, we propose a resource-aware extension of \tpnids,
in which all object manipulations are systematically guarded by a finite number of typed resources.
For this class of \tpnids we propose a correctness criterion similar to identifier soundness, which also takes into account conditions for correct resource management. The resource soundness is deemed to be undecidable as well.

\medskip
Many systems allow for a dynamic number of simultaneously active objects.
In theory, this number can be infinite, and thus such models become  width-unbounded.
However, for many systems there is a natural upper bound, which can be either assumed or guaranteed with different modeling techniques (such as multiplicity upper bounds on objects \cite{MonC16} or resources \cite{MonR16,SidorovaS13}.  This gives potential for different directions on new analysis techniques.
As an example, some models may have a minimum bound such that its correct behavior is guaranteed above this bound, in a similar way as 1-soundness of WF-nets guarantees correctness of its EC-closure.
One can extend \tpnids by enriching objects with attributes over different datatypes, and transitions with the ability to query such attributes and express conditions and updates over them, using their datatype-specific predicates. Of particular interest are comparisons and arithmetic operations for numerical datatypes. This calls for combining the techniques studied in this paper with data abstraction techniques used to deal with numerical datatypes, possibly equipped with arithmetics \cite{FelliLM19,FeMW22}.

\medskip
We plan to provide tool support for the designer of such systems.
Although many correctness criteria are undecidable, designers should be left in the dark.
Since the ISM-suite~\cite{WerfP20} already allows to model \tpnids, 
we intend to work on extending it with verification techniques to support the modeler in designing systems where processes and data are in resonance.

\end{document}